\documentclass[journal]{IEEEtran}
\usepackage[cmex10]{amsmath}
\usepackage{eqparbox}
\usepackage{mathtools,amssymb,bm}
\usepackage{amstext}
\usepackage{amssymb}
\usepackage{graphicx}
\usepackage{amsfonts}
\usepackage{dsfont}
\usepackage{array}
\usepackage{algorithmicx}
\usepackage[ruled]{algorithm}
\usepackage{algpseudocode}
\usepackage{algpascal}
\usepackage{algc}
\usepackage{cases}
\usepackage{pgfplots}
\usepackage{accents}
\usepackage{caption}
\usepackage{amsthm}
\DeclareCaptionLabelSeparator{periodspace}{.\quad}
\usepackage{float}
\usepackage{cite}
\usepackage [autostyle, english = american]{csquotes}
\usepackage{hyperref}
\theoremstyle{remark}

\newcommand\ASTART{\bigskip\noindent\begin{minipage}[b]{0.5\linewidth}}

\newcommand\AENDSKIP{\end{minipage}\bigskip}
\newcommand\AEND{\end{minipage}}
\ifCLASSOPTIONcaptionsoff
  \usepackage[nomarkers]{endfloat}
 \let\MYoriglatexcaption\caption
 \renewcommand{\caption}[2][\relax]{\MYoriglatexcaption[#2]{#2}}
\fi
\newtheorem{thm}{Theorem}[section]
\newtheorem{lem}[thm]{Lemma}
\newtheorem{prop}[thm]{Proposition}

\newtheorem{defn}{Definition}[section]

\newcommand*{\rom}[1]{\expandafter\@slowromancap\romannumeral #1@}

\begin{document}
%
%\onecolumn
\title{How to exploit prior information in low-complexity models}
\author{Sajad~Daei, Farzan~Haddadi}%

\maketitle

\begin{abstract}
Compressed Sensing refers to extracting a low-dimensional structured signal of interest from its incomplete random linear observations. A line of recent work has studied that, with the extra prior information about the signal, one can recover the signal with much fewer observations. For this purpose, the general approach is to solve weighted convex function minimization problem. In such settings, the convex function is chosen to promote the low-dimensional structure and the optimal weights are so chosen to reduce the number of measurements required for the optimization problem. In this paper, we consider a generalized non-uniform model in which the structured signal falls into some partitions, with entries of each partition having a definite probability to be an element of the structure support. Given these probabilities and regarding the recent developments in conic integral geometry, we provide a method to choose the unique optimal weights for any general low-dimensional signal model. This class of low-dimensional signal model includes many popular examples such as $\ell_1$ analysis (entry-wise sparsity in an arbitrary redundant dictionary), $\ell_{1,2}$ norm (block sparsity) and total variation semi-norm (for piece-wise constant signals). We show through precise analysis and simulations that the weighted convex optimization problem significantly improves the regular convex optimization problem as we choose the unique optimal weights.
\end{abstract}


% \fi
\IEEEpeerreviewmaketitle

\section{Introduction}
 \IEEEPARstart{O}{ver} the past decade, it has become an established fact that based on prior information about the sparsity feature of the signal\footnote{Signals with this feature have much less nonzero elements than the ambient dimension.} $\bm{x}\in \mathbb{R}^n$, one can recover it from a few random linear measurements known as compressed sensing (CS)\cite{donoho2006most}. For this purpose, it is necessary to introduce a function that promotes sparsity. This leads to the following optimization problem (known as $\ell_0$ problem or $\mathsf{P}_0$):\\
  \begin{align}\label{eq.l0}
   \mathsf{P}_0:~~~~~~&\min_{\bm{z}\in\mathbb{R}^n}~\|\bm{z}\|_{0}\nonumber\\
   &\mathrm{s.t.} ~\bm{Az} =\bm{b},
\end{align}
where in (\ref{eq.l0}), $\bm{A}\in \mathbb{R}^{m\times n}$ and $\bm{b}\in \mathbb{R}^m$ are the measurement matrix and measurement vector, respectively\footnote{We typically consider the case that $s$ is much smaller than $n$. }. The reason behind this optimization problem is that among all possible solutions of the equation $\bm{Az} =\bm{b}$, the sparsest one is unique \cite{donoho2006most}.\par
The optimization problem $\mathsf{P}_0$ is non-convex, computationally intractable, and in fact  NP-hard. Although there has not ever been any algorithm that can solve $\mathsf{P}_0$ in polynomial time, but if it existed, it would require $\mathcal{O}(s)$ measurements in the case that the original signal $\bm{x}\in \mathbb{R}^n$ is $s$-sparse\footnote{$\|\bm{x}\|_0\le s$}.\par
Candes and Tao in \cite{tao2005decoding} and Donoho in \cite{donoho2006most} proved that under the so-called Restricted Isometry Property (RIP) on the measurement matrix $\bm{A}\in \mathbb{R}^{m\times N}$ , instead of $\mathsf{P}_0$, one can alternatively solve the following optimization problem known as $\ell_1$ minimization problem or $\mathsf{P}_1$ and recover the original $s$-sparse signal from $\mathcal{O}(s\log(\frac{n}{s}))$ random linear measurements  $\bm{y}\in\mathbb{R}^m$. By $\mathsf{P}_1$, we pay for a tractable solution to $\bm{Az}=\bm{b}$  with suboptimal number of measurements.
\begin{align}\label{eq.l1}
   \mathsf{P}_1:~~~~~~&\min_{\bm{z}\in\mathbb{R}^n}~\|\bm{z}\|_{1}=\sum_{i=1}^{n}|z(i)|\nonumber\\
   &\mathrm{s.t.} ~\bm{Az} =\bm{b}.
\end{align}
Ever since, most of the literature on compressed sensing has focused on the case where the only side information is that the signal of interest is sparse. In many applications however, there may exist some other extra information about the underlying signal in addition to sparsity. It's plausible that we can do even better than the case that the only information is sparsity. The common way to exploit both sparsity and extra information without entering the non-convex optimization world is weighted convex minimization defined as follows:
\begin{align}\label{eq.wl1}
   \mathsf{P}_{1,\bm{w}}:~~~~~~&\min_{\bm{z}\in\mathbb{R}^n}~\|\bm{z}\|_{1,\bm{w}}=\sum_{i=1}^{n}w_i|z(i)|\nonumber\\
   &\mathrm{s.t.} ~\bm{Az} =\bm{b},
\end{align}
where in (\ref{eq.wl1}), $\bm{w}\in \mathbb{R}_{+}^n$ refers to the positive weights applied to the components of the structured signal. Note that, a strictly structured sparse signal is characterized by two parameters:(1) the set of non-zero locations (support) and (2) the signal values. If one has the signal support, the undetermined equation $\bm{Az}=\bm{b}$ reduces to the solvable overdetermined problem. However, if one has an estimated support which is only fractionally wrong it's better to associate larger weights to zero entries and smaller weights to the non-zeros. This intuition was first analysed by Candes et. al. in \cite{candes2008enhancing}.
A major question that has been the subject of recent researches (\cite{friedlander2012recovering,khajehnejad2011analyzing,stojnic2009various,mansour2015recovery,vaswani2010modified}) is how to suitably choose the weights in $\mathsf{P}_{1,\bm{w}}$. In other words, given prior information, how we can compute optimal weights.
We mean optimal\footnote{Although this meaning is not universal in the CS literature.} weights in the sense that the required number of measurements for $\mathsf{P}_{1,\bm{w}}$ to succeed with overwhelming probability, is minimized.
\subsection{Related Works}
Von Borries et al. in \cite{von2007compressed}, showed that the number of measurements needed to recover a sparse signal in discrete fourier transform (DFT) by weighted-$\ell_1$ analysis problem can be decreased when there is some information on the support of the sparse signal in the dictionary. They showed that the amount of this reduction is exactly the prior knowledge of frequencies in the support.\par Vaswani et al. in \cite{vaswani2010modified}, studied modified CS where a part of the signal support is known. They used weighted $\ell_1$ minimization to exploit this prior information and assigned zero weights to the known part. They obtained sufficient conditions for exact recovery and developed the so-called RegModCS algorithm that exploits prior knowledge of the signal. Jacques in \cite{jacques2010short}, generalized the results of Vaswani et al. to the case of corrupted and noisy measurements.\par
Khajehnejad et al. in \cite{khajehnejad2011analyzing}, investigated the non-uniform sparse model where the entries of the underlying sparse signal fall into a fixed number of partitions. Each partition has a different probability of being non-zero. Given these probabilities and using a Grassmann angle approach, they computed optimal weights in the case that the underlying signal divided into two subclasses. They used weighted $\ell_1$ minimization to recover the sparse signal. Their approach does not contain explicit strategies for more than two partitions.\par
Based on \cite{stojnic2009various}, and \cite{chandrasekaran2012convex}, Krishnaswamy et al. in \cite{krishnaswamy2012simpler}, proposed a simpler and heuristic method to estimate the optimal weights. Unlike \cite{khajehnejad2011analyzing}, their work can be extended to more than two sparsity partitions.\par
Oymak et al. in a different analysis in \cite{oymak2012recovery}, considered the sparse signal with two partitions (non-uniform sparse model). To exploit prior knowledge, they used weighted $\ell_1$ minimization with each partition being assigned a weight. Using the "escape through a mesh" lemma \cite{gordon1988milman}, they derived an upper bound for the minimum number of Gaussian measurements required for $\mathsf{P}_{1,\bm{w}}$ to succeed. In addition, they computed optimal weights by minimizing this upper bound with respect to the weights. Although they minimize the upper bound of the minimum number of measurements but they showed that this upper bound is tight asymptotically (as the ambient dimension goes to infinity) by relating the bound to the regularized normalized Minimum Mean Squared Error (MMSE) of a certain basis pursuit denoising problem (BPDN). Further, in \cite{bayati2012lasso}, it has been shown that asymptotic normalized MMSE of BPDN equals to the asymptotic phase transition of Approximate Message Passing (AMP) algorithms.\par
Recently, Flinth in \cite{flinth2015optimal}, generalized the results of Oymak et al. to the non-uniform setting with arbitrary number of  partitions. He associated a weight with each partition and used the recent results of Amelunxen et al. in \cite{amelunxen2013living}. Amelunxen et al. had proved that the number of Gaussian measurements one needs to solve a convex program, is lower bounded by statistical dimension of a certain convex cone $\delta(C)$. They also calculated  upper and lower bounds for $\delta(C)$. Flinth in \cite{flinth2015optimal} calculated the optimal weights by simultaneously minimizing the lower and upper bounds of $\delta(C)$  with respect to the weights.
\subsection{Applications}
There are many important applications where in addition to inherent structures (entrywise sparsity in a dictionary, block sparsity and gradient sparsity), there exist some prior information about the signal of interest. In practice, the extraction of this prior information is realistic. As an example, one can investigate the statistic of used training data. Examples of applications with such information are listed below:
\begin{itemize}
\item \textbf{Entrywise sparsity}: In many digital signal processing applications, the desired discrete time signal is shown to be sparse in a specific dictionary. In \cite{vaswani2010modified}, the recovery of a sparse signal is considered in which the support estimate in previous support (related to previous time instant) can improve the recovery performance by reducing the number of measurements. Natural images are often sparse in the discrete wavelet transform (DWT) as the dictionary. It is possible for wavelet image coders to incorporate some prior knowledge about the locations of large wavelet coefficients \cite{baraniuk2010model}. In fact, we have non-uniformity in the DWT and the largest coefficients are highly concentrated around zero. Also in many other applications, non-uniformity is present along sparsity in the frequency domain. One can exploit this feature and reduce number of measurements required for recovery using prior knowledge about the locations of the most dense parts in the frequency dictionary \cite{von2007compressed}.
\item \textbf{Block sparsity}: In block sparsity, the non-zero locations of the sparse signal appears in fixed blocks. There exist some applications with non-uniformity in the set of signal blocks. In DNA microarray \cite{stojnic2008reconstruction}, we have prior information that some blocks are most probable to include the non-zero elements. In computational neuroscience problems \cite{computationalnero}, the behavior of neurons exhibit non-uniform clustered responses. In \cite{eldar2009robust}, it is shown that non-uniform sampling problems over union of subspaces can be considered as block sparsity with non-uniform prior information on the blocks. In such cases, the general approach is to solve weighted $\ell_{1,2}$ optimization problem ( known as $\mathsf{P}_{1,2,\bm{w}}$) and assign larger weights to the blocks that are less likely to have non-zero entries.
\item \textbf{Gradient sparsity}: Here, the signal of interest is sparse in the gradient dictionary. Conventional total variation (TV) semi norm assigns the same cost to smooth and non-smooth regions. However, natural images usually have non-uniform smoothness. Prior information in TV has long been investigated where we have some uncertainty about the smoothness of the 1- or 2-dimensional signals. In \cite{nonlocaltv}, non-uniform weight penalization is considered. In \cite{duran2016collaborative}, weighted total variation optimization problem (known as $\mathsf{P}_{\mathrm{TV},\bm{w}}$) is used to consider non-uniform smoothness of the signal.
\end{itemize}
\subsection{Contributions}
Unlike the prior works that only consider prior information on the signal entries in the identity dictionary, in Section \ref{section.model}, we define a generalized framework that consider prior information on blocks, signal variations, or on the coefficients of the signal in a redundant dictionary, and extend the results of \cite{flinth2015optimal}, and \cite{oymak2012recovery}, to general signal models. Structures we cover in this paper, include sparse vectors in a specific possibly redundant dictionary, block sparse vectors, and piecewise constant signals. $\|\bm{\Omega}.\|_0$ (number of non-zero entries in the dictionary $\bm{\Omega}$), $\|.\|_{0,2}$ (number of blocks with non-zero $\ell_2$ norm), and $\|.\|_{\mathrm{NV}}$ (number of variations) promote the mentioned structures, respectively. The convex relaxed form of the above functions is $\|\bm{\Omega}.\|_1$, $\|.\|_{1,2}$ and $\|.\|_{\mathrm{TV}}$, respectively. Convex functions that exploit both the inherent structure and prior knowledge are $\|\bm{\Omega}.\|_{1,\bm{w}}$, $\|.\|_{1,2,\bm{w}}$ and $\|.\|_{\mathrm{TV},\bm{w}}$, respectively. Related convex optimization problems are called $\mathsf{P}_{1,\bm{\Omega},\bm{w}}$, $\mathsf{P}_{1,2,\bm{w}}$, and $\mathsf{P}_{\mathrm{TV},\bm{w}}$, respectively, which are introduced in Section \ref{section.model}. The prior information in this paper, is the knowledge of partitions and probabilities. This setup exists in applications that multiple estimates of the support are available with different levels of confidence. In other words, these support estimates have different levels of accuracy. In this paper, we divide the structured signal\footnote{The structured signal may be block sparse, sparse in a dictionary, or smooth} into some partitions and assign a weight to each partition. We compute optimal weights following the same strategy as in \cite{flinth2015optimal} and \cite{oymak2012recovery} regarding the given probability of partitions intersected by the support. However, our definition of support and partitions differ from those works due to their different structures.
 More precisely, with each inherent structure, we define partitions and support as follows:
 \begin{enumerate}
   \item \textbf{Entrywise sparsity}: The support is defined as the non-zero locations in the possibly redundant sparsity dictionary. In this case, we have the prior information that some fixed sets (partitions) in the dictionary, are intersected by the support with known probabilities. Throughout the paper, we refer to this kind of support as entrywise support denoted by $\mathcal{S}$.
   \item \textbf{Block sparsity}: The support is defined as the blocks with non-zero $\ell_2$ norm. In this case, we have some sets of blocks with known probability that each set contributes to the support. Later on, we refer to this kind of support as block support $\mathcal{B}$.
   \item \textbf{Gradient sparsity}: In this case, the support is defined as locations of signal variations. We refer to this kind of support as gradient support $\mathcal{S}_g$. The gradient support is divided into two sets: consecutive and individual support which refer to consecutive and individual variations, respectively. Some predefined sets in the gradient domain are intersected by the subsequent and individual supports with known probability.
 \end{enumerate}
Our main contributions are as follows:
\begin{enumerate}
\item Given the probability of partitions to be in the support, unique optimal weights will be computed in the three models mentioned above\footnote{Optimality of weights in block and gradient sparsity is in a different sense from entrywise sparsity in redundant dictionary as is made precise later.}.
  \item With optimal weights in an arbitrary redundant dictionary $\bm{\Omega}$ in $\mathsf{P}_{1,\bm{\Omega},\bm{w}}$, the required number of measurements to succeed, exactly equals the total number of measurements one needs to recover each partitioned sparse vector in the dictionary separately by solving $\mathsf{P}_{1}$.
  \item With optimal weights on the blocks, the number of measurements required for successful recovery of a block sparse vector using $\mathsf{P}_{1,2,\bm{w}}$ with high probability, exactly equals the total number of measurements needed for recovery of each partition of blocks separately by solving $\mathsf{P}_{1,2}$.
  \item With optimal weights on the signal variations, the required number of measurements for $\mathsf{P}_{\mathrm{TV},\bm{w}}$ to succeed with high probability equals the whole number of measurements one needs to recover each partitioned smooth vector separately by solving $\mathsf{P}_{\mathrm{TV}}$.
   \item Explicit formulas are found for the number of measurements one needs to recover block sparse, smooth, and sparse vectors in a redundant dictionary.
\end{enumerate}
This work includes the results of \cite{flinth2015optimal} and \cite{oymak2012recovery} as a special case of block sparsity with unit size blocks and entrywise sparsity with identity dictionary.
%\IEEEpubidadjcol
\subsection{Outline of the paper}
The paper is organised as follows: Signal model and methodology are given in Section \ref{section.model}. In Section \ref{section.conicgeometry}, basic concepts of conic integral geometry are reviewed. In Sections \ref{section.l1analysis}, \ref{section.l12block}, and \ref{section.TV}, entry-wise sparsity, block and gradient sparsity, are investigated, respectively. Numerical simulations that support the theory are presented in ُSection \ref{section.simulation}. Finally, the paper is concluded in Section \ref{section.conclusion}.
\subsection{Notation}
Throughout the paper, scalars are denoted by lowercase letters, vectors by lowercase boldface letters, and matrices by uppercase boldface letters. The $i$th element of the vector $\bm{x}$ is given either by ${x}(i)$ or $x_i$. $(\cdot)^\dagger$ denotes pseudo inverse. We reserve calligraphic uppercase letters for sets (e.g. $\mathcal{S}$). The cardinality of a set $\mathcal{S}$ is denoted by $|\mathcal{S}|$. $[n]$ refers to $\{1,..., n\}$. Furthermore, we write ${\mathcal{\bar{S}}}$ for the complement $[n]\setminus\mathcal{
S}$ of a set $\mathcal{S}$ in $[n]$. For a matrix $\bm{X}\in\mathbb{R}^{m\times n}$ and a subset $\mathcal{S}\subseteq [n]$, the notation $\bm{X}_\mathcal{S}$ is used to indicate the column submatrix of $\bm{X}$ consisting of the columns indexed by $\mathcal{S}$. Similarly, for $\bm{x}\in\mathbb{R}^n$, $\bm{x}_\mathcal{S}$ is either the subvector in $\mathbb{R}^{|\mathcal{S}|}$ consisting of the entries
indexed by $\mathcal{S}$, that is, $(\bm{x}_S)_i = x_{j_i}~:~\mathcal{S}=\{j_i\}_{i=1}^{|\mathcal{S}|}$, or the vector in $\mathbb{R}^n$ which coincides with $\bm{x}$ on the entries in $\mathcal{S}$ and is zero on the entries outside $\mathcal{S}$. In this paper, $\bm{1}_{\mathcal{E}}$ denotes the indicator of the set $\mathcal{E}$. Nullspace and range of linear operators are denoted by $\mathrm{null}(\cdot)$, and $\mathrm{range}(\cdot)$, respectively. Hadamard product and Hadamard inverse on vectors are denoted by $\odot$ and $(\cdot)^{-1\odot}$ symbols, respectively. Given a vector $\bm{x}\in\mathbb{R}^n$ and a set $\mathcal{C}\subseteq \mathbb{R}^n$, denote the set obtained by scaling elements of $\mathcal{C}$ by elements of $\bm{x}$ by $\bm{x}\odot \mathcal{C}$. For a matrix $\bm{A}$, the operator norm is defined as $\|\bm{A}\|_{p\rightarrow q}=\underset{\|\bm{x}\|_p\le1}{\sup}\|\bm{Ax}\|_q$. For $\bm{x},\bm{y}\in\mathbb{R}^n$, $\bm{x}\le \bm{y}$ denotes component-wise inequality while $\bm{x}<\bm{y}$ denotes component-wise inequality with strict inequality in at least one component. $\mathds{B}_{\epsilon}^n$ refers to the $\epsilon$-ball $\mathds{B}_{\epsilon}^n=\{\bm{x}\in \mathbb{R}^n:~\|\bm{x}\|_2\le\epsilon\}$.
\section{Signal Model and Methodology}\label{section.model}
\subsection{Structured Signal Model}\label{subsec.signalmodel}
First, three models are defined as follow:
\begin{defn}{\textbf{(Non-uniform Sparse Signal in Dictionary $\bm{\Omega} \in \mathbb{R}^{p\times n}$):}} Consider a vector $\bm{x}\in\mathbb{R}^n$ that is sparse in $\bm{\Omega}$ with support $\mathcal{S}$ and relative sparsity $\sigma=\frac{\|\bm{\Omega x}\|_0}{p}$. Let $\{\mathcal{P}_i\}_{i=1}^L$ be a partition of $[p]$. Each set $\mathcal{P}_i$ is associated with accuracy $\alpha_i=\frac{|\mathcal{P}_i\cap \mathcal{S}|}{|\mathcal{P}_i|}$ and relative size $\rho_i=\frac{|\mathcal{P}_i|}{p}$. A vector $\bm{x}\in \mathbb{R}^{n}$ is called a non-uniform sparse vector in the dictionary $\bm{\Omega} \in \mathbb{R}^{p\times n}$, if its sparsity pattern is defused along the partition $\{\mathcal{P}_i\}_{i=1}^L$. Number of non-zero entries of each set $\mathcal{P}_i$ is $\alpha_i \rho_i p$ such that the total sparsity in $\bm{\Omega}$ is $s=\sigma p=\sum_{i=1}^{L}\alpha_i \rho_i p$.
\end{defn}
\begin{defn}{\textbf{(Non-uniform Block Sparse Signal in $\mathbb{R}^n$):}} Consider a block sparse vector in $\bm{x}\in\mathbb{R}^n$ with the block support $\mathcal{B}$ that is divided into $q$ fixed blocks $\{\mathcal{V}_b\}_{b=1}^q \subset [n]$ of integer length $k=n/q$. Let $\{\mathcal{P}_i\}_{i=1}^L$ be a partition of the blocks $[q]$. To each set of blocks, accuracy $\alpha_i=\frac{|\mathcal{P}_i\cap \mathcal{B}|}{|\mathcal{P}_i|}$, and relative size $\rho_i=\frac{|\mathcal{P}_i|}{q}$ are assigned. A vector $\bm{x}\in \mathbb{R}^{n}$ is called a non-uniform block sparse vector with relative block sparsity $\sigma=\frac{\|\bm{x}\|_{0,2}}{q}$ if its sparsity pattern is defused along the blocks $[q]$. Number of non-zero blocks of each set $\mathcal{P}_i$ is $\alpha_i \rho_i q$ while the total block sparsity is $s=\sigma q=\sum_{i=1}^{L}\alpha_i \rho_i q$.
\end{defn}
\begin{defn}{\textbf{(Non-uniform Smooth Signal in $\mathbb{R}^n$):}} Consider a smooth vector $\bm{x}\in\mathbb{R}^n$ with gradient support $\mathcal{S}_g$ and relative smoothness $\sigma=\frac{\|\bm{x}\|_{NV}}{n-1}$. Let $\{\mathcal{P}_i\}_{i=1}^L$ be a partition of the signal variations $[n-1]$. Each set is associated with accuracy $\frac{|\mathcal{P}_i\cap\mathcal{S}_g|}{n-1}$ and relative size $\rho_i=\frac{|\mathcal{P}_i|}{n-1}$. A vector $\bm{x}\in \mathbb{R}^{n}$ is called a non-uniform smooth vector if its variation pattern is defused along $[n-1]$. Number of non-zero entries of each set $\mathcal{P}_i$ is $|\mathcal{P}_i\cap\mathcal{S}_g|$ such that the total gradient sparsity is $s=\sigma(n-1)=\sum_{i=1}^{L}|\mathcal{P}_i\cap\mathcal{S}_g|$.
\end{defn}
%a weight $\omega_i \in \mathbb{R}_{+}$,
In this paper, three inherent structures are considered:
\begin{enumerate}
\item \textbf{Entrywise sparsity}: Each of $L$ partitions is assigned a fixed weight. The related optimization problem  becomes:\\
    \begin{align}\label{eq.wl1analysis}
   \mathsf{P}_{1,\bm{\Omega},\bm{w}}:~~~~~~&\min_{z\in\mathbb{R}^n}~\|\bm{\Omega z}\|_{1,\bm{w}}=\sum_{i=1}^{p}w_i|({\Omega z})_i|\nonumber\\
   &\mathrm{s.t.} ~\bm{Az} =\bm{b},
\end{align}
in which, $\bm{w}=\sum_{i=1}^{L}\omega_i \bm{1}_{\mathcal{P}_i}\in   \mathbb{R}_{+}^p$ and $(\bm{1}_{\mathcal{P}_i})_j = 1  :  \forall j \in \mathcal{P}_i$ and $0$ otherwise.
\item \textbf{Block sparsity}: Each of $L$ partitions is assigned a fixed weight. The related optimization problem will be:
    \begin{align}\label{eq.wl1,2}
   \mathsf{P}_{1,2,\bm{w}}:~~~~~~&\min_{z\in\mathbb{R}^n}~\|\bm{z}\|_{1,2,\bm{w}}=\sum_{b=1}^{q}w_b\|\bm{z}_{\mathcal{V}_b}\|_2\nonumber\\
   &\mathrm{s.t.} ~\bm{Az} =\bm{b},
\end{align}
where, $\bm{w}=\sum_{i=1}^{L}\omega_i 1_{\mathcal{P}_i}\in   \mathbb{R}_{+}^q$.
\item \textbf{Gradient sparsity}: Each of $L$ partitions is assigned a fixed weight. The related optimization problem is:
    \begin{align}\label{eq.wTV}
   \mathsf{P}_{\mathrm{TV},\bm{w}}:~~~~~~&\min_{\bm{z}\in\mathbb{R}^N}~\|\bm{z}\|_{\mathrm{TV},\bm{w}}=\sum_{i=1}^{n-1}w_i|z(i+1)-z(i)|\nonumber\\
   &\mathrm{s.t.} ~\bm{Az} =\bm{b},
\end{align}
and, $\bm{w}=\sum_{i=1}^{L}\omega_i 1_{\mathcal{P}_i}\in \mathbb{R}_{+}^{n-1}$
\end{enumerate}
\subsection{Measurement Matrix}
The common approach for sampling procedure is to draw a random measurement matrix $\bm{A}\in \mathbb{R}^{m\times n}$ consisting of i.i.d Gaussian entries. The authors in \cite{amelunxen2013living} show that the phase transition occurs when the measurement matrix $\bm{A}$ in $\mathsf{P}_f$ is standard normal with independent entries. However, the recent result of \cite{oymak2015universality} shows that in the asymptotic case one can use a larger class of random measurements that have symmetric distribution and bounded moments. In this paper, we choose the entries of $\bm{A}$ from i.i.d standard Gaussian distribution.
\subsection{Methodology}
For each structure defined in Section \ref{subsec.signalmodel} , the number of measurements is calculated as a function of the weights ($w=\sum_{i=1}^{L}\omega_i 1_{\mathcal{P}_i}$) and the optimal weights are found by minimizing the upper and lower bounds of statistical dimension. It's shown that these bounds are asymptotically tight. The optimal weights obtained in the problems $\mathsf{P}_{1,\bm{\Omega},\bm{w}}$, $\mathsf{P}_{1,2,\bm{w}}$, and $\mathsf{P}_{\mathrm{TV},\bm{w}}$ are unique up to an irrelevant positive scaling.
\section{Conic Geometry}\label{section.conicgeometry}
In this section, basic concepts of conic integral geometry to be used later are reviewed.
\subsection{Subdifferntial}
Subdifferntial of a proper\footnote{An everywhere defined function taking values in $(-\infty,\infty]$ with at least one finite value in $(-\infty,\infty)$.} convex function $f:\mathbb{R}^n\rightarrow \mathbb{R}\cup \{\pm\infty\}$ at $\bm{x}\in\mathbb{R}^n$ is given by:
\begin{align}\label{eq.subdiff}
\partial f(\bm{x}):=\{\bm{z}\in \mathbb{R}^n: f(\bm{y})\ge f(\bm{x})+\langle \bm{z},\bm{y}-\bm{x} \rangle~:~\forall \bm{y}\in \mathbb{R}^n\}.
\end{align}
\begin{prop}\label{prop.simplerform of subdiff}
Let $f: \mathbb{R}^n\rightarrow\mathbb{R}\cup \{\pm\infty\}$ be a proper convex function that is $1$-homogenous i.e. $f(\alpha \bm{z})=|\alpha|f(\bm{z})~:\forall \alpha \in \mathbb{R}$ and sub-additive i.e. $f(\bm{x}+\bm{y})\le f(\bm{x})+f(\bm{y})~~:\forall \bm{x},\bm{y} \in \mathbb{R}^n$, then we have a simpler form of subdifferntial given by:
\begin{align}\label{eq.subdiff2}
\partial f(\bm{x})=\{\bm{z}\in \mathbb{R}^n: \langle \bm{z},\bm{x} \rangle=f(\bm{x}) , f^*(\bm{z})=1\},
\end{align}
where, $f^*(\bm{z})=\underset{f(\bm{y})\le 1}{\sup}\langle \bm{z},\bm{y} \rangle$ is the dual function of $f(\bm{z})$.
\end{prop}
\begin{proof}
By taking $\bm{y}=\bm{0}$ in (\ref{eq.subdiff}), we have:
$\langle \bm{z},\bm{x} \rangle \ge f(\bm{x})-f(\bm{0})$ which with $f(\bm{0})={0}$ in mind due to the homogeneity, $f^*(\bm{z})\ge 1$. Also, by taking $\bm{y}=\bm{v}+\bm{x}$ in (\ref{eq.subdiff}) and taking superemum from both sides under the condition $f(\bm{v})=1$ we have:
\begin{align}
f^*(\bm{z})=\sup_{f(\bm{v})=1}\langle \bm{z},\bm{v} \rangle \le \sup_{f(\bm{v})=1}(f(\bm{v}+\bm{x})-f(\bm{x}))\le 1,
\end{align}
where we used sub-additivity of $f$. Hence, we have $f^*(\bm{z})=1$ and subsequently $\langle \bm{z},\bm{x} \rangle = f(\bm{x})$.\par
On the other hand, if we have $\bm{z}$ such that $f^*(\bm{z})=1$ and $\langle \bm{z},\bm{x} \rangle = f(\bm{x})$, then for each $\bm{y}\in\mathbb{R}^n$:
\begin{align}
f(\bm{x})+\langle \bm{z},\bm{y}-\bm{x} \rangle = \langle \bm{z}, \bm{x}\rangle+\langle \bm{z},\bm{y}-\bm{x} \rangle \le f(\bm{y})\cdot
\end{align}
\end{proof}
\subsection{Descent Cones and Normal Cones}
The descent cone of a proper convex function $f:\mathbb{R}^n\rightarrow \mathbb{R}\cup \{\pm\infty\}$ at point $\bm{x}\in \mathbb{R}^n$ is the set of directions from $\bm{x}$ that do not increase $f$ at least for one step size:
\begin{align}\label{eq.descent cone}
\mathcal{D}(f,\bm{x})=\bigcup_{t\ge0}\{\bm{z}\in\mathbb{R}^n: f(\bm{x}+t\bm{z})\le f(\bm{x})\}\cdot
\end{align}
The descent cone of a convex function is a convex set. The normal cone of a convex function at $\bm{x}\in \mathbb{R}^n$ is defined as the polar of the descent cone defined by:
\begin{align}\label{eq.normalcone}
\mathcal{N}_{f}(\bm{x})=\mathcal{D}^{\circ}(f,\bm{x})=\{\bm{y}\in \mathbb{R}^n: \langle \bm{y},\bm{z}\rangle\le0~:~\forall \bm{z}\in \mathcal{D}(f,\bm{x})\}.
\end{align}
There is a famous duality \cite{rockafellar2015convex} between decent cone and subdifferntial of a convex function  given by:
\begin{align}\label{eq.D(f,x)}
\mathcal{D}^{\circ}(f,\bm{x})=\mathrm{cone}(\partial f(\bm{x})):=\bigcup_{t\ge0}t.\partial f(\bm{x}).
\end{align}
\subsection{Gaussian Width and Statistical Dimension}
\begin{defn}{Gaussian Width}\cite{chandrasekaran2012convex}: The Gaussian width of a set $\mathcal{C}\subset \mathbb{R}^n$ is defined as:
\begin{align}\label{eq.Gaussianwidth}
\omega(\mathcal{C}):=\mathbb{E}\sup_{\bm{v}\in \mathcal{C}} \langle \bm{g},\bm{v} \rangle,
\end{align}
where, $\bm{g}\in \mathbb{R}^n$ is a vector with independent and identically distributed (i.i.d) standard normal entries.
\end{defn}
\begin{defn}{Statistical Dimension}\cite{amelunxen2013living}:
Let $\mathcal{C}\subseteq\mathbb{R}^n$ be a convex closed cone. Statistical dimension of $\mathcal{C}$ is defined as:
\begin{align}\label{eq.statisticaldimension}
\delta(\mathcal{C}):=\mathds{E}\|\mathcal{P}_\mathcal{C}(g)\|_2^2=\mathds{E}\mathrm{dist}^2(g,\mathcal{C}^\circ),
\end{align}
where, $\mathcal{P}_\mathcal{C}(\bm{x})$ is the projection of $\bm{x}\in \mathbb{R}^n$ onto the set $\mathcal{C}$ defined as: $\mathcal{P}_\mathcal{C}(\bm{x})=\underset{\bm{z} \in \mathcal{C}}{\arg\min}\|\bm{z}-\bm{x}\|_2$.
\end{defn}
Statistical dimension generalizes the concept of dimension of subspaces to the class of convex cones.
\subsection{Linear Inverse Problems and Optimality Condition}
In \cite{amelunxen2013living}, it is proved that any random convex optimization problem undergoes a phase transition as number of measurements increased. The location of the transition is determined by the statistical dimension of descent cone of $f$ at $\bm{x}\in \mathbb{R}^n$ i.e. $\delta(\mathcal{D}(f,\bm{x}))$. Width of the transition from failure to success of $\mathsf{P}_f$ is $\mathcal{O}(\sqrt{n})$ measurements.
\begin{align}\label{eq.mainprob}
\mathsf{P}_f:~~~\min_{\bm{x}\in \mathbb{R}^n} ~f(\bm{x})~~~~
\mathrm{s.t.} ~~\bm{Ax}=\bm{b}.
\end{align}
First we express the optimality condition for $\mathsf{P}_f$ in the noise-free case.
\begin{prop}\cite[Proposition 2.1]{chandrasekaran2012convex} Optimality condition: Let $f$ be a proper convex function. The vector $\bm{x}\in \mathbb{R}^n$ is the unique optimal point of $\mathsf{P}_f$ if $\mathcal{D}(f,\bm{x})\cap \mathrm{null}(\bm{A})=\{\bm{0}\}$.
\end{prop}
The next theorem determines number of measurements needed for successful recovery in $\mathsf{P}_f$ for any proper convex function $f$.
\begin{thm}\label{thm.Pfmeasurement}\cite[Theorem 2]{amelunxen2013living}:
Let $f:\mathbb{R}^n\rightarrow \mathbb{R}\cup \{\pm\infty\}$ be a proper convex function and $\bm{x}\in \mathbb{R}^n$ a fixed sparse vector. Suppose that $m$ independent Gaussian linear measurements are taken from $\bm{x}$ collected in a vector $\bm{y}=\bm{Ax} \in \mathbb{R}^m$. Then for a given tolerance $\eta \in [0,1]$ if
$m\ge \delta(\mathcal{D}(f,\bm{x}))+\sqrt{8\log(\frac{4}{\eta})n}$ we have: $\mathds{P}(\mathcal{D}(f,\bm{x})\cap \mathrm{null}(\bm{A})=\{\bm{0}\})\ge 1-\eta$. On the other hand, if $m\le \delta(\mathcal{D}(f,\bm{x}))-\sqrt{8\log(\frac{4}{\eta})n}$ then, $\mathds{P}(\mathcal{D}(f,\bm{x})\cap \mathrm{null}(\bm{A})=\{\bm{0}\})\le \eta$.
\end{thm}
Also in \cite{amelunxen2013living}, an error bound for the statistical dimension is given by the following theorem.
\begin{thm}\cite[Theorem 4.3]{amelunxen2013living} For any $\bm{x}\in \mathbb{R}^n\setminus\{\bm{0}\}$:
\begin{align}\label{eq.errorbound}
0\le\inf_{t\ge0}\mathds{E}\mathrm{dist}^2(\bm{g},t\partial f(\bm{x}))- \delta(\mathcal{D}(f,\bm{x}))\le \frac{2\sup_{s\in \partial f(\bm{x})}\|s\|_2}{f(\frac{\bm{x}}{\|\bm{x}\|_2})}.
\end{align}
\end{thm}
\section{Entrywise sparsity in a dictionary}\label{section.l1analysis}
In this section, the following questions are investigated about a non-uniform sparse model in $\bm{\Omega}\in \mathbb{R}^{p\times n}$:
\begin{enumerate}
\item How many measurements one needs to recover an $s$-sparse vector in $\bm{\Omega}\in \mathbb{R}^{p\times n}$ by solving $\mathsf{P}_{1,\bm{\Omega}}$ and $\mathsf{P}_{1,\bm{\Omega},\bm{w}}$?
\item What is the optimal choice of weights in $\mathsf{P}_{1,\bm{\Omega},\bm{w}}$ given extra prior information?
\end{enumerate}
\subsection{Number of Measurements for successful Recovery}
First we state the following theorem.
\begin{thm}\label{thm.Omega}
Let $\bm{x}\in \mathbb{R}^n$ be a fixed $s$-sparse vector in the dictionary $\bm{\Omega}$ (i.e. $\|\bm{\Omega x}\|_0\le s$). Suppose that $m$ independent Gaussian linear measurements are taken from $\bm{x}$ collected in a vector $\bm{y}=\bm{Ax} \in \mathbb{R}^m$. Then for a given tolerance $\eta \in[0,1]$, if
\begin{align}
m\ge {\kappa}^2(\bm{\Omega})\delta(\mathcal{D}(\|.\|_1,\bm{\Omega x}))+1+\sqrt{8\log\frac{4}{\eta}n} ,\nonumber
\end{align}
 then the program $\mathsf{P}_{1,\bm{\Omega}}$ succeeds with probability at least $1- \eta \exp(-\frac{\mu^2}{8n})$. In other words, we have:
  \begin{align}
  \mathds{P}(\mathcal{D}(\|\bm{\Omega}.\|_1,\bm{x})\cap \mathrm{null}(\bm{A})=\{\bm{0}\})\ge 1- \eta \exp\big(-\frac{\mu^2}{8n}\big),\nonumber
   \end{align}
   where $\kappa(\bm{\Omega})=\frac{\sigma_{max}(\bm{\Omega})}{\sigma_{min}(\bm{\Omega})}$ is the condition number of the matrix $\bm{\Omega}$ and $\mu=\kappa^2(\bm{\Omega})\delta(\mathcal{D}(\|.\|_1,\bm{\Omega x}))-\delta(\mathcal{D}(\|\bm{\Omega}.\|_1, \bm{x}))+1$.
\end{thm}
In what follows in this section, we denote the normalized number of measurements required to solve $\mathsf{P}_{1,\bm{\Omega}}$ and $\bm{P}_{1,\bm{\Omega},\bm{w}}$ with probability $1- \eta \exp(-\frac{\mu^2}{8n})$ by $m_{p,s}$ and $m_{p,s,\bm{w}}$, respectively which are defined as:
\begin{align}\label{eq.mps}
m_{p,s}&:=\kappa^2(\bm{\Omega})\frac{\delta(\mathcal{D}(\|.\|_1,\bm{\Omega} \bm{x}))}{p}+\frac{1}{p}.
\end{align}
\begin{align}\label{eq.mpsw}
m_{p,s,\bm{w}}&:=\kappa^2(\bm{\Omega})\frac{\delta(\mathcal{D}(\|.\|_{1,\bm{w}},\bm{\Omega x}))}{p}+\frac{1}{p}.
\end{align}
Theorem \ref{thm.Omega} is proved in Appendix \ref{appendix.2}. Meanwhile in the proof, we derive an upper bound for $\delta(\mathcal{D}(\|\bm{\Omega}.\|_1,\bm{x}))$ given in the following.
\begin{prop}\label{prop.upperforstatis} Statistical dimension of the descent cone of $\|\bm{\Omega}.\|_1$ at $\bm{x}\in \mathbb{R}^n$ is upper bounded by statistical dimension of the descent cone of $\|\cdot\|_1$ at $\bm{\Omega x}$ up to a positive scaling depending on $\bm{\Omega}$. Precisely, we have:
\begin{align}\label{eq.upper}
\delta(\mathcal{D}(\|\bm{\Omega}.\|_1,\bm{x}))\le \kappa^2(\bm{\Omega})\delta(\mathcal{D}(\|\cdot\|_1,\bm{\Omega x}))+1.
\end{align}
\end{prop}
Proof of Proposition \ref{prop.upperforstatis} is given in Appendix \ref{appendix.1}. In the following, we obtain upper bounds for number of measurements required for $\mathsf{P}_{1,\bm{\Omega}}$ and $\mathsf{P}_{1,\bm{\Omega},\bm{w}}$ to succeed.
\begin{lem}\label{lemma.l1anaupper}
Let $\bm{x}\in \mathbb{R}^n$ be an $s$-sparse vector in $\bm{\Omega}\in \mathbb{R}^{p\times n}$. Then an upper bound for normalized number of measurements required for $\mathsf{P}_{1,\bm{\Omega}}$ to succeed (i.e. $m_{p,s}$) is given by:
\begin{align}\label{eq.mhatps}
\hat{m}_{p,s}:=\kappa^2(\bm{\Omega})\underset{{t\ge0}}{\inf}\Psi_t(\sigma)+\frac{1}{p},
\end{align}
with $\Psi_t(\sigma)$ defined as:
\begin{align}\label{eq.analyisupper}
\Psi_t(\sigma)=\sigma(1+t^2)+(1-\sigma)\phi(t),
\end{align}
where $\phi(z):=\sqrt{\frac{2}{\pi}}\int_{z}^{\infty}(u-z)^2\exp(-\frac{u^2}{2})du$.
\end{lem}
\begin{lem}\label{lemma.l1weightedanaupper}
Let $\bm{x}\in \mathbb{R}^n$ be a non-uniform $s$-sparse vector in $\bm{\Omega}\in \mathbb{R}^{p\times n}$ with parameters $\{\rho_i\}_{i=1}^L$ and $\{\alpha_i\}_{i=1}^L$. Then an upper bound for normalized number of measurements required for $\mathsf{P}_{1,\bm{\Omega},\bm{w}}$ with $\bm{w}=\sum_{i=1}^{L}\omega_i1_{\mathcal{P}_i}$ to succeed (i.e. $m_{p,s,\bm{w}}$) is given by:
\begin{align}\label{eq.mhat_psw}
\hat{m}_{p,s,\bm{w}}:=\kappa^2(\bm{\Omega})\underset{{t\ge0}}{\inf}\Psi_{t,\bm{w}}(\sigma,\bm{\rho},\bm{\alpha})+\frac{1}{p},
\end{align}
with $\Psi_{t,\bm{w}}(\sigma,\bm{\rho},\bm{\alpha})$ defined as:
\begin{align}\label{eq.weightedanalyisupper}
\Psi_{t,\bm{w}}(\sigma,\bm{\rho},\bm{\alpha})=
\sum_{i=1}^{L}\rho_i\big(\alpha_i(t^2\omega_i^2+1)+(1-\alpha_i)\phi(t\omega_i)\big).
\end{align}
\end{lem}
\begin{proof}
By considering (\ref{eq.mpsw}), we only need to find an upper bound for $\delta(\mathcal{D}(\|.\|_{1,\bm{w}},\bm{\Omega x}))/p$ defined as $\Psi_{t,\bm{w}}(\sigma,\bm{\rho},\bm{\alpha})$. By definition of the statistical dimension in (\ref{eq.statisticaldimension}), the fact that infimum of an affine function is concave and Jensen's inequality, we have:
\begin{align}\label{eq.upforstatistical}
p^{-1}\delta(\mathcal{D}(\|.\|_{1,\bm{w}},\bm{\Omega x}))\le \inf_{t\ge0}\underbrace{{p^{-1}}\mathds{E}\mathrm{dist}^2(\bm{g},t\partial\|.\|_{1,\bm{w}}(\bm{\Omega x}))}_{\Psi_{t,\bm{w}}(\sigma,\rho,\alpha)}.
\end{align}
By (\ref{eq.upforstatistical}), we define our upper bound for $m_{p,s,\bm{w}}$ (i.e. $\hat{m}_{p,s,\bm{w}}$) which is given by:
\begin{align}\label{eq.mhatpsw}
\hat{m}_{p,s,\bm{w}}=\kappa^2(\bm{\Omega})\inf_{t\ge0}\Psi_{t,\bm{w}}(\sigma,\bm{\rho},\bm{\alpha})+\frac{1}{p}.
\end{align}
The next step is to calculate $\Psi_{t,\bm{w}}(\sigma,\bm{\rho},\bm{\alpha})$. For this purpose, $\partial\|.\|_{1,\bm{w}}(\bm{\Omega x})$ should be computed. From Proposition \ref{prop.simplerform of subdiff}, we have:
\begin{align}\label{eq.subw1}
\partial\|.\|_{1,\bm{w}}(\bm{\Omega x})=\{\bm{z}\in \mathbb{R}^p:~\langle \bm{z},\bm{\Omega x}\rangle=\|\bm{\Omega x}\|_{1,\bm{w}},~\|\bm{z}\|_{1,\bm{w}}^*=1\}.
\end{align}
\begin{align}\label{eq.subw}
&\|\bm{z}\|_{1,\bm{w}}^*=\sup_{\|\bm{y}\|_{1,\bm{w}}\le1}\langle \bm{z},\bm{y} \rangle=\sup_{\|\bm{y}\odot \bm{w}\|_{1}\le1}\langle \bm{w}\odot \bm{y},\bm{z}\odot \bm{w}^{\odot -1}  \rangle \nonumber\\
\le &\|\bm{y}\odot \bm{w}\|_1\|\bm{z}\odot \bm{w}^{\odot -1}\|_{\infty}\le\|\bm{z}\odot \bm{w}^{\odot -1}\|_{\infty}=\max_{i\in[p]}|w_i^{-1}z_i|,
\end{align}
where in (\ref{eq.subw}), the inequality comes from H\"older's inequality with equality when
\[
    \bm{y} = \left\{\begin{array}{lr}
        sgn(z_k), &  k=\underset{i\in[p]}{\arg\max}~~ {w_i}^{-1}z_i\\
        0, & \mathrm{o.w.}
        \end{array}\right\}\in \mathbb{R}^p.
  \]
From the first part of (\ref{eq.subw1}), we find that $\bm{z}_\mathcal{S}=\bm{w}_\mathcal{S}\odot sgn(\bm{\Omega x})$ and from the second part together with (\ref{eq.subw}), we find out that $|w_i^{-1}z_i|\le1$. Hence,
\begin{align}\label{eq.l1subdiff}
    \partial\|\cdot\|_{1,\bm{w}}(\bm{\Omega x}) = \left\{\bm{z}\in\mathbb{R}^p:\begin{array}{lr}
        z_i=w_i~sgn(\Omega x)_i, &  i\in \mathcal{S}\\
        |z_i|\le w_i, & i\in {\mathcal{\bar{S}}}
        \end{array}\right\}.
  \end{align}
The distance between dilated subdifferntial of the descent cone of $\ell_{1,\bm{w}}$ norm at $\bm{\Omega x}$ and a standard Gaussian vector $\bm{g}$ is given by:
\begin{align}\label{eq.distsubdiff}
&\mathrm{dist}^2(\bm{g},t\partial\|.\|_{1,\bm{w}}(\bm{\Omega x}))=\inf_{\bm{z}\in\partial\|.\|_{1,\bm{w}}(\bm{\Omega x})}\|\bm{g}-t\bm{z}\|_2^2=\nonumber\\
&\sum_{i\in \mathcal{S}}(g_i-tw_isgn(\Omega x)_i)^2+\sum_{i\in {\mathcal{\bar{S}}}}\inf_{|z_i|\le w_i}(g_i-tz_i)^2=\nonumber\\
&\sum_{i\in \mathcal{S}}(g_i-tw_isgn(\Omega x)_i)^2+\sum_{i\in {\mathcal{\bar{S}}}}(|g_i|-tw_i)_{+}^2,
\end{align}
where, $(a)_{+}:=\max\{a,0\}$ and triangle inequality is used in the second part. The expected value of (\ref{eq.distsubdiff}) is:
\begin{align}\label{eq.Edist2}
&\mathds{E}\mathrm{dist}^2(\bm{g},t\partial\|.\|_{1,\bm{w}}(\bm{\Omega x}))=\nonumber\\
&s+\sum_{i\in \mathcal{S}}(tw_i)^2+\sum_{i\in {\mathcal{\bar{S}}}}\mathds{E}(\zeta-tw_i)_{+}^2,
\end{align}
where, $\zeta=|g_i|$ is distributed as a folded normal variable.
Also:
\begin{align}\label{eq.Ezetal1}
&\mathds{E}(\zeta-tw_i)_{+}^2=2\int_{0}^{\infty}a\mathds{P}(\zeta\ge a+tw_i)da\nonumber\\
&=2\sqrt{\frac{2}{\pi}}\int_{0}^{\infty}\int_{tw_i+a}^{\infty}a~e^{-\frac{u^2}{2}}du ~da\nonumber\\
&=2\sqrt{\frac{2}{\pi}}\int_{tw_i}^{\infty}\int_{0}^{u-tw_i}a~e^{-\frac{u^2}{2}}da~du=\phi(tw_i),
\end{align}
where in the third line, order of integration was changed. Thus, (\ref{eq.Edist2}) reads:
\begin{align}\label{eq.Edist2last}
&\mathds{E}\mathrm{dist}^2(g,t\partial\|.\|_{1,\bm{w}}(\bm{\Omega x}))=s+\sum_{i\in \mathcal{S}}(tw_i)^2+\sum_{i\in {\mathcal{\bar{S}}}}\phi(tw_i).
\end{align}
By normalizing to the ambient dimension $p$ and incorporating prior information using $\bm{w}=\sum_{i=1}^{L}\omega_i1_{\mathcal{P}_i}$ we reach:
\begin{align}
&\mathds{E}\mathrm{dist}^2(\bm{g},t\partial\|.\|_{1,\bm{w}}(\bm{\Omega x}))=\nonumber\\
&s+\sum_{i=1}^{L}|\mathcal{P}_i\cap \mathcal{S}|t^2\omega_i^2+|\mathcal{P}_i\cap {\mathcal{\bar{S}}}|\phi(t\omega_i)
\nonumber\\
&=p\bigg(\sigma+\sum_{i=1}^{L}\rho_i\big(\alpha_it^2\omega_i^2+(1-\alpha_i)\phi(t\omega_i)\big)\bigg)\nonumber\\
&=p\bigg(\sum_{i=1}^{L}\rho_i\big(\alpha_i(t^2\omega_i^2+1)+(1-\alpha_i)\phi(t\omega_i)\big)\bigg),\label{eq.Edist2part2}
\end{align}
where in the last line above, we benefited the fact that $\sigma=\sum_{i=1}^{L}\rho_i\alpha_i$.
\end{proof}
\begin{proof}[Proof of Lemma \ref{lemma.l1anaupper}]
Regarding (\ref{eq.mps}), we find an upper bound for $\delta(\mathcal{D}(\|.\|_{1},\bm{\Omega x}))/p$ which we denote by $\Psi_{t}(\sigma)$. The procedure is exactly the same as the proof of Lemma \ref{lemma.l1weightedanaupper} with the weights set to one i.e. $\bm{w}=\bm{1}\in \mathbb{R}^p$. In fact, we have: $\Psi_t(\sigma)=\Psi_{t,\bm{w}}(\sigma,\bm{\rho},\bm{\alpha})$. By replacing $\bm{w}=\bm{1}\in \mathbb{R}^p$ in (\ref{eq.Edist2last}) and (\ref{eq.mhatpsw}), we reach $\Psi_t(\sigma)$  and $\hat{m}_{p,s}$ in Lemma \ref{lemma.l1anaupper}.
\end{proof}
In the following propositions, using the same technique as \cite{amelunxen2013living}, we find out that the obtained upper bounds in Lemmas \ref{lemma.l1anaupper} and \ref{lemma.l1weightedanaupper} are asymptotically tight.
\begin{prop}\label{prop.analysis}. Normalized number of Gaussian linear measurements required for $\mathcal{P}_{1,\bm{\Omega}}$ to succeed (i.e. $m_{p,s}$) satisfies the following error bound.
\begin{align}\label{eq.l1analysisbound}
\hat{m}_{p,s}-\frac{2\kappa^2(\mathbf\Omega)}{\sqrt{sp}}\le m_{p,s}\le \hat{m}_{p,s}
\end{align}
\end{prop}
\begin{proof}
By the error bound in (\ref{eq.errorbound}) and (\ref{eq.l1subdiff}) with $\bm{w}=\bm{1}\in \mathbb{R}^p$, we obtain the numerator of the error bound (\ref{eq.errorbound}) as:
\begin{align}
2\sup_{s\in\partial\|.\|_1(\bm{\Omega x})}\|s\|_2\le 2\sqrt{p}.\nonumber
\end{align}
Also, for the denominator we have:
\begin{align}\label{eq.denominatorana}
\bigg\|\frac{\bm{\Omega x}}{\|\bm{\Omega x}\|_2}\bigg\|_1\le \sqrt{s}.
\end{align}
The error bound in (\ref{eq.errorbound}) depends only on $\mathcal{D}(\|.\|_1,\bm{\Omega x})$.
Further, regarding (\ref{eq.D(f,x)}), $\mathcal{D}(\|.\|_1,\bm{\Omega x})$ depends only on $sgn(\bm{\Omega x})$ not the magnitudes of $\bm{\Omega x}$. So we can choose a vector
\begin{align}
    \bm{z} = \left\{\begin{array}{lr}
        z_i=sgn(\Omega x)_i, &  i\in \mathcal{S}\\
        z_i=0, & i\in {\mathcal{\bar{S}}}
        \end{array}\right\}\in \mathbb{R}^p,\nonumber
  \end{align}
  with $sgn(\bm{z})=sgn(\bm{\Omega x})$ to have equality in (\ref{eq.denominatorana}). Therefore, the error in obtaining $\Psi_t(\sigma)$ is at most $\frac{1}{p}\frac{2\sqrt{p}}{\sqrt{s}}=\frac{2}{\sqrt{sp}}$. Further, the error of $\hat{m}_{p,s}$ in (\ref{eq.mhatps}) is at most $\frac{2\kappa^2(\bm{\Omega})}{\sqrt{sp}}$.
\end{proof}
\begin{prop}\label{prop.weightedanalysis}
Normalized number of Gaussian linear measurements required for $P_{1,\bm{\Omega},\bm{w}}$ to successfully recover a non-uniform $s$-sparse vector in $\bm{\Omega}\in\mathbb{R}^{p\times n}$ with parameters $\{\rho_i\}_{i=1}^L$ and $\{\alpha_i\}_{i=1}^L$  (i.e. $m_{p,s,\bm{w}}$) satisfies the following error bound.
\begin{align}\label{eq.l1weightedanalysisbound}
\hat{m}_{p,s,\bm{w}}-\frac{2\kappa^2(\bm{\Omega})}{\sqrt{pL}}\le m_{p,s,\bm{w}}\le \hat{m}_{p,s,\bm{w}}.
\end{align}
\end{prop}
It i interesting that the error bound in Proposition \ref{prop.upperforstatis} is a special case of the error bound of Proposition \ref{prop.weightedanalysis} where one has $s$ sets with size $\frac{p}{L}$ and knows with probability $\frac{L}{p}$ that each set contributes to the entrywise support.
\begin{proof}
By using (\ref{eq.errorbound}) and (\ref{eq.l1subdiff}), with $f(\bm{x}):=\|\bm{\Omega x}\|_{1,\bm{w}}$ and $\bm{w}=\sum_{i=1}^{L}\omega_i1_{\mathcal{P}_i}$, the numerator is given by:
\begin{align}
2\sup_{\bm{s}\in\partial\|.\|_{1,\bm{w}}(\bm{\Omega x})}\|\bm{s}\|_2\le 2\sqrt{\sum_{i=1}^{p}w_i^2}=\nonumber\\
2\sqrt{\sum_{i=1}^{L}|\mathcal{P}_i|\omega_i^2}=2\sqrt{\sum_{i=1}^{L}p\rho_i\omega_i^2}.\nonumber
\end{align}
Also, for the denominator we have:
\begin{align}\label{eq.denomweight}
\frac{\|\bm{\Omega x}\|_{1,\bm{w}}}{{\|\bm{\Omega x}\|_2}}\le \sqrt{\sum_{i\in \mathcal{S}}w_i^2}=\sqrt{\sum_{i=1}^{L}|\mathcal{P}_i\cap \mathcal{S}|\omega_i^2}=\sqrt{\sum_{i=1}^{L}p\alpha_i\rho_i\omega_i^2},
\end{align}
where the first inequality comes from Cauchy Schwartz inequality. With the same justification as in the proof of Proposition \ref{prop.weightedanalysis} and the fact that
$\mathcal{D}(\|.\|_{1,\bm{w}},\bm{\Omega x})$ depends only on $sgn(\bm{\Omega x})$, a vector
\begin{align}
    \bm{z} = \left\{\begin{array}{lr}
        z_i=w_isgn(\Omega x)_i, &  i\in \mathcal{S}\\
        z_i=0, & i\in {\mathcal{\bar{S}}}
        \end{array}\right\}\in \mathbb{R}^p,\nonumber
  \end{align}
  with $sgn(\bm{z})=sgn(\bm{\Omega x})$ can be chosen to have equality in (\ref{eq.denomweight}). Therefore, the error of obtaining the upper bound of $\delta(\mathcal{D}(\|.\|_{1,\bm{w}},\bm{\Omega x}))/p$ i.e. $\Psi_{t,\bm{w}}(\sigma,\bm{\rho},\bm{\alpha})$ is
  \begin{align}
   \frac{2\sqrt{\sum_{i=1}^{L}p\rho_i\omega_i^2}}{p\sqrt{\sum_{i=1}^{L}p\alpha_i\rho_i\omega_i^2}}\le
   \frac{2}{p}\sqrt{\frac{1}{\underset{{i\in [p]}}\min\frac{|\mathcal{P}_i\cap \mathcal{S}|}{|\mathcal{P}_i|}}}
   \le \frac{2}{\sqrt{pL}},
  \end{align}
  in which, the last inequality follows from the facts that $|\mathcal{P}_i\cap \mathcal{S}|\ge1$, $|\mathcal{P}_i|\le\frac{p}{L}$ for at least one $i\in[p]$ and thus $ \underset{{i\in [p]}}\min\frac{|\mathcal{P}_i\cap \mathcal{S}|}{|\mathcal{P}_i|}\ge \frac{L}{p}$. Further the error in obtaining $\hat{m}_{p,s,\bm{w}}$ in (\ref{eq.mhatpsw}) is at most $\frac{2\kappa^2(\bm{\Omega})}{\sqrt{pL}}$.
\end{proof}
\subsection{Optimal Weights}
Infimum of (\ref{eq.l1weightedanalysisbound}) gives:
\begin{align}\label{eq.infmhatpserrorbound}
\inf_{\bm{\omega}\in\mathbb{R}_{+}^L}\hat{m}_{p,s,\bm{D}\bm{\omega}}-\frac{2\kappa^2(\bm{\Omega})}{\sqrt{pL}}\le\inf_{\bm{\omega}\in\mathbb{R}_{+}^L}{m}_{p,s,\bm{D}\bm{\omega}}\le \inf_{\bm{\omega}\in\mathbb{R}_{+}^L}\hat{m}_{p,s,\bm{D\omega}},
\end{align}
where $\bm{D}:=[\bm{1}_{\mathcal{P}_1},...,\bm{1}_{\mathcal{P}_L}]$.
We call the weight $\bm{\omega}^*=\underset{\bm{\omega}\in\mathbb{R}_{+}^L}{\arg\min}~~\hat{m}_{p,s,\bm{D}\bm{\omega}}\in\mathbb{R}_{+}^L$ optimal since it asymptotically minimizes the number of measurements required for $\mathsf{P}_{1,\bm{\Omega},\bm{D\omega}}$ to succeed. Before showing that there exist unique optimal weights for $\mathsf{P}_{1,\bm{\Omega},\bm{D\omega}}$, we state the following lemma.
\begin{lem}{\label{lemmma.J(nu)}}
Let $\mathcal{C}:=\partial\|.\|_1(\bm{\Omega x})$. Suppose that $\mathcal{C}$ does not contain the origin. In particular, it is compact and there are upper and lower bounds that satisfy $1\le\|\bm{z}\|_2\le \sqrt{p}$ for all $\bm{z}\in \mathcal{C}$. Also denote the standard normal vector by $\bm{g}\in \mathbb{R}^p$. Consider the function
\begin{align}
J(\bm{\nu}):=\mathds{E}\mathrm{dist}^2(\bm{g},\bm{\upsilon}\odot \mathcal{C})=\mathds{E}[J_{\bm{g}}(\bm{\nu})] , \nonumber\\
 \text{with}~~\bm{\upsilon}=\underbrace{[\bm{1}_{\mathcal{P}_1},..., \bm{1}_{\mathcal{P}_L}]}_{\bm{D}_{p\times L}}\bm{\nu} ~~\text{for}~~ \bm{\nu}\in\mathbb{R}_{+}^L.
\end{align}
The function $J$ is strictly convex, continuous at $\bm{\nu}\in \mathbb{R}_{+}^L$ and differentiable for $\bm{\nu}\in \mathbb{R}_{++}^L$. More over, there exists a unique point that minimize $J$.
\end{lem}
\begin{proof}
\textit{Continuity in bounded points}.
We must show that sufficiently small changes in $\bm{\nu}$ result in arbitrary small changes in $J(\bm{\nu})$. By definition of $J_{\bm{g}}(\bm{\nu})$, we have:
\begin{align}
&J_{\bm{g}}(\bm{\nu})-J_{\bm{g}}(\tilde{\bm{\nu}})=\|\bm{g}-\mathcal{P}_{\bm{\upsilon}\odot \mathcal{C}}(\bm{g})\|_2^2-\|\bm{g}-\mathcal{P}_{\tilde{\bm{\upsilon}}\odot \mathcal{C}}(\bm{g})\|_2^2\nonumber\\&=2\langle \bm{g},\mathcal{P}_{\tilde{\bm{\upsilon}}\odot \mathcal{C}}(\bm{g})-\mathcal{P}_{{\bm{\upsilon}}\odot \mathcal{C}}(\bm{g})\rangle+\nonumber\\
&\big(\|\mathcal{P}_{\bm{\upsilon}\odot \mathcal{C}}(\bm{g})\|_2-\|\mathcal{P}_{\tilde{\bm{\upsilon}}\odot \mathcal{C}}(\bm{g})\|_2\big)\big(\|\mathcal{P}_{\bm{\upsilon}\odot \mathcal{C}}(\bm{g})\|_2+\|\mathcal{P}_{\tilde{\bm{\upsilon}}\odot \mathcal{C}}(\bm{g})\|_2\big)\nonumber\\
&\text{The absolute value satisfies:}\nonumber\\
&|J_{\bm{g}}(\bm{\nu})-J_{\bm{g}}(\tilde{\bm{\nu}})|\nonumber\\
&\le \bigg(2\|\bm{g}\|_2\sqrt{p}+p^2(\|\bm{\nu}\|_1+\|\tilde{\bm{\nu}}\|_1)\bigg)\|\tilde{\bm{\nu}}-\bm{\nu}\|_1,
\end{align}
where, we used the fact that,
\begin{align}
&\|\mathcal{P}_{\bm{\upsilon}\odot \mathcal{C}}(\bm{g})\|_2\le\sup_{\bm{z}\in \mathcal{C}}\|\bm{\upsilon}\odot \bm{z}\|_2\le\|\bm{D\nu}\|_{\infty}\sqrt{p}\le\nonumber\\
&\|\bm{\nu}\|_1\sqrt{p}\|\bm{D}\|_{1\rightarrow\infty}=\|\bm{\nu}\|_1\sqrt{p},
\end{align}
and,
\begin{align}
&\|\mathcal{P}_{\bm{\upsilon}\odot \mathcal{C}}(\bm{g})\|_2-\|\mathcal{P}_{\tilde{\bm{\upsilon}}\odot \mathcal{C}}(\bm{g})\|_2\le \sup_{\bm{z}\in \mathcal{C}}\bigg(\|\bm{\upsilon}\odot \bm{z}\|_2-\|\tilde{\bm{\upsilon}}\odot \bm{z}\|_2\bigg)\nonumber\\
&\le\sup_{\bm{z}\in \mathcal{C}}\big(\|(\bm{\upsilon}-\tilde{\bm{\upsilon}})\odot \bm{z}\|_2\big)\le
\|\bm{\nu}-\tilde{\bm{\nu}}\|_1\sqrt{p}\|\bm{D}\|_{1\rightarrow\infty}=\nonumber\\
&\|\bm{\nu}-\tilde{\bm{\nu}}\|_1\sqrt{p}.
\end{align}
As a consequence, we obtain:
\begin{align}
&|J(\bm{\nu})-J(\tilde{\bm{\nu}})|\nonumber\\
&\le \bigg(2p\sqrt{L}+p\sqrt{L}(\|\bm{\nu}\|_1+\|\tilde{\bm{\nu}}\|_1)\bigg)\|\tilde{\bm{\nu}}-\bm{\nu}\|_2\rightarrow 0 \nonumber\\
 &~\text{as}~ \bm{\nu}\rightarrow~ \tilde{\bm{\nu}}.
\end{align}
Since $\|\bm{\nu}\|_1$ is bounded, continuity holds.\par
\textit{Convexity}. Let $\bm{\nu}~,\tilde{\bm{\nu}}\in\mathbb{R}_{+}^L$ and $\theta\in[0,1]$ with $\bm{\upsilon}=\bm{D\nu}$ and $\tilde{\bm{\upsilon}}=\bm{D}\tilde{\bm{\nu}}$. Then we have:
\begin{align}\label{helpconvexityana}
&\forall \epsilon , \tilde{\epsilon}>0~\exists \bm{z} ,\tilde{\bm{z}}\in \mathcal{C} ~~\text{such that}\nonumber\\
&\|\bm{g}-\bm{\upsilon}\odot \bm{z}\|_2\le \mathrm{dist}(\bm{g},\bm{\upsilon}\odot \mathcal{C})+\epsilon\nonumber\\
&\|\bm{g}-\tilde{\bm{\upsilon}}\odot \tilde{\bm{z}}\|_2\le \mathrm{dist}(\bm{g},\tilde{\bm{\upsilon}}\odot \mathcal{C})+\tilde{\epsilon}.
\end{align}
Since otherwise we have:
\begin{align}
&\forall \bm{z},\tilde{\bm{z}}\in \mathcal{C}:\nonumber\\
&\|\bm{g}-\bm{\upsilon}\odot \bm{z}\|_2>\mathrm{dist}(\bm{g},\bm{\upsilon}\odot \mathcal{C})+\epsilon\nonumber\\
&\|\bm{g}-\tilde{\bm{\upsilon}}\odot \tilde{\bm{z}}\|_2> \mathrm{dist}(\bm{g},\tilde{\bm{\upsilon}}\odot \mathcal{C})+\tilde{\epsilon}.
\end{align}
By taking the infimum over $\bm{z},\tilde{\bm{z}}\in \mathcal{C}$, we reach a contradiction. We proceed to prove convexity of $\mathrm{dist}(\bm{g},(\bm{D\nu})\odot \mathcal{C})$.
\begin{align}\label{eq.convexitydistana}
&\mathrm{dist}(\bm{g},(\theta\bm{\upsilon}+(1-\theta)\tilde{\bm{\upsilon}})\odot \mathcal{C})=\nonumber\\
&\inf_{\bm{z} \in \mathcal{C}}\|\bm{g}-(\theta\bm{\upsilon}+(1-\theta)\tilde{\bm{\upsilon}})\odot \bm{z}\|_2\nonumber\\
&\le\inf_{\bm{z}_1\in \mathcal{C},\bm{z}_2\in \mathcal{C}}\|\bm{g}-\theta\bm{\upsilon}\odot \bm{z}_1-(1-\theta)\bm{\upsilon}\odot \bm{z}_2\|_2\le\nonumber\\
&\theta\|\bm{g}-\bm{\upsilon}\odot \bm{z}_1\|_2+(1-\theta)\|\bm{g}-\tilde{\bm{\upsilon}}\odot \bm{z}_2\|_2\le\nonumber\\
&\theta \mathrm{dist}(\bm{g},\bm{\upsilon}\odot \mathcal{C})+(1-\theta)\mathrm{dist}(\bm{g},\tilde{\bm{\upsilon}}\odot \mathcal{C})+\epsilon+\tilde{\epsilon}.
\end{align}
Since this holds for any $\epsilon$ and $\tilde{\epsilon}$, $\mathrm{dist}(\bm{g},(\bm{D\nu})\odot \mathcal{C})$ is a convex function. As the square of a non-negative convex function is convex, $J_{\bm{g}}(\bm{\nu})$ is a convex function. At last, the function $J(\bm{\nu})$ is the average of convex functions, hence is convex.
In (\ref{eq.convexitydistana}), the first inequality comes from the fact that $\forall \bm{z}_1,\bm{z}_2 \in \mathcal{C}~~\exists \bm{z}\in \mathcal{C}$:
\begin{align}\label{eq.benefitl1}
&\theta\bm{\upsilon}\odot \bm{z}_1+(1-\theta)\tilde{\bm{\upsilon}}\odot \bm{z}_2=\nonumber\\
&\left\{\bm{y}\in\mathbb{R}^p:\begin{array}{lr}
       y_i=\\
       \big(\theta\upsilon(i)+(1-\theta)\tilde{\upsilon}(i)\big)sgn(\Omega x)_i, &  i\in \mathcal{S}\\
        |y(i)|\le\theta\upsilon(i)|z_1(i)|+(1-\theta)\tilde{\upsilon}|z_2(i)|, & i\in {\mathcal{\bar{S}}}
        \end{array}\right\}\nonumber\\
&\in\left\{\bm{y}\in\mathbb{R}^p:\begin{array}{lr}
       y_i=\\
       \big(\theta\upsilon(i)+(1-\theta)\tilde{\upsilon}(i)\big)sgn(\Omega x)_i, &  i\in \mathcal{S}\\
        |y(i)|\le\big(\theta\upsilon(i)+(1-\theta)\tilde{\upsilon}\big)|z(i)|, & i\in {\mathcal{\bar{S}}}
        \end{array}\right\}\nonumber\\
&=(\theta \bm{\upsilon}+(1-\theta)\tilde{\bm{\upsilon}})\odot \bm{z}.
\end{align}
To verify (\ref{eq.benefitl1}), note that:
\begin{align}
&\forall \bm{z}\in \mathcal{C} ~\exists j\in {\mathcal{\bar{S}}}~~\text{such that}:\nonumber\\
&(\theta \upsilon(j)+(1-\theta)\tilde{\upsilon}(j))|z(j)|<\theta\upsilon(j)|z_1(j)|+(1-\theta)\tilde{\upsilon}(j)|z_2(j)|\nonumber\\
&\le \theta \upsilon(j)+(1-\theta)\tilde{\upsilon}(j),
\end{align}
Then, setting $z(j)=1$ we reach a contradiction.
In the second inequality in (\ref{eq.convexitydistana}), we used triangle inequality of norms. The third inequality uses (\ref{helpconvexityana}).\par
\textit{Strict convexity}. We show strict convexity by contradiction. If $J(\bm{\nu})$ were not strictly convex, there would be vectors $\bm{\nu},\tilde{\bm{\nu}}\in\mathbb{R}_{+}^L$ with $\bm{\upsilon}=\bm{D}\bm{\nu}, \tilde{\bm{\upsilon}}=\bm{D}\tilde{\bm{\nu}}$ and $\theta\in (0,1)$ such that,
\begin{align}\label{eq.strictconvexana}
\mathds{E}[J_{\bm{g}}(\theta\bm{\nu}+(1-\theta)\tilde{\bm{\nu}})]=\mathds{E}[\theta J_{\bm{g}}(\bm{\nu})+(1-\theta)J_{\bm{g}}(\tilde{\bm{\nu}})].
\end{align}
For each $\bm{g}$ in (\ref{eq.strictconvexana}) the left-hand side is smaller than or equal to the right-hand side. Therefore, in (\ref{eq.strictconvexana}), $J_{\bm{g}}(\theta\bm{\nu}+(1-\theta)\tilde{\bm{\nu}})$ and $\theta J_{\bm{g}}(\bm{\nu})+(1-\theta)J_{\bm{g}}(\tilde{\bm{\nu}})$ are almost surely equal (except at a measure zero set) with respect to Gaussian measure. Moreover, we have:
\begin{align}\label{eq.J0}
&J_{\bm{0}}(\theta\bm{\nu}+(1-\theta)\tilde{\bm{\nu}})=\mathrm{dist}^2(\bm{0},\big(\theta\bm{\upsilon}+(1-\theta)\tilde{\bm{\upsilon}}\big)\odot \mathcal{C})\nonumber\\
&\le \inf_{\bm{z}_1,\bm{z}_2\in \mathcal{C}}\|\theta \bm{\upsilon}\odot \bm{z}_1+(1-\theta)\tilde{\bm{\upsilon}}\odot \bm{z}_2\|_2^2\nonumber\\
&<\theta\inf_{\bm{z}_1\in \mathcal{C}}\|\bm{\upsilon}\odot \bm{z}_1\|_2^2+(1-\theta)\inf_{\bm{z}_2\in \mathcal{C}}\|\tilde{\bm{\upsilon}}\odot \bm{z}_2\|_2^2\nonumber\\
&=\theta J_{\bm{0}}(\bm{\nu})+(1-\theta)J_{\bm{0}}(\tilde{\bm{\nu}}),
\end{align}
where, the first inequality comes from (\ref{eq.benefitl1}) and the second inequality stems from the strict convexity of $\|.\|_2^2$. From (\ref{eq.benefitl1}), it is easy to verify that the set $\bm{\nu}\odot \mathcal{C}$ is a convex set. The distance to a convex set e.g. $\mathcal{E}$ i.e. $\mathrm{dist}(\bm{g},\mathcal{E})$ is a 1-lipschitz function (i.e. $|\mathrm{dist}(\bm{g},\mathcal{E})-\mathrm{dist}(\tilde{\bm{g}},\mathcal{E})|\le\|\bm{g}-\tilde{\bm{g}}\|_2~:~\forall~\bm{g},\tilde{\bm{g}}\in\mathbb{R}^p$) and hence continuous with respect to $\bm{g}$. Therefore, $J_{\bm{g}}(\bm{\nu})$ is continuous with respect to $\bm{g}$. So there exist an open ball around $\bm{g}=\bm{0}\in\mathbb{R}^p$ that similar to (\ref{eq.J0}), we may write the following relation for some $\epsilon>0$
\begin{align}
&\exists \bm{u}\in\mathds{B}_{\epsilon}^p: \nonumber\\
&J_{\bm{u}}(\theta\bm{\nu}+(1-\theta)\tilde{\bm{\nu}})<\theta J_{\bm{u}}(\bm{\nu})+(1-\theta)J_{\bm{u}}(\tilde{\bm{\nu}}).
\end{align}
$\mathds{B}_{\epsilon}^p$ is not a measure zero set. Thus, the above statement contradicts (\ref{eq.strictconvexana}) and hence we have strict convexity. Continuity along with convexity of $J$ implies that $J$ is convex on the whole domain $\bm{\nu}\in\mathbb{R}_{+}^L$.\par
\textit{Differentiability}. The function $J_{\bm{g}}(\bm{\nu})$ is continuously differentiable and the gradient for $\bm{\nu}\in\mathbb{R}_{++}^L$ is:
\begin{align}\label{eq.diffJg}
&\nabla_{\bm{\nu}}J_{\bm{g}}(\bm{\nu})=\nonumber\\
&\frac{\partial J_{\bm{g}}(\bm{\nu})}{\partial\bm{\nu}}=-2\bm{D}^T(\bm{D}\bm{\nu})^{-1\odot}\odot\mathcal{P}_{(\bm{D}\bm{\nu})\odot\mathcal{C}}(\bm{g})\nonumber\\
&\odot(\bm{g}-\mathcal{P}_{\bm{D\nu}\odot\mathcal{C}}(\bm{g})).
\end{align}
Continuity of $\frac{\partial J_{\bm{g}}(\bm{\nu})}{\partial\bm{\nu}}$ at $\bm{\nu}\in\mathbb{R}_{+}^L$ stems from the fact that the projection onto a convex set is continuous. For each compact set $\mathcal{I}\subseteq\mathbb{R}_{+}^L$ we have:
\begin{align}
\mathds{E}\sup_{\bm{\nu}\in\mathcal{I}}\|\nabla_{\bm{\nu}} J_{\bm{g}}(\bm{\nu})\|_2\le2\|\bm{D}\|_{2\rightarrow2}p(1+\big(\sup_{\bm{\nu}\in\mathcal{I}}\nu_{\max}\big))<\infty,
\end{align}
where $\nu_{\max}:=\underset{i\in[L]}{\max}~\nu(i)$. Therefore, we have:
\begin{align}
\nabla_{\bm{\nu}} J(\bm{\nu})=\big(\frac{\partial}{\partial\bm{\nu}}\big)\mathds{E} J_{\bm{g}}(\bm{\nu})=\mathds{E}[\nabla_{\bm{\nu}} J_{\bm{g}}(\bm{\nu})] ~:~\forall \bm{\nu}\in\mathbb{R}_{+}^L,
\end{align}
where in the last equality, we used the Lebesgue's dominated convergence theorem. \par
\textit{Attainment of the minimum}. Suppose that $\bm{\nu}>{\|\bm{g}\|_2}\bm{1}_{L\times 1}$. With this assumption we may write:
\begin{align}\label{eq.attainana}
&\mathrm{dist}(\bm{g},(\bm{D\nu})\odot \mathcal{C})=\inf_{\bm{z}\in \mathcal{C}}\|\bm{g}-\bm{\upsilon}\odot \bm{z}\|_2\ge\nonumber\\
&\inf_{\bm{z}\in \mathcal{C}}(\|\bm{\upsilon}\odot \bm{z}\|_2-\|\bm{g}\|_2)\ge\nu_{\min}-\|\bm{g}\|_2\ge0,
\end{align}
where in (\ref{eq.attainana}), $\nu_{\min}:=\underset{i\in[L]}{\min}~\nu(i)$. By squaring (\ref{eq.attainana}), we reach:
\begin{align}\label{eq.attainJ_gana}
J_{\bm{g}}(\bm{\nu})\ge (\nu_{\min}-\|\bm{g}\|_2)^2~~:~\forall \bm{\nu}>{\|\bm{g}\|_2}\bm{1}_{L\times 1}.
\end{align}
Using the relation $\mathds{E}\|\bm{g}\|_2\ge\frac{p}{\sqrt{p+1}}$ (\cite[Proposition 8.1.]{foucart2013mathematical}) and Marcov's inequality we obtain:
\begin{align}
\mathds{P}(\|\bm{g}\|_2\le\sqrt{p})\ge1-\sqrt{\frac{p}{p+1}}\nonumber.
\end{align}
Then we reach:
\begin{align}\label{eq.attainJana}
&J(\bm{\nu})\ge\mathds{E}[J_{\bm{g}}(\bm{\nu})|\|\bm{g}\|_2\le\sqrt{p}]\mathds{P}(\|\bm{g}\|_2\le\sqrt{p})\nonumber\\
&\ge(1-\sqrt{\frac{p}{p+1}})\mathds{E}\big[(\nu_{\min}-\|\bm{g}\|_2)^2|\|\bm{g}\|_2\le\sqrt{p}\big]\nonumber\\
&\ge(1-\sqrt{\frac{p}{p+1}})(\nu_{\min}-\sqrt{p})^2,
\end{align}
where in (\ref{eq.attainJana}), the first inequality stems from total probability theorem, the second inequality comes from (\ref{eq.attainJ_gana}). From (\ref{eq.attainJana}), we find out that $J(\bm{\nu})>J(\bm{0})$ when $\bm{\nu}>(2^{\frac{1}{4}}+1){\sqrt{p}}\bm{1}_{L\times1}$. Therefore, the unique minimizer of the function $J$ must occur in the interval $[\bm{0},(2^{\frac{1}{4}}+1){\sqrt{p}}\bm{1}_{L\times1}]$.
\end{proof}
\begin{prop}\label{prop.uniqnessoptimalweights}
Let $\bm{x}\in \mathbb{R}^p$ be a non-uniform $s$-sparse vector in $\bm{\Omega}\in\mathbb{R}^{p\times n}$ with parameters $\{\rho_i\}_{i=1}^L$ and $\{\alpha_i\}_{i=1}^L$. Then there exist unique optimal weights $\bm{\omega}^* \in R_{+}^L$ that minimize $\hat{m}_{p,s,\bm{w}}$ with $\bm{w}=\bm{D\omega} \in \mathbb{R}^p$. Moreover, the optimal weights are obtained via the following integral equations.
\begin{align}\label{eq.optimalweightsl1analysis}
\alpha_i\omega^*_i+(1-\alpha_i)\phi'(\omega^*_i)=0~:~i=1,..., L.
\end{align}
\end{prop}
\begin{proof}
Define $\mathcal{C}:=\partial\|.\|_{1}(\bm{\Omega x})$ and use Lemma \ref{lemma.l1weightedanaupper} and (\ref{eq.l1subdiff}) to obtain:
\begin{align}
&\inf_{\bm{\omega}\in\mathbb{R}_{+}^L}\hat{m}_{p,s,\bm{D\omega}}=\kappa^2(\bm{\Omega})\inf_{\bm{\omega}\in\mathbb{R}_{+}^L}\underset{{t\in\mathbb{R}_{+}}}{\inf}\Psi_{t,\bm{D\omega}}(\sigma,\bm{\rho},\bm{\alpha})+\frac{1}{p}\nonumber\\
&=\inf_{\bm{\nu}\in\mathbb{R}_{+}^L}\bigg[\kappa^2(\bm{\Omega})J(\bm{\nu})+\frac{1}{p}\bigg],
\end{align}
where, $\Psi_{t,\bm{D\omega}}(\sigma,\bm{\rho},\bm{\alpha})$ is defined in (\ref{eq.weightedanalyisupper}). Also, we used a change of variable $\bm{\upsilon}=t\bm{w}$ to convert the multivariate optimization problem to a single variable optimization problem. Thus, $J(\bm{\nu})$ is obtained via the following equation:
\begin{align}
J(\bm{\nu})=\sum_{i=1}^{L}\rho_i\big(\alpha_i(\nu(i)^2+1)+(1-\alpha_i)\phi(\nu(i))\big).
\end{align}
As proved in Lemma \ref{lemma.l1weightedanaupper}, $J(\bm{\nu})$ is continuous and strictly convex. Thus, the unique minimizer can be obtained using $\nabla J(\bm{\nu})=\bm{0}$ which leads to:
\begin{align}
\alpha_i\nu^*(i)+(1-\alpha_i)\phi'(\nu^*(i))=0~:~i=1,..., L.
\end{align}
\end{proof}
\begin{thm}\label{thm.l1weighted analysis}
Let $\bm{x}\in\mathbb{R}^n$ be a non-uniform $s$-sparse vector in $\bm{\Omega}\in\mathbb{R}^{p\times n}$ with parameters $\{\rho_i\}_{i=1}^L$ and $\{\alpha_i\}_{i=1}^L$. Then, number of measurements required for $\mathsf{P}_{1,\bm{\Omega},\bm{D}\bm{\omega}^*}$ is exactly equals the whole number of measurements required for $\mathsf{P}_{1}$ to recover each $\{(\bm{\Omega x})_{\mathcal{P}_i}\in\mathbb{R}^p\}_{i=1}^L$ separately up to an asymptotically negligible error term.
\end{thm}
\begin{proof}
As previously defined in (\ref{eq.mhat_psw}), with optimal weights, the upper bound for normalized number of measurements required for $\mathsf{P}_{1,\bm{\Omega},\bm{D}\bm{\omega}^*}$ to succeed is given by:
%
%\hat{m}_{p,s,w}=\kappa^2(\Omega)\underset{{t\ge0}}{\inf}\Psi_{t,w}(\sigma,\rho,\alpha)+\frac{1}{p},
%\end{align}
%with $\Psi_{t,w}(\sigma,\rho,\alpha)$ defined as:
%\begin{align}\label{eq.weightedanalyisupper}
%\Psi_{t,w}(\sigma,\rho,\alpha)=
%\sum_{i=1}^{L}\rho_i\bigg(\alpha_i(t^2\omega_i^2+1)+(1-\alpha_i)\phi(t\omega_i)\bigg)
%\end{align}
%
\begin{align}
&\hat{m}_{p,s,\bm{w}^*}=\inf_{\bm{\omega}\in\mathbb{R}^L}\hat{m}_{p,s,\bm{D\omega}}=\nonumber\\
&\sum_{i=1}^{L}\bigg[\kappa^2(\bm{\Omega})\nonumber\\
&\inf_{\nu_i\in\mathbb{R_{+}}}\underbrace{\bigg(\frac{\|(\bm{\Omega x})_{\mathcal{P}_i}\|_0}{p}(\nu_i^2+1)+(1-\frac{\|(\bm{\Omega x})_{\mathcal{P}_i}\|_0}{p})\phi(\nu_i)\bigg)}_{\Psi_{\nu_i,\|(\bm{\Omega x})_{\mathcal{P}_i}\|_0}(\frac{\|(\bm{\Omega x})_{\mathcal{P}_i}\|_0}{p})}+\frac{1}{p}\bigg]\nonumber\\
&+\frac{1-L}{p}=\sum_{i=1}^{L}\hat{m}_{p,\|(\bm{\Omega x})_{P_i}\|_0}+\frac{1-L}{p}\cdot
\end{align}
The expression in the bracket is exactly the upper bound for normalized number of measurements required for successful recovery of $(\bm{\Omega x})_{\mathcal{P}_i}\in\mathbb{R}^p$ using $\mathsf{P}_{1,\bm{\Omega}}$ i.e. $\hat{m}_{p,\|(\bm{\Omega x})_{\mathcal{P}_i}\|_0}$. Thus, regarding the error bounds obtained in Propositions \ref{prop.analysis} and \ref{prop.weightedanalysis},  the relation between $m_{p,s,\bm{D}\bm{\omega}^*}$ and $m_{p,\|(\bm{\Omega x})_{\mathcal{P}_i}\|_0}$ is given by:
\begin{align}
\frac{1-L}{p}-\frac{2\kappa^2(\bm{\Omega})}{\sqrt{pL}}\le m_{p,s,\bm{D}\bm{\omega}^*}-\sum_{i=1}^{L}m_{p,\|(\bm{\Omega x})_{\mathcal{P}_i}\|_0}\nonumber\\
\le\frac{2\kappa^2(\bm{\Omega})}{\sqrt{p}}\sum_{i=1}^{L}(\|(\bm{\Omega x})_{\mathcal{P}_i}\|_0)^{-\frac{1}{2}}+\frac{1-L}{p}.
\end{align}
\end{proof}
\section{Block Sparsity}\label{section.l12block}
In this section, the following questions are investigated about a non-uniform block sparse model:
\begin{enumerate}
  \item How many measurements is required for $\mathsf{P}_{1,2}$ and $\mathsf{P}_{1,2,\bm{w}}$ to successfully recover
an $s$-block sparse vector from independent Gaussian linear measurements?
  \item Given extra prior information, what is the optimal choice of weights in $\mathsf{P}_{1,2,\bm{w}}$?
\end{enumerate}
In what follows in this section, we precisely answer these questions in two subsections.
\subsection{Number of Measurement for successful Recovery}
In this subsection, regarding Theorem (\ref{thm.Pfmeasurement}) for $\|.\|_{1,2}$ and $\|.\|_{1,2,\bm{w}}$, we denote normalized number of measurements required for $\mathsf{P}_{1,2}$ and $\mathsf{P}_{1,2,\bm{w}}$ to recover an $s$-block sparse vector with probability $1-\eta$ by $m_{q,s}$ and $m_{q,s,\bm{w}}$, respectively which are defined as:
\begin{align}
m_{q,s}:=\frac{\delta(\mathcal{D}(\|.\|_{1,2},\bm{x}))}{q}.
\end{align}
\begin{align}
m_{q,s,\bm{w}}:=\frac{\delta(\mathcal{D}(\|.\|_{1,2,\bm{w}},\bm{x}))}{q}.
\end{align}
In the following, we obtain upper bounds for number of measurements required for $\mathsf{P}_{1,2}$ and $\mathsf{P}_{1,2,\bm{w}}$ to succeed with probability $1-\eta$.
\begin{lem}\label{lemma.mhat qs}
Let $\bm{x}\in \mathbb{R}^n$ be an $s$-block sparse vector. Then an upper bound for normalized number of measurements required for $\mathsf{P}_{1,2}$ to succeed (i.e. $m_{q,s}$) is given by:
\begin{align}
\hat{m}_{q,s}=\underset{{t\ge0}}{\inf}\Psi_t(\sigma),
\end{align}
with $\Psi_t(\sigma)$ defined as:
\begin{align}
\Psi_t(\sigma)=\sigma(k+t^2)+\frac{(1-\sigma)}{2^{\frac{k}{2}-1}\Gamma(\frac{k}{2})}\phi_B(t),
\end{align}
where $k=\frac{n}{q}$, $\phi_B(z):=\int_{z}^{\infty}(u-z)^2u^{k-1}\exp(-\frac{u^2}{2})du$ and $\sigma:=\frac{\|\bm{x}\|_{0,2}}{q}$
\end{lem}
\begin{lem}\label{lemma.mhat qsw}
Let $\bm{x}\in \mathbb{R}^n$ be a non-uniform $s$-block sparse vector in $\mathbb{R}^{n}$ with parameters $\{\rho_i\}_{i=1}^L$ and $\{\alpha_i\}_{i=1}^L$. Then an upper bound for normalized number of measurements required for $\mathsf{P}_{1,2,\bm{w}}$ with $\bm{w}=\bm{D}\bm{\omega}~\in\mathbb{R}^q$ to succeed (i.e. $m_{q,s,\bm{w}}$) is given by:
\begin{align}\label{eq.hatm-qsw}
\hat{m}_{q,s,\bm{w}}=\underset{{t\ge0}}{\inf}\Psi_{t,\bm{w}}(\sigma,\bm{\rho},\bm{\alpha}),
\end{align}
with $\Psi_{t,\bm{w}}(\sigma,\bm{\rho},\bm{\alpha})$ defined as:
\begin{align}
\Psi_{t,\bm{w}}(\sigma,\bm{\rho},\bm{\alpha})=
\sum_{i=1}^{L}\rho_i\bigg(\alpha_i(k+t^2\omega_i^2)+\frac{(1-\alpha_i)}{2^{\frac{k}{2}-1}\Gamma(\frac{k}{2})}\phi_B(t\omega_i)\bigg).
\end{align}
\end{lem}
\begin{proof}
With the same reasoning as in (\ref{eq.upforstatistical}) and the definition of statistical dimension for $\mathcal{D}(\|.\|_{1,2,\bm{w}},\bm{x})$ we have:
\begin{align}\label{eq.mhatqsw}
m_{q,s,\bm{w}}\le \inf_{t\ge0}\underbrace{{q^{-1}}\mathds{E}\mathrm{dist}^2(\bm{g},t\partial\|.\|_{1,2,\bm{w}}( \bm{x}))}_{\Psi_{t,\bm{w}}(\sigma,\bm{\rho},\bm{\alpha})}:=\hat{m}_{q,s,\bm{w}}.
\end{align}
The next step is to calculate $\partial\|.\|_{1,2,\bm{w}}(\bm{x})$. From Proposition \ref{prop.simplerform of subdiff}, we have:
\begin{align}\label{eq.subw12asli}
\partial\|.\|_{1,2,\bm{w}}(\bm{x})=\{\bm{z}\in \mathbb{R}^p:~\langle \bm{z}, \bm{x}\rangle=\| \bm{x}\|_{1,2,\bm{w}},~\|\bm{z}\|_{1,2,\bm{w}}^*=1\}.
\end{align}
From the first part of (\ref{eq.subw12asli}), we obtain:
\begin{align}
\sum_{b\in \mathcal{B}}\langle \bm{z}_{\mathcal{V}_b},\bm{x}_{\mathcal{V}_b}\rangle=\sum_{b\in \mathcal{B}}w_b\|\bm{x}_{\mathcal{V}_b}\|_2,
\end{align}
which then we reach $\bm{z}_{\mathcal{V}_b}=\frac{w_b\bm{x}_{\mathcal{V}_b}}{\|\bm{x}_{\mathcal{V}_b}\|_2}~:~\forall b\in \mathcal{B}$. To compute the dual of $\|.\|_{1,2,\bm{w}}$ i.e. $\|.\|_{1,2,\bm{w}}^*$, we have:
\begin{align}\label{eq.dualnorm12}
&\|\bm{z}\|_{1,2,\bm{w}}^*=\sup_{\|\bm{y}\|_{1,\bm{w}}\le1}\langle \bm{z},\bm{y} \rangle=\nonumber\\
&\sup_{\sum_{b=1}^{q}w_b\|\bm{y}_{\mathcal{V}_b}\|_{2}\le1}\sum_{b=1}^{q}\langle \bm{z}_{\mathcal{V}_b},\bm{y}_{\mathcal{V}_b}\rangle \le\nonumber\\
&\sup_{\sum_{b=1}^{q}w_b\|\bm{y}_{\mathcal{V}_b}\|_{2}\le1}\sum_{b=1}^{q}w_b^{-1}\|\bm{z}_{\mathcal{V}_b}\|_2w_b\|\bm{y}_{\mathcal{V}_b}\|_2 \nonumber\\
&\le \max_{b\in[q]}w_b^{-1}\|\bm{z}_{\mathcal{V}_b}\|_2,
\end{align}
where in, the first and second inequality comes from H\"older's inequality with equality when
\[
    \bm{y} = \left\{\begin{array}{lr}
        \frac{\bm{z}_{\mathcal{V}_c}}{\|\bm{z}_{\mathcal{V}_c}\|_2}, &  c=\underset{b\in[q]}{\arg\max}~~ {w_b}^{-1}\|\bm{z}_{\mathcal{V}_b}\|_2\\
        0, & \mathrm{o.w.}
        \end{array}\right\}\in \mathbb{R}^n.
  \]
From (\ref{eq.dualnorm12}) and (\ref{eq.subw12asli}), we find out that $w_b^{-1}\|\bm{z}_{\mathcal{V}_b}\|_2\le1~:~\forall b\in[q]$. Hence,
\begin{align}\label{eq.l12subdiff}
    \partial\|.\|_{1,2,\bm{w}}(\bm{x}) = \left\{\bm{z}\in\mathbb{R}^n:\begin{array}{lr}
        \frac{w_b~\bm{x}_{\mathcal{V}_b}}{\|\bm{x}_{\mathcal{V}_b}\|_2}, &  b\in \mathcal{B}\\
        \|\bm{z}_{\mathcal{V}_b}\|_2\le w_b, & b\in {\mathcal{\bar{B}}}
        \end{array}\right\}.
  \end{align}
Now to calculate $\Psi_{t,\bm{w}}(\sigma,\bm{\rho},\bm{\alpha})$, regarding (\ref{eq.l12subdiff}), we compute the distance of the dilated subdifferntial of descent cone of $\ell_{1,2,\bm{w}}$ norm at $\bm{x}\in\mathbb{R}^n$ from a standard Gaussian vector $\bm{g}\in\mathbb{R}^n$ which is given by:
\begin{align}\label{eq.distsubdiff12}
&\mathrm{dist}^2(\bm{g},t\partial\|.\|_{1,2,\bm{w}}(\bm{x}))=\inf_{\bm{z}\in\partial\|.\|_{1,2,\bm{w}}(\bm{x})}\|\bm{g}-t\bm{z}\|_2^2=\nonumber\\
&\sum_{b\in \mathcal{B}}{\|\bm{g}_{\mathcal{V}_b}-tw_b\frac{\bm{x}_{\mathcal{V}_b}}{\|\bm{x}_{\mathcal{V}_b}\|_2}\|_2^2}+\sum_{b\in {{\mathcal{\bar{B}}}}}\inf_{\|\bm{z}_{\mathcal{V}_b}\|_2\le w_b}{\|\bm{g}_{\mathcal{V}_b}-t\bm{z}_{\mathcal{V}_b}\|_2^2}=\nonumber\\
&\sum_{b\in \mathcal{B}}{\|\bm{g}_{\mathcal{V}_b}-tw_b\frac{\bm{x}_{\mathcal{V}_b}}{\|\bm{x}_{\mathcal{V}_b}\|_2}\|_2^2}+\sum_{b\in {{\mathcal{\bar{B}}}}}{(\|\bm{g}_{\mathcal{V}_b}\|_2-tw_b)_{+}^2},
\end{align}
where we used triangle inequality in the second part. By taking expectation from both sides, we reach:
\begin{align}\label{eq.12Edist2}
&\mathds{E}\mathrm{dist}^2(\bm{g},t\partial\|.\|_{1,2,\bm{w}}(\bm{x}))=\nonumber\\
&ks+\sum_{b\in \mathcal{B}}(tw_b)^2+\sum_{b\in {\mathcal{\bar{B}}}}\mathds{E}(\underbrace{\|\bm{g}_{\mathcal{V}_b}\|_2}_{\zeta}-tw_b)_{+}^2,
\end{align}
where, $k=\frac{n}{q}$. and $\zeta^2:=\|\bm{g}_{\mathcal{V}_b}\|_2^2$ has chi-square distribution with $k$ degrees of freedom. Moreover,
\begin{align}\label{eq.Ezetal12}
&\mathds{E}(\zeta-tw_b)_{+}^2=2\int_{0}^{\infty}a\mathds{P}(\zeta^2\ge(a+tw_b)^2)da=\nonumber\\
&\frac{2}{2^{\frac{k}{2}}\Gamma(\frac{k}{2})}\int_{0}^{\infty}\int_{(tw_b+a)^2}^{\infty}au^{\frac{k}{2}-1}e^{-\frac{u}{2}}du~da\nonumber\\
&\frac{2}{2^{\frac{k}{2}}\Gamma(\frac{k}{2})}\int_{(tw_b)^2}^{\infty}\int_{0}^{\sqrt{u}-tw_b}au^{\frac{k}{2}-1}e^{-\frac{u}{2}}da~du\nonumber\\
&\frac{1}{2^{\frac{k}{2}-1}\Gamma(\frac{k}{2})}\int_{tw_b}^{\infty}(u-tw_b)^2u^{k-1}e^{-\frac{u^2}{2}}du\nonumber\\
&:=\frac{1}{2^{\frac{k}{2}-1}\Gamma(\frac{k}{2})}\phi_B(tw_b),
\end{align}
where in the third line, the order of integration is changed and in the forth line, a change of variable is used. As a consequence, (\ref{eq.12Edist2}) becomes:
\begin{align}\label{eq.12Edist2asli}
&\mathds{E}\mathrm{dist}^2(\bm{g},t\partial\|.\|_{1,2,\bm{w}}(\bm{x}))=\nonumber\\
&ks+\sum_{b\in \mathcal{B}}(tw_b)^2+\frac{1}{2^{\frac{k}{2}-1}\Gamma(\frac{k}{2})}\sum_{b\in {\mathcal{\bar{B}}}}\phi_B(tw_b).
\end{align}
By normalizing to the number of blocks $q$ and incorporating block prior information using $\bm{w}=\bm{D\omega} \in\mathbb{R}^n$ we reach:
\begin{align}
&\mathds{E}\mathrm{dist}^2(\bm{g},t\partial\|.\|_{1,2,\bm{w}}(\bm{x}))=\nonumber\\
&ks+\sum_{i=1}^{L}|\mathcal{P}_i\cap \mathcal{B}|t^2\omega_i^2+|\mathcal{P}_i\cap {\mathcal{\bar{B}}}|\frac{1}{2^{\frac{k}{2}-1}\Gamma(\frac{k}{2})}\phi_B(t\omega_i)
\nonumber\\
&=q\bigg(\sigma+\sum_{i=1}^{L}\rho_i\big(\alpha_it^2\omega_i^2+\frac{1}{2^{\frac{k}{2}-1}\Gamma(\frac{k}{2})}(1-\alpha_i)\phi_B(t\omega_i)\big)\bigg)\nonumber\\
&=q\bigg(\sum_{i=1}^{L}\rho_i\big(\alpha_i(t^2\omega_i^2+k)+(1-\alpha_i)\frac{1}{2^{\frac{k}{2}-1}\Gamma(\frac{k}{2})}\phi_B(t\omega_i)\big)\bigg),\label{eq.Edist2part2}
\end{align}
where in the last line above, we benefited the fact that $\sigma=\sum_{i=1}^{L}\rho_i\alpha_i$.
\end{proof}
\begin{proof}[Proof of Lemma \ref{lemma.mhat qs}]
 We find an upper bound for $m_{q,s}$. The procedure is exactly the same as the proof of Lemma \ref{lemma.mhat qsw} with the assigned weight to each block set to one i.e. $w=1\in \mathbb{R}^q$. In fact, we have: $\Psi_t(\sigma)=\Psi_{t,\bm{w}}(\sigma,\bm{\rho},\bm{\alpha})$. By replacing $\bm{w}=\bm{1}\in \mathbb{R}^q$ in (\ref{eq.12Edist2}) and (\ref{eq.mhatqsw}), we reach $\hat{m}_{q,s}$ in Lemma \ref{lemma.mhat qs}.
\end{proof}
In the following Propositions, we prove that the obtained upper bounds in Lemmas \ref{lemma.mhat qs} and \ref{lemma.mhat qsw} for normalized number of measurements are asymptotically tight.
\begin{prop}\label{prop.errorfor mhat qs}. Normalized number of Gaussian linear measurements required for $\mathsf{P}_{1,2}$ to succeed (i.e. $m_{q,s}$) satisfies the following error bound.
\begin{align}\label{eq.l1analysisbound}
\hat{m}_{q,s}-\frac{2}{\sqrt{sq}}\le m_{q,s}\le \hat{m}_{q,s}.
\end{align}
\end{prop}
\begin{proof}
By the error bound in (\ref{eq.errorbound}) and also (\ref{eq.l12subdiff}) with $\bm{w}=\bm{1}\in \mathbb{R}^q$, we obtain the numerator of the error bound (\ref{eq.errorbound}) as:
\begin{align}
2\sup_{\bm{s}\in\partial\|.\|_{1,2}(\bm{x})}\|\bm{s}\|_2= 2\sup_{\bm{s}\in\partial\|.\|_{1,2}(\bm{x})}\sqrt{\sum_{b=1}^{q}\|\bm{s}_{\mathcal{V}_b}\|_2^2}\le 2\sqrt{q}.
\end{align}
Also, for the denominator we have:
\begin{align}\label{eq.denominatorblock}
\bigg\|\frac{\bm{x}}{\|\bm{x}\|_2}\bigg\|_{1,2}\le \sqrt{s}=\sqrt{\|\bm{x}\|_{0,2}}.
\end{align}
The error bound in (\ref{eq.errorbound}) for $\|.\|_{1,2}$ depends only on $\mathcal{D}(\|.\|_{1,2},\bm{x})$. Moreover, $\mathcal{D}(\|.\|_{1,2},\bm{x})$ only requires that $\|\bm{x}_{\mathcal{V}_b}\|_2=1~:~\forall b\in \mathcal{B}$. So a vector
\begin{align}
    \bm{z} = \left\{\begin{array}{lr}
        \|\bm{z}_{\mathcal{V}_b}\|_2=1, &  b\in \mathcal{B}\\
        0, & b\in {\mathcal{\bar{B}}}
        \end{array}\right\}\in \mathbb{R}^n,\nonumber
  \end{align}
   can be chosen to have equality in (\ref{eq.denominatorblock}). Therefore, the error of obtaining $m_{q,s}$ is at most $\frac{1}{q}\frac{2\sqrt{q}}{\sqrt{s}}=\frac{2}{\sqrt{sq}}$.
\end{proof}
\begin{prop}\label{prop.error for mhat qsw}
Normalized number of Gaussian linear measurements required for $\mathsf{P}_{1,2,\bm{w}}$ to successfully recover a non-uniform $s$-block sparse vector in $\mathbb{R}^n$ with parameters $\{\rho_i\}_{i=1}^L$ and $\{\alpha_i\}_{i=1}^L$ (i.e. $m_{q,s,\bm{w}}$) satisfies the following error bound.
\begin{align}\label{eq.l12w error}
\hat{m}_{q,s,\bm{w}}-\frac{2}{\sqrt{qL}}\le m_{q,s,\bm{w}}\le \hat{m}_{q,s,\bm{w}}.
\end{align}
\end{prop}
It is interesting that the error bound in Proposition \ref{prop.errorfor mhat qs} is a special case of the error bound of Proposition \ref{prop.error for mhat qsw} where one has $s$ sets of blocks with size $\frac{q}{L}$ and knows with probability $\frac{L}{q}$ that each set of blocks contributes to the block support.
\begin{proof}
By the error bound (\ref{eq.errorbound}) and (\ref{eq.l12subdiff}), with $f(\bm{x})=\|\bm{x}\|_{1,2,\bm{w}}$ and $\bm{w}=\sum_{i=1}^{L}\omega_i\bm{1}_{\mathcal{P}_i} \in \mathbb{R}^q$, the numerator of (\ref{eq.errorbound}) is given by:
\begin{align}
2\sup_{\bm{s}\in\partial\|.\|_{1,2,\bm{w}}(\bm{x})}\|\bm{s}\|_2\le 2\sqrt{\sum_{i=1}^{q}w_b^2}=\nonumber\\
2\sqrt{\sum_{i=1}^{L}|\mathcal{P}_i|\omega_i^2}=2\sqrt{\sum_{i=1}^{L}q\rho_i\omega_i^2}.\nonumber
\end{align}
Also, for the denominator we have:
\begin{align}\label{eq.denomweight12}
\frac{\|\bm{x}\|_{1,2,\bm{w}}}{{\|\bm{x}\|_2}}\le \sqrt{\sum_{i\in\mathcal{B}}w_i^2}=\sqrt{\sum_{i=1}^{L}|\mathcal{P}_i\cap \mathcal{B}|\omega_i^2}=\sqrt{\sum_{i=1}^{L}q\alpha_i\rho_i\omega_i^2},
\end{align}
where the first inequality in (\ref{eq.denomweight12}) comes from Cauchy Schwartz inequality.
With the same justification as in the proof of Proposition \ref{prop.errorfor mhat qs} and the fact that
$\mathcal{D}(\|.\|_{1,2,\bm{w}},\bm{x})$ only requires that $\|\bm{x}_{\mathcal{V}_b}\|_2=w_b~:~\forall b\in \mathcal{B}$ a vector
\begin{align}
    \bm{z} = \left\{\begin{array}{lr}
        \|\bm{z}_{\mathcal{V}_b}\|_2=w_b, &  b\in \mathcal{B}\\
        0, & b\in {\mathcal{\bar{B}}}
        \end{array}\right\}\in \mathbb{R}^n,\nonumber
  \end{align}
can be chosen to have equality in (\ref{eq.denomweight12}). Therefore, the error of obtaining the upper bound of $m_{q,s,\bm{w}}$ i.e. $\hat{m}_{q,s,\bm{w}}$ is:
\begin{align}
   \frac{2\sqrt{\sum_{i=1}^{L}q\rho_i\omega_i^2}}{q\sqrt{\sum_{i=1}^{L}q\alpha_i\rho_i\omega_i^2}}\le
   \frac{2}{q}\sqrt{\frac{1}{\underset{{i\in [q]}}\min\frac{|\mathcal{P}_i\cap \mathcal{B}|}{|\mathcal{P}_i|}}}
   \le \frac{2}{\sqrt{qL}},
\end{align}
in which, the last inequality follows from the facts that $|\mathcal{P}_i\cap \mathcal{B}|\ge1$, $|\mathcal{P}_i|\le\frac{q}{L}$ for at least one $i\in[q]$ and thus $ \underset{{i\in [q]}}\min\frac{|\mathcal{P}_i\cap \mathcal{B}|}{|\mathcal{P}_i|}\ge \frac{L}{q}$. Further, the error of $\hat{m}_{q,s,\bm{w}}$ in (\ref{eq.mhatqsw}) is at most $\frac{2}{\sqrt{qL}}$
\end{proof}
\subsection{Optimal Weights}
Infimum of (\ref{eq.l12w error}) gives:
\begin{align}\label{eq.infmhatqserrorbound}
\inf_{\bm{\omega}\in\mathbb{R}_{+}^L}\hat{m}_{q,s,\bm{D\omega}}-\frac{2}{\sqrt{qL}}\le\inf_{\bm{\omega}\in\mathbb{R}_{+}^L}{m}_{q,s,\bm{D\omega}}\le \inf_{\bm{\omega}\in\mathbb{R}_{+}^L}\hat{m}_{q,s,\bm{D\omega}},
\end{align}
where $\bm{D}:=[\bm{1}_{\mathcal{V}_1},..., \bm{1}_{\mathcal{V}_q}]_{n\times q}[\bm{1}_{\mathcal{P}_1},..., \bm{1}_{\mathcal{P}_L}]_{q\times L}$.
In weighted block sparsity, we call the weight $\bm{\omega}^*=\underset{\bm{\omega}\in\mathbb{R}_{+}^L}{\arg\min}~~\hat{m}_{q,s,\bm{D\omega}}\in\mathbb{R}_{+}^L$ optimal since it asymptotically minimizes number of measurements required for $\mathsf{P}_{1,2,\bm{w}}$ to succeed.
Similar to previous section, before proving the uniqueness of optimal weights for $\mathsf{P}_{1,2,\bm{D\omega}}$, we state the following lemma.
\begin{lem}{\label{lemma.Jb(nu)}}
Let $\mathcal{C}:=\partial\|.\|_{1,2}(\bm{x})$. Suppose that $\mathcal{C}$ does not contain the origin. In particular, it is compact and there are upper and lower bounds that satisfy $1\le\|\bm{z}\|_2\le \sqrt{q}$ for all $\bm{z}\in \mathcal{C}$. Also denote the standard normal vector by $\bm{g}\in \mathbb{R}^n$. Consider the function
\begin{align}
J(\bm{\nu}):=\mathds{E}\mathrm{dist}^2(\bm{g},\bm{\upsilon}\odot \mathcal{C})=\mathds{E}[J_{\bm{g}}(\bm{\nu})] \nonumber\\
 \text{with}~~\bm{\upsilon}=\bm{D\nu} \in \mathbb{R}^n ~~\text{for}~~ \bm{\nu}\in\mathbb{R}_{+}^L.
\end{align}
The function $J$ is strictly convex, continuous at $\bm{\nu}\in \mathbb{R}_{+}^L$ and differentiable for $\bm{\nu}\in \mathbb{R}_{++}^L$. More over, there exists a unique point that minimize $J$.
\end{lem}
\begin{proof}
\textit{Continuity in bounded points}.
We must show that sufficiently small changes in $\bm{\nu}$ result in arbitrary small changes in $J(\bm{\nu})$. By definition of $J_{\bm{g}}(\bm{\nu})$, we have:
\begin{align}
&J_{\bm{g}}(\bm{\nu})-J_{\bm{g}}(\tilde{\bm{\nu}})=\|\bm{g}-\mathcal{P}_{\bm{\upsilon}\odot \mathcal{C}}(\bm{g})\|_2^2-\|\bm{g}-\mathcal{P}_{\tilde{\bm{\upsilon}}\odot \mathcal{C}}(\bm{g})\|_2^2\nonumber\\&=2\langle \bm{g},\mathcal{P}_{\tilde{\bm{\upsilon}}\odot \mathcal{C}}(\bm{g})-\mathcal{P}_{{\bm{\upsilon}}\odot \mathcal{C}}(\bm{g})\rangle+\nonumber\\
&\big(\|\mathcal{P}_{\bm{\upsilon}\odot \mathcal{C}}(\bm{g})\|_2-\|\mathcal{P}_{\tilde{\bm{\upsilon}}\odot \mathcal{C}}(\bm{g})\|_2\big)\big(\|\mathcal{P}_{\bm{\upsilon}\odot \mathcal{C}}(\bm{g})\|_2+\|\mathcal{P}_{\tilde{\bm{\upsilon}}\odot \mathcal{C}}(\bm{g})\|_2\big).\nonumber\\
&\text{The absolute value satisfies:}\nonumber\\
&|J_{\bm{g}}(\bm{\nu})-J_{\bm{g}}(\tilde{\bm{\nu}})|\nonumber\\
&\le \bigg(2\|\bm{g}\|_2\sqrt{q}+q^2(\|\bm{\nu}\|_1+\|\tilde{\bm{\nu}}\|_1)\bigg)\|\tilde{\bm{\nu}}-\bm{\nu}\|_1,
\end{align}
where, we used the fact that,
\begin{align}
&\|\mathcal{P}_{\bm{\upsilon}\odot \mathcal{C}}(\bm{g})\|_2\le\sup_{\bm{z}\in \mathcal{C}}\|\bm{\upsilon}\odot \bm{z}\|_2\le\|\bm{D\nu}\|_{\infty}\sqrt{q}\le\nonumber\\
&\|\bm{\nu}\|_1\sqrt{q}\|\bm{D}\|_{1\rightarrow\infty}=\|\bm{\nu}\|_1\sqrt{q},
\end{align}
and,
\begin{align}
&\|\mathcal{P}_{\bm{\upsilon}\odot \mathcal{C}}(\bm{g})\|_2-\|\mathcal{P}_{\tilde{\bm{\upsilon}}\odot \mathcal{C}}(\bm{g})\|_2\le \sup_{\bm{z}\in \mathcal{C}}\bigg(\|\bm{\upsilon}\odot \bm{z}\|_2-\|\tilde{\bm{\upsilon}}\odot \bm{z}\|_2\bigg)\nonumber\\
&\le\sup_{\bm{z}\in \mathcal{C}}\big(\|(\bm{\upsilon}-\tilde{\bm{\upsilon}})\odot \bm{z}\|_2\big)\le
\|\bm{\nu}-\tilde{\bm{\nu}}\|_1\sqrt{q}\|\bm{D}\|_{1\rightarrow\infty}\nonumber\\
&=\|\bm{\nu}-\tilde{\bm{\nu}}\|_1\sqrt{q}.
\end{align}
As a consequence, we obtain:
\begin{align}
&|J(\bm{\nu})-J(\tilde{\bm{\nu}})|\nonumber\\
&\le \bigg(2\sqrt{nqL}+q\sqrt{L}(\|\bm{\nu}\|_1+\|\tilde{\bm{\nu}}\|_1)\bigg)\|\tilde{\bm{\nu}}-\bm{\nu}\|_2\rightarrow 0 \nonumber\\
 &~\text{as}~ \bm{\nu}\rightarrow~ \tilde{\bm{\nu}}.
\end{align}
Since $\|\bm{\nu}\|_1$ is bounded, continuity holds.\par
\textit{Convexity}. Let $\bm{\nu}~,\tilde{\bm{\nu}}\in\mathbb{R}_{+}^L$ and $\theta\in[0,1]$ with $\bm{\upsilon}=\bm{D\nu}$ and $\tilde{\bm{\upsilon}}=\bm{D}\tilde{\bm{\nu}}$. Then we have:
\begin{align}\label{helpconvexity}
&\forall \epsilon , \tilde{\epsilon}>0~\exists \bm{z} ,\tilde{\bm{z}}\in \mathcal{C} ~~\text{such that}\nonumber\\
&\|\bm{g}-\bm{\upsilon}\odot \bm{z}\|_2\le \mathrm{dist}(\bm{g},\bm{\upsilon}\odot \mathcal{C})+\epsilon\nonumber\\
&\|\bm{g}-\tilde{\bm{\upsilon}}\odot \tilde{\bm{z}}\|_2\le \mathrm{dist}(\bm{g},\tilde{\bm{\upsilon}}\odot \mathcal{C})+\tilde{\epsilon}.
\end{align}
Since otherwise we have:
\begin{align}
&\forall \bm{z},\tilde{\bm{z}}\in \mathcal{C}:\nonumber\\
&\|\bm{g}-\bm{\upsilon}\odot \bm{z}\|_2>\mathrm{dist}(\bm{g},\bm{\upsilon}\odot \mathcal{C})+\epsilon\nonumber\\
&\|\bm{g}-\tilde{\bm{\upsilon}}\odot \tilde{\bm{z}}\|_2> \mathrm{dist}(\bm{g},\tilde{\bm{\upsilon}}\odot \mathcal{C})+\tilde{\epsilon}.
\end{align}
By taking the infimum over $\bm{z},\tilde{\bm{z}}\in \mathcal{C}$, we reach a contradiction. We proceed to prove convexity of $\mathrm{dist}(\bm{g},(\bm{D\nu})\odot \mathcal{C})$.
\begin{align}\label{eq.convexitydist}
&\mathrm{dist}(\bm{g},(\theta\bm{\upsilon}+(1-\theta)\tilde{\bm{\upsilon}})\odot \mathcal{C})=\inf_{\bm{z} \in \mathcal{C}}\|\bm{g}-(\theta\bm{\upsilon}+(1-\theta)\tilde{\bm{\upsilon}})\odot \bm{z}\|_2\nonumber\\
&\le\inf_{\bm{z}_1\in \mathcal{C},\bm{z}_2\in \mathcal{C}}\|\bm{g}-\theta\bm{\upsilon}\odot \bm{z}_1-(1-\theta)\bm{\upsilon}\odot \bm{z}_2\|_2\le\nonumber\\
&\theta\|\bm{g}-\bm{\upsilon}\odot \bm{z}_1\|_2+(1-\theta)\|\bm{g}-\tilde{\bm{\upsilon}}\odot \bm{z}_2\|_2\le\nonumber\\
&\theta \mathrm{dist}(\bm{g},\bm{\upsilon}\odot \mathcal{C})+(1-\theta)\mathrm{dist}(\bm{g},\tilde{\bm{\upsilon}}\odot \mathcal{C})+\epsilon+\tilde{\epsilon}.
\end{align}
Since this holds for any $\epsilon$ and $\tilde{\epsilon}$, $\mathrm{dist}(\bm{g},(\bm{D\nu})\odot \mathcal{C})$ is a convex function. As the square of a non-negative convex function is convex, $J_{\bm{g}}(\bm{\nu})$ is a convex function. At last, the function $J(\bm{\nu})$ is the average of convex functions, hence is convex.
In (\ref{eq.convexitydist}), the first inequality comes from the fact that $\forall \bm{z}_1,\bm{z}_2 \in \mathcal{C}~~\exists \bm{z}\in \mathcal{C}$:
\begin{align}\label{eq.l12benefit}
&\theta\bm{\upsilon}\odot \bm{z}_1+(1-\theta)\tilde{\bm{\upsilon}}\odot \bm{z}_2=\nonumber\\
&\small{\left\{\bm{y}_{n\times1}:\begin{array}{lr}
       \bm{y}_{\mathcal{V}_b}=\big(\theta\bm{\upsilon}_{\mathcal{V}_b}+(1-\theta)\tilde{\bm{\upsilon}}_{\mathcal{V}_b}\big)\odot\frac{\bm{x}_{\mathcal{V}_b}}{\|\bm{x}_{\mathcal{V}_b}\|_2}, &  b\in \mathcal{B}\\
        \|\bm{y}_{\mathcal{V}_b}\|_2\le\theta\|\bm{\upsilon}\|_{\infty}\|{\bm{z}_1}_{\mathcal{V}_b}\|_2\nonumber\\
        +(1-\theta)\|\tilde{\bm{\upsilon}}\|_{\infty}\|{\bm{z}_2}_{\mathcal{V}_b}\|_2, & b\in {\mathcal{\bar{B}}}
        \end{array}\right\}}\nonumber\\
&\in\small{\left\{\bm{y}_{n\times1}:\begin{array}{lr}
       \bm{y}_{\mathcal{V}_b}=\big(\theta\bm{\upsilon}_{\mathcal{V}_b}+(1-\theta)\tilde{\bm{\upsilon}}_{\mathcal{V}_b}\big)\odot\frac{\bm{x}_{\mathcal{V}_b}}{\|\bm{x}_{\mathcal{V}_b}\|_2}, &  b\in \mathcal{B}\\
        \|\bm{y}_{\mathcal{V}_b}\|_2\le\nonumber\\
        \big(\theta\|\bm{\upsilon}\|_{\infty}+(1-\theta)\|\tilde{\bm{\upsilon}}\|_{\infty}\big)\|{\bm{z}}_{\mathcal{V}_b}\|_2, & b\in {\mathcal{\bar{B}}}
        \end{array}\right\}}\nonumber\\
&=(\theta \bm{\upsilon}+(1-\theta)\tilde{\bm{\upsilon}})\odot \bm{z}.
\end{align}
To verify (\ref{eq.l12benefit}), we argue by contradiction.
\begin{align}
&\forall \bm{z}\in \mathcal{C} ~\exists d\in {\mathcal{\bar{B}}}~~\text{such that}:\nonumber\\
&(\theta\|\bm{\upsilon}\|_{\infty}+(1-\theta)\|\tilde{\bm{\upsilon}}\|_{\infty}\big)\|{\bm{z}}_{\mathcal{V}_d}\|_2
<\nonumber\\
&\theta\|\bm{\upsilon}\|_{\infty}\|{\bm{z}_1}_{\mathcal{V}_d}\|_2+(1-\theta)\|\tilde{\bm{\upsilon}}\|_{\infty}\|{\bm{z}_2}_{\mathcal{V}_d}\|_2\nonumber\\
&\le \theta\|\bm{\upsilon}\|_{\infty}+(1-\theta)\|\tilde{\bm{\upsilon}}\|_{\infty}.
\end{align}
Then, by taking $\bm{z}_{\mathcal{V}_d}=\bm{e}_i\in\mathbb{R}^{k}$ for some $i\in[k]$, we reach a contradiction.
In the second inequality in (\ref{eq.convexitydist}), we used triangle inequality of norms. The third inequality uses the relation (\ref{helpconvexity}).\par
\textit{Strict convexity}. We show strict convexity by contradiction. If $J(\bm{\nu})$ were not strictly convex, there would be vectors $\bm{\nu},\tilde{\bm{\nu}}\in\mathbb{R}_{+}^L$ with $\bm{\upsilon}=\bm{D\nu}, \tilde{\bm{\upsilon}}=\bm{D}\tilde{\bm{\nu}}$ and $\theta\in (0,1)$ such that,
\begin{align}\label{eq.strictconvex}
\mathds{E}[J_{\bm{g}}(\theta\bm{\nu}+(1-\theta)\tilde{\bm{\nu}})]=\mathds{E}[\theta J_{\bm{g}}(\bm{\nu})+(1-\theta)J_{\bm{g}}(\tilde{\bm{\nu}})].
\end{align}
For each $\bm{g}$ in (\ref{eq.strictconvex}) the left-hand side is smaller than or equal to the right-hand side. Therefore, in (\ref{eq.strictconvex}), $J_{\bm{g}}(\theta\bm{\nu}+(1-\theta)\tilde{\bm{\nu}})$ and $\theta J_{\bm{g}}(\bm{\nu})+(1-\theta)J_{\bm{g}}(\tilde{\bm{\nu}})$ are almost surely equal (except at a measure zero set) with respect to Gaussian measure. Moreover, we have:
\begin{align}\label{eq.J012}
&J_{\bm{0}}(\theta\bm{\nu}+(1-\theta)\tilde{\bm{\nu}})=\mathrm{dist}^2(\bm{0},\big(\theta\bm{\upsilon}+(1-\theta)\tilde{\bm{\upsilon}}\big)\odot \mathcal{C})\nonumber\\
&\le \inf_{\bm{z}_1,\bm{z}_2\in \mathcal{C}}\|\theta \bm{\upsilon}\odot \bm{z}_1+(1-\theta)\tilde{\bm{\upsilon}}\odot \bm{z}_2\|_2^2\nonumber\\
&<\theta\inf_{\bm{z}_1\in \mathcal{C}}\|\bm{\upsilon}\odot \bm{z}_1\|_2^2+(1-\theta)\inf_{\bm{z}_2\in \mathcal{C}}\|\tilde{\bm{\upsilon}}\odot \bm{z}_2\|_2^2\nonumber\\
&=\theta J_{\bm{0}}(\bm{\nu})+(1-\theta)J_{\bm{0}}(\tilde{\bm{\nu}})
\end{align}
where, the first inequality comes from (\ref{eq.l12benefit}) and the second inequality stems from the strict convexity of $\|.\|_2^2$. From (\ref{eq.l12benefit}), it is easy to verify that the set $\bm{\nu}\odot \mathcal{C}$ is a convex set. The distance to a convex set e.g. $\mathcal{E}$ i.e. $\mathrm{dist}(\bm{g},\mathcal{E})$ is a $1$-lipschitz function (i.e. $|\mathrm{dist}(\bm{g},\mathcal{E})-\mathrm{dist}(\tilde{\bm{g}},\mathcal{E})|\le\|\bm{g}-\tilde{\bm{g}}\|_2~:~\forall~\bm{g},\tilde{\bm{g}}\in\mathbb{R}^n$) and hence continuous with respect to $\bm{g}$. Therefore, $J_{\bm{g}}(\bm{\nu})$ is continuous with respect to $\bm{g}$. So there exist an open ball around $\bm{g}=\bm{0}\in\mathbb{R}^n$ that similar to (\ref{eq.J012}), we may write the following relation for some $\epsilon>0$
\begin{align}
&\exists \bm{u}\in\mathds{B}_{\epsilon}^n: \nonumber\\
&J_{\bm{u}}(\theta\bm{\nu}+(1-\theta)\tilde{\bm{\nu}})<\theta J_{\bm{u}}(\bm{\nu})+(1-\theta)J_{\bm{u}}(\tilde{\bm{\nu}})
\end{align}
The above statement contradicts with (\ref{eq.strictconvex}) and hence we have strict convexity. Continuity along with convexity of $J$ implies that $J$ is convex on the whole domain $\bm{\nu}\in\mathbb{R}_{+}^L$.\par
\textit{Differentiability}. The function $J_{\bm{g}}(\bm{\nu})$ is continuously differentiable and the gradient for $\bm{\nu}\in\mathbb{R}_{++}^L$ is:
\begin{align}\label{eq.diffJg}
&\nabla_{\bm{\nu}}J_{\bm{g}}(\bm{\nu})=\nonumber\\
&\frac{\partial J_{\bm{g}}(\bm{\nu})}{\partial\bm{\nu}}=-2\bm{D}^T(\bm{D}\bm{\nu})^{-1\odot}\odot\mathcal{P}_{(\bm{D}\bm{\nu})\odot\mathcal{C}}(\bm{g})\nonumber\\
&\odot(\bm{g}-\mathcal{P}_{\bm{D\nu}\odot\mathcal{C}}(\bm{g}))
\end{align}
Continuity of $\frac{\partial J_{\bm{g}}(\bm{\nu})}{\partial\bm{\nu}}$ at $\bm{\nu}\in\mathbb{R}_{+}^L$ stems from the fact that the projection onto a convex set is continuous. For each compact set $\mathcal{I}\subseteq\mathbb{R}_{+}^L$ we have:
\begin{align}
&\mathds{E}\sup_{\bm{\nu}\in\mathcal{I}}\|\nabla_{\bm{\nu}} J_{\bm{g}}(\bm{\nu})\|_2\le\nonumber\\
&2\|\bm{D}\|_{2\rightarrow2}\sqrt{q}(\sqrt{n}+2\sqrt{q}\big(\sup_{\bm{\nu}\in\mathcal{I}}\nu_{\max}\big))<\infty
\end{align}
where $\nu_{\max}:=\underset{i\in[L]}{\max}~\nu(i)$. Therefore, we have:
\begin{align}
\nabla_{\bm{\nu}} J(\bm{\nu})=\big(\frac{\partial}{\partial\bm{\nu}}\big)\mathds{E} J_{\bm{g}}(\bm{\nu})=\mathds{E}[\nabla_{\bm{\nu}} J_{\bm{g}}(\bm{\nu})] ~:~\forall \bm{\nu}\in\mathbb{R}_{+}^L
\end{align}
where in the last equality, we used the Lebesgue's dominated convergence theorem. \par
\textit{Attainment of the minimum}. Suppose that $\bm{\nu}\ge{\|\bm{g}\|_2}\bm{1}_{L\times 1}$. With this assumption we may write:
\begin{align}\label{eq.attain}
&\mathrm{dist}(\bm{g},(\bm{D\nu})\odot \mathcal{C})=\inf_{\bm{z}\in \mathcal{C}}\|\bm{g}-\bm{\upsilon}\odot \bm{z}\|_2\ge\nonumber\\
&\inf_{\bm{z}\in \mathcal{C}}(\|\bm{\upsilon}\odot \bm{z}\|_2-\|\bm{g}\|_2)\ge\nu_{\min}-\|\bm{g}\|_2\ge0,
\end{align}
where in (\ref{eq.attain}), $\nu_{\min}:=\underset{i\in[L]}{\min}~\nu(i)$. By squaring (\ref{eq.attain}), we reach:
\begin{align}\label{eq.attainJ_g}
J_{\bm{g}}(\bm{\nu})\ge (\nu_{\min}-\|\bm{g}\|_2)^2~~:~\forall \nu>{\|\bm{g}\|_2}\bm{1}_{L\times 1}
\end{align}
Using the relation $\mathds{E}\|\bm{g}\|_2\ge\frac{n}{\sqrt{n+1}}$ and Marcov's inequality we obtain:
\begin{align}
\mathds{P}(\|\bm{g}\|_2\le\sqrt{n})\ge1-\sqrt{\frac{n}{n+1}}\nonumber.
\end{align}
Then we reach:
\begin{align}\label{eq.attainJ}
&J(\bm{\nu})\ge\mathds{E}[J_{\bm{g}}(\bm{\nu})|\|\bm{g}\|_2\le\sqrt{n}]\mathds{P}(\|\bm{g}\|_2\le\sqrt{n})\nonumber\\
&\ge(1-\sqrt{\frac{n}{n+1}})\mathds{E}\big[(\nu_{\min}-\|\bm{g}\|_2)^2|\|\bm{g}\|_2\le\sqrt{n}\big]\nonumber\\
&\ge(1-\sqrt{\frac{n}{n+1}})(\nu_{\min}-\sqrt{n})^2,
\end{align}
where in (\ref{eq.attainJ}), the first inequality stems from total probability theorem, the second inequality comes from (\ref{eq.attainJ_g}). From (\ref{eq.attainJ}), we find out that $J(\bm{\nu})>J(\bm{0})$ when $\nu>(2^{\frac{1}{4}}+1){\sqrt{n}}\bm{1}_{L\times1}$. Therefore, the unique minimizer of the function $J$ must occur in the interval $[\bm{0}, (2^{\frac{1}{4}}+1){\sqrt{n}}\bm{1}_{L\times1}]$
\end{proof}
\begin{prop}\label{prop.uniqnessoptimal12weights}
Let $\bm{x}$ be a non-uniform $s$-block-sparse vector in $\mathbb{R}^n$ with parameters $\{\rho_i\}_{i=1}^L$ and $\{\alpha_i\}_{i=1}^L$. Then there exist unique optimal weights $\bm{\omega}^* \in \mathbb{R}_{+}^L$ (up to a positive scaling) that minimize $\hat{m}_{q,s,\bm{D\omega}}$. Moreover, the optimal weights $\bm{\omega}^*\in\mathbb{R}_{+}^L$ are obtained via the following integral equations.
\begin{align}\label{eq.12optimalweights}
\alpha_i\omega^*_i+\frac{1}{2^{\frac{k}{2}-1}\Gamma(\frac{k}{2})}(1-\alpha_i)\phi_B'(\omega^*_i)=0~:~i=1,..., L.
\end{align}
\end{prop}
\begin{proof}
Define $\mathcal{C}:=\partial\|.\|_{1,2}(\bm{x})$ and use Lemma \ref{lemma.mhat qsw} and \ref{lemmma.J(nu)} to obtain:
\begin{align}
&\inf_{\bm{\omega}\in\mathbb{R}_{+}^L}\hat{m}_{q,s,\bm{D\omega}}=\inf_{\bm{\omega}\in\mathbb{R}_{+}^L}\underset{{t\in\mathbb{R}_{+}}}{\inf}\Psi_{t,\bm{D\omega}}(\sigma,\bm{\rho},\bm{\alpha})\nonumber\\
&=\inf_{\bm{\nu}\in\mathbb{R}_{+}^L}J_b(\bm{\nu}),
\end{align}
where, $\Psi_{t,\bm{D\omega}}(\sigma,\bm{\rho},\bm{\alpha})$ is defined in (\ref{eq.hatm-qsw}). Also, we used a change of variable $\bm{\nu}=t\bm{\omega}$ to convert multivariate optimization problem to a single variable optimization problem. Thus, the function $J_b(\bm{\nu})$ is obtained via the following equation:
\begin{align}
&J_b(\bm{\nu})=\nonumber\\
&\sum_{i=1}^{L}\rho_i\bigg(\alpha_i(\nu(i)^2+1)+\frac{1}{2^{\frac{k}{2}-1}\Gamma(\frac{k}{2})}(1-\alpha_i)\phi_B(\nu(i))\bigg)
\end{align}
By considering Lemma \ref{lemmma.J(nu)} and $\bm{D}:=[\bm{1}_{\mathcal{V}_1},..., \bm{1}_{\mathcal{V}_q}]_{n\times q}[\bm{1}_{\mathcal{P}_1},..., \bm{1}_{\mathcal{P}_L}]_{q\times L}$, the function $J_b(\bm{\nu})$ is continuous and strictly convex and thus the unique minimizer can be obtained using $\nabla J_b(\bm{\nu})=\bm{0}\in\mathbb{R}^{L}$ which leads to:
\begin{align}
\alpha_i\nu^*(i)+\frac{1}{2^{\frac{k}{2}-1}\Gamma(\frac{k}{2})}(1-\alpha_i)\phi_B'(\nu^*(i))=0~:~i=1,..., L.
\end{align}
\end{proof}
\begin{thm}\label{thm.l12weighted}
Let $\bm{x}$ be a non-uniform $s$-block-sparse vector in $\mathbb{R}^{n}$ with parameters $\{\rho_i\}_{i=1}^L$ and $\{\alpha_i\}_{i=1}^L$. Then, number of measurements required for $\mathsf{P}_{1,2,\bm{D\omega}^*}$ is exactly equals the whole number of measurements required for $\mathsf{P}_{1,2}$ to recover each $\{\bm{x}_{\mathcal{P}_i}\in\mathbb{R}^n\}_{i=1}^L$ separately up to an asymptotically negligible error term.
\end{thm}
\begin{proof}
As previously defined in (\ref{eq.hatm-qsw}), with optimal weights, the upper bound for normalized number of measurements required for $\mathsf{P}_{1,2,\bm{D\omega}^*}$ to succeed is:
\begin{align}
&\hat{m}_{q,s,\bm{w}^*}=\inf_{\omega\in\mathbb{R}_{+}^L}\hat{m}_{q,s,\bm{D\omega}}=\nonumber\\
&\sum_{i=1}^{L}[\inf_{\nu_i\in\mathbb{R}_{+}}\underbrace{\bigg(\frac{\|\bm{x}_{\mathcal{P}_i}\|_{0,2}}{q}(\nu_i^2+k)+(1-\frac{\|\bm{x}_{\mathcal{P}_i}\|_{0,2}}{q})\phi_B(\nu_i)\bigg)}_{\Psi_{\nu_i,\|\bm{x}_{\mathcal{P}_i}\|_{0,2}}(\frac{\|\bm{x}_{\mathcal{P}_i}\|_{0,2}}{q})}]\nonumber\\
&=\sum_{i=1}^{L}\hat{m}_{q,\|\bm{x}_{\mathcal{P}_i}\|_{0,2}}
\end{align}
The expression in the bracket is exactly the upper bound for normalized number of measurements required for successful recovery of $\bm{x}_{\mathcal{P}_i}\in\mathbb{R}^n$ using $\mathsf{P}_{1,2}$ i.e. $\hat{m}_{q,\|\bm{x}_{\mathcal{P}_i}\|_{0,2}}$. Thus, regarding the error bounds obtained in Propositions \ref{prop.errorfor mhat qs} and \ref{prop.error for mhat qsw}, the relation between $m_{q,s,\bm{D\omega}^*}$ and $m_{q,\|\bm{x}_{\mathcal{P}_i}\|_{0,2}}$ is given by:
\begin{align}
&-\frac{2}{\sqrt{qL}}\le m_{q,s,\bm{D\omega}^*}-\sum_{i=1}^{L}m_{q,\|\bm{x}_{\mathcal{P}_i}\|_{0,2}}
\nonumber\\
&\le\frac{2}{\sqrt{q}}\sum_{i=1}^{L}(\|\bm{x}_{\mathcal{P}_i}\|_{0,2})^{-\frac{1}{2}}.
\end{align}
\end{proof}
\section{Gradient Sparsity}\label{section.TV}
In this section, the following questions are investigated about a non-uniform smooth model.
\begin{enumerate}
\item How many measurements one needs to recover an $s$-gradient sparse vector in $\mathbb{R}^{n}$ by solving $\mathsf{P}_{\mathrm{TV}}$ and $\mathsf{P}_{\mathrm{TV},\bm{w}}$?
\item What is the optimal choice of weights in $\mathsf{P}_{\mathrm{TV},\bm{w}}$ given extra prior information?
\end{enumerate}
In what follows in this section, we precisely answer these questions in two subsections.
\subsection{Number of Measurements for successful Recovery}
In this subsection, regarding Theorem (\ref{thm.Pfmeasurement}) for the functions $\|.\|_{\mathrm{TV}}$ and $\|.\|_{\mathrm{TV},\bm{w}}$, we denote normalized number of measurements required for $\mathsf{P}_{\mathrm{TV}}$ and $\mathsf{P}_{\mathrm{TV},\bm{w}}$ to recover an $s$-gradient sparse vector with probability $1-\eta$ by $m_{\mathrm{TV},s}$ and $m_{\mathrm{TV},\bm{w}}$, respectively which are defined as:
\begin{align}\label{eq.mTVs}
m_{\mathrm{TV},s}:=\frac{\delta(\mathcal{D}(\|.\|_{\mathrm{TV}},\bm{x}))}{n-1}.
\end{align}
\begin{align}\label{eq.mTVsweighted}
m_{\mathrm{TV},\bm{w}}:=\frac{\delta(\mathcal{D}(\|.\|_{\mathrm{TV},\bm{w}},\bm{x}))}{n-1}.
\end{align}
In the following, we obtain upper bounds for number of measurements required for $\mathsf{P}_{\mathrm{TV}}$ and $\mathsf{P}_{\mathrm{TV},\bm{w}}$ to succeed with probability $1-\eta$ where $\eta$ is the tolerance. In what follows in this section, $\bm{\Omega}_d\in\mathbb{R}^{{n-1}\times n}$ given below denotes the difference operator and $\bm{d}:=\bm{\Omega}_d \bm{x}$ is the difference $s$-sparse vector in $\mathbb{R}^{n-1}$.
\begin{align}
\bm{\Omega}_d=\begin{bmatrix}
                   1& -1 & 0 & \cdots & 0 \\
                   0& 1 & -1 & \cdots & 0  \\
                     &  \ddots    & \ddots       & \ddots & \\
                   0 & \cdots & \cdots & 1 & -1 &
               \end{bmatrix}\nonumber
\end{align}
\begin{lem}\label{lemma.mhat tv}
Let $\bm{x}\in \mathbb{R}^n$ be an $s$-gradient sparse vector. Then an upper bound for normalized number of measurements required for $\mathsf{P}_{\mathrm{TV}}$ to succeed (i.e. $m_{\mathrm{TV},s}$) is given by:
\begin{align}\label{eq.mhat TVs}
&\hat{m}_{\mathrm{TV},s}:=\inf_{t\ge0}\Psi_t(\sigma_1,\sigma_2,\sigma_3,\sigma_4)
\end{align}
with $\Psi_t(\sigma_1,\sigma_2,\sigma_3,\sigma_4)$ defined as:
\begin{align}
&\Psi_t(\sigma_1,\sigma_2,\sigma_3,\sigma_4)=\sigma_1+\sigma_2(1+4t^2)\nonumber\\
&+\sigma_3\phi_1(t,t)+(1-\sigma_1-\sigma_2-\sigma_3-\frac{1}{n-1})\phi_2(2t)+\nonumber\\
&\sigma_4(1+t^2)+\bar{\sigma}_4\phi_2(t)
\end{align}
where,
\begin{align}\label{eq.gradientsupport}
&\mathcal{S}_1=\nonumber\\
&\{i\in[n-1]:i\in\mathcal{S}_{{g}},~i-1\in\mathcal{S}_{{g}},sgn(d_i)sgn(d_{i-1})>0\},\nonumber\\
&\mathcal{S}_2=\nonumber\\
&\{i\in[n-1]:i\in\mathcal{S}_{{g}},~i-1\in\mathcal{S}_{{g}},sgn(d_i)sgn(d_{i-1})<0\},\nonumber\\
&\mathcal{S}_3=\{i\in[n-1]:i\in\mathcal{S}_{{g}},~i-1\in{\mathcal{\bar{S}}}_{{g}}\}\nonumber\\
&\mathcal{S}_4=\{i\in[n-1]:i\in{\mathcal{\bar{S}}}_{{g}},~i-1\in{\mathcal{S}}_{{g}}\}\nonumber\\
&\mathcal{S}_5=\{i\in[n-1]:i\in{\mathcal{\bar{S}}}_{{g}},~i-1\in{\mathcal{\bar{S}}}_{{g}}\}\nonumber\\
&\mathcal{S}_6=\{i\in[n-1]:i\in\mathcal{S}_{{g}},~i-1\notin[n-1]\}\nonumber\\
&\mathcal{\tilde{S}}_6=\{i\in[n-1]:i\in\mathcal{\bar{S}}_{{g}},~i-1\notin[n-1]\}\nonumber\\
&\mathcal{S}_7=\{i\in[n-1]:i\in\mathcal{{S}}_{{g}},~i+1\notin[n-1]\}\nonumber\\
&\mathcal{\tilde{S}}_7=\{i\in[n-1]:i\in\mathcal{\bar{S}}_{{g}},~i+1\notin[n-1]\}\nonumber\\
&\phi_1(a,b):=\frac{1}{\sqrt{2\pi}}\int_{b}^{\infty}(u-b)^2[e^{-\frac{(u-a)^2}{2}}+e^{-\frac{(u+a)^2}{2}}]du\nonumber\\
&\phi_2(z):=\sqrt{\frac{2}{\pi}}\int_{z}^{\infty}(u-z)^2e^{-\frac{u^2}{2}}du\nonumber\\
&\sigma_1=\frac{|\mathcal{S}_1|}{n-1},~\sigma_2=\frac{|\mathcal{S}_2|}{n-1}, ~\sigma_3=\frac{|\mathcal{S}_3\cup\mathcal{S}_4|}{n-1},\sigma_4=\frac{|\mathcal{S}_6\cup\mathcal{S}_7|}{n-1},\nonumber\\
&\bar{\sigma}_4=\frac{|\mathcal{\bar{S}}_6\cup\mathcal{\bar{S}}_7|}{n-1}.\nonumber\\
&s=|\mathcal{S}_1\cup\mathcal{S}_2\cup\mathcal{S}_3\cup\mathcal{S}_6|.
\end{align}
In (\ref{eq.gradientsupport}), $\mathcal{S}_1\cup\mathcal{S}_2$ and $\mathcal{S}_3\cup\mathcal{S}_4$ refer to consecutive and individual support, respectively.
\end{lem}
\begin{lem}\label{lemma.mhat tvweighted}
Let $\bm{x}\in \mathbb{R}^n$ be a non-uniform gradient sparse vector. Then an upper bound for normalized number of measurements required for $\mathsf{P}_{\mathrm{TV},\bm{w}}$ to succeed with high probability (i.e. $m_{\mathrm{TV},\bm{w}}$) is given by:
\begin{align}\label{eq.mhat wTVs}
&\hat{m}_{\mathrm{TV},\bm{w}}:=\inf_{t\ge0}\Psi_{t,\bm{w}}(\bm{\alpha},\bm{\alpha}',\bm{\beta},\bm{\beta}',\bm{\gamma},\bm{\gamma}',\bm{\varsigma},\bm{\varsigma}',\xi,\breve{\xi})
\end{align}
with $\Psi_{t,\bm{w}}(\bm{\alpha},\bm{\alpha}',\bm{\beta},\bm{\beta}',\bm{\gamma},\bm{\gamma}',\bm{\varsigma},\bm{\varsigma}',\xi,\breve{\xi})$ defined as:
%\bigg(\alpha'[1+t^2(\omega_i-\omega_{i-1})^2]+\beta'[1+t^2(\omega_i+\omega_{i-1})^2]+\gamma'\phi_1(t\omega_{i-1},t\omega_i)+\varsigma'\phi_1(t\omega_{i-1},t\omega_i)+(1-\alpha'-\beta'-\gamma'-\varsigma')\phi_2(t(\omega_i+\omega_{i-1}))\bigg)
\begin{align}\label{eq.psitv}
&\Psi_{t,\bm{w}}(\bm{\alpha},\bm{\alpha}',\bm{\beta},\bm{\beta}',\bm{\gamma},\bm{\gamma}',\bm{\varsigma},\bm{\varsigma}',\xi,\breve{\xi})=\nonumber\\
&\sum_{i=1}^{L}\rho_i\bigg(\alpha_i+\beta_i[1+4t^2\omega_i^2]+(\gamma_i+\varsigma_i)\phi_1(t\omega_i,t\omega_i)+\nonumber\\
&(1-\alpha_i-\beta_i-\gamma_i-\varsigma_i-\frac{1}{|\mathcal{P}_i|})\phi_2(2t\omega_i)+\nonumber\\
&+\alpha'_i[1+t^2(\omega_i-\omega_{i-1})^2]+\beta'_i[1+t^2(\omega_i+\omega_{i-1})^2]+\nonumber\\
&\gamma'_i\phi_1(t\omega_{i-1},t\omega_i)+\varsigma'_i\phi_1(t\omega_{i-1},t\omega_i)+\nonumber\\
&(1-\alpha'_i-\beta'_i-\gamma'_i-\varsigma'_i-\xi_i)\phi_2(t(\omega_i+\omega_{i-1}))+\nonumber\\
&(\xi_i+\breve{\xi}_i)(1+t^2\omega_i^2)+(\bar{\xi}_i+\bar{\breve{\xi}}_i)\phi_2(t\omega_i)\bigg),
\end{align}
where,
\begin{align}\label{eq.S1prim}
&\alpha_i=\frac{|\mathcal{P}_i\cap\mathcal{S}_1|}{|\mathcal{P}_i|},\beta_i=\frac{|\mathcal{P}_i\cap\mathcal{S}_1|}{|\mathcal{P}_i|},\gamma_i=\frac{|\mathcal{P}_i\cap\mathcal{S}_3|}{|\mathcal{P}_i|},\nonumber\\
&\varsigma_i=\frac{|\mathcal{P}_i\cap\mathcal{S}_4|}{|\mathcal{P}_i|},\alpha_i'=\frac{|\mathcal{P}_i\cap\mathcal{S}_1'|}{|\mathcal{P}_i|},\beta_i'=\frac{|\mathcal{P}_i\cap\mathcal{S}_1'|}{|\mathcal{P}_i|},\nonumber\\
&\gamma_i'=\frac{|\mathcal{P}_i\cap\mathcal{S}_3'|}{|\mathcal{P}_i|},\varsigma_i'=\frac{|\mathcal{P}_i\cap\mathcal{S}_4'|}{|\mathcal{P}_i|},\xi_i=\frac{|\mathcal{P}_i\cap\mathcal{S}_6|}{|\mathcal{P}_i|}\nonumber\\
&\breve{\xi}_i=\frac{|\mathcal{P}_i\cap\mathcal{S}_7|}{|\mathcal{P}_i|},\bar{\xi}_i=\frac{|\mathcal{P}_i\cap\mathcal{\tilde{S}}_6|}{|\mathcal{P}_i|},\bar{\breve{\xi}}_i=\frac{|\mathcal{P}_i\cap\mathcal{\tilde{S}}_7|}{|\mathcal{P}_i|}
\end{align}
In (\ref{eq.S1prim}), $\mathcal{P}_i\cap\mathcal{S}_{i}'$ denotes consecutive indices with each index located in a different partition.
\end{lem}
Before proceeding, we state a chain-rule lemma about subdifferential of a convex function which is later required.
\begin{lem}[\cite{bertsekas2014convex}]\label{lemma.chainrule}
let $f:\mathbb{R}^m\rightarrow (-\infty,\infty]$ be a proper convex function and $\bm{\Omega}$ be a matrix. Consider $F(\bm{x})=f(\bm{\Omega x})$. If $\mathrm{range}(\bm{\Omega})\cap \mathrm{relint}(dom(f))\neq\varnothing$, then:
\begin{align}
\partial F(\bm{x})=\bm{\Omega}^T\partial f(\bm{\Omega x})
\end{align}
\end{lem}
\begin{proof}[Proof of Lemma \ref{lemma.mhat tvweighted}]
By definition of the statistical dimension in (\ref{eq.statisticaldimension}) for the semi-norm function $\|.\|_{\mathrm{TV},\bm{w}}$ and Jensen's inequality, we have:
\begin{align}\label{eq.mTVweighted}
&m_{\mathrm{TV},\bm{w}}\le \inf_{t\ge0}\underbrace{{(n-1)^{-1}}\mathds{E}\mathrm{dist}^2(\bm{g},t\partial\|.\|_{\mathrm{TV},\bm{w}}( \bm{x}))}_{\Psi_{t,\bm{w}}(\alpha,\alpha',\beta,\beta',\gamma,\gamma',\varsigma,\varsigma',\xi,\tilde{\xi})}\nonumber\\
&:=\hat{m}_{\mathrm{TV},\bm{w}}.
\end{align}
The next step is to calculate $\partial\|.\|_{\mathrm{TV},\bm{w}}(\bm{x})$.
From Lemma \ref{lemma.chainrule}, we have:
\begin{align}\label{eq.subdiffTVW}
&\partial\|.\|_{\mathrm{TV},\bm{w}}(\bm{x})=\bm{\Omega}_d^T\partial\|.\|_{1,\bm{w}}(\bm{d})=\nonumber\\
&\bm{\Omega}_d^T\left\{\bm{z}\in\mathbb{R}^{n-1}:\begin{array}{lr}
        z_i=w_isgn(d_i), &  i\in \mathcal{S}_{{g}}\\
        |z_i|\le w_i, & \mathrm{o.w.}
        \end{array}\right\},
\end{align}
where we used (\ref{eq.l1subdiff}) with $\bm{\Omega}=\bm{I}_d\in\mathbb{R}^{n\times n}$. Now to calculate $\Psi_{t,\bm{w}}(\bm{\alpha},\bm{\alpha}',\bm{\beta},\bm{\beta}',\bm{\gamma},\bm{\gamma}',\bm{\varsigma},\bm{\varsigma}',\bm{\xi},\tilde{\bm{\xi}})$, regarding (\ref{eq.subdiffTVW}), we compute the distance of the dilated subdifferntial of descent cone of weighted TV norm at $\bm{x}\in\mathbb{R}^n$ from a standard Gaussian vector $\bm{g}\in\mathbb{R}^n$ which is given by:
\begin{align}\label{eq.distsubdiffwTV}
&\mathrm{dist}^2(\bm{g},t\partial\|.\|_{\mathrm{TV},\bm{w}}(\bm{x}))=\inf_{\bm{z}\in\partial\|.\|_{\mathrm{TV},\bm{w}}(\bm{x})}\|\bm{g}-t\bm{z}\|_2^2=\nonumber\\
&\inf_{\bm{z}\in\partial\|.\|_{\mathrm{TV},\bm{w}}(\bm{x})}\sum_{i=1}^{n}(g_i-t(\Omega_d^Tz)_i)^2=\nonumber\\
&\inf_{\bm{z}\in\partial\|.\|_{\mathrm{TV},\bm{w}}(\bm{x})}\sum_{i=1}^{n}\big(g_i-t\sum_{j\in\mathcal{S}_{{g}}}{\Omega}_d(j,i)w_jsgn(d_j)
\nonumber\\
&-t\sum_{j\in{\mathcal{\bar{S}}}_{{g}}}\bm{\Omega}_d(j,i)z(j)\big)^2=\inf_{\bm{z}\in\partial\|.\|_{\mathrm{TV},\bm{w}}(\bm{x})}\sum_{i=1}^{n}\bigg(g_i-\nonumber\\
&tw_isgn(d_i)\bm{1}_{i\in\mathcal{S}_{{g}}}+tw_{i-1}sgn(d_{i-1})\bm{1}_{i-1\in\mathcal{S}_{{g}}}-tz(i)\bm{1}_{i\in{\mathcal{\bar{S}}}_g}\nonumber\\
&+tz(i-1)\bm{1}_{i-1\in{\mathcal{\bar{S}}}_{{g}}}\bigg)^2=\nonumber\\
&\inf_{\bm{z}\in\partial\|.\|_{\mathrm{TV},\bm{w}}(\bm{x})}\Bigg(\sum_{i\in\mathcal{S}_1\cup\mathcal{S}_2}(g_i-tw_isgn(d_i)+tw_{i-1}sgn(d_{i-1}))^2\nonumber\\
&+\sum_{i\in\mathcal{S}_3}(g_i-tw_isgn(d_i)+tz(i-1))^2+\nonumber\\
&\sum_{i\in\mathcal{S}_4}(g_i+tw_{i-1}sgn(d_{i-1})-tz(i))^2+\sum_{i\in\mathcal{S}_5}(g_i-tz(i)\nonumber\\
&+tz(i-1))^2+(g_1-tw_1sgn(d_1))^2 1_{1\in\mathcal{S}_g}+\nonumber\\
&(g_1-tz(1))^2 1_{1\in{\mathcal{\bar{S}}}_g}+(g_n-tw_{n-1}sgn(d_{n-1}))^2 1_{n-1\in\mathcal{S}_g}+\nonumber\\
&(g_n-tz(n-1))^21_{n-1\in{\mathcal{\bar{S}}}_g}\Bigg)\nonumber\\
&=\sum_{i\in\mathcal{S}_1\cup\mathcal{S}_2}(g_i-tw_isgn(d_i)+tw_{i-1}sgn(d_{i-1}))^2\nonumber\\
&+\sum_{i\in\mathcal{S}_3}(\zeta_1-tw_{i-1})_{+}^2+\sum_{i\in\mathcal{S}_4}(\zeta_2-tw_{i})_{+}^2\nonumber\\
&+\sum_{i\in\mathcal{S}_5}(|g_i|-tw_{i}-tw_{i-1})_{+}^2+(g_1-tw_1sgn(d_1))^21_{1\in\mathcal{S}_g}+\nonumber\\
&(|g_1|-tw_1)_{+}^2 1_{1\in{\mathcal{\bar{S}}}_g}+(g_n-tw_{n-1}sgn(d_{n-1}))^2 1_{n-1\in\mathcal{S}_g}+\nonumber\\
&(|g_n|-tw_{n-1})^21_{n-1\in{\mathcal{\bar{S}}}_g},\nonumber\\
\end{align}
where $\zeta_1=|g_i-tw_isgn(d_i)|$, $\zeta_2=|g_i+tw_{i-1}sgn(d_{i-1})|$ and we used triangle inequality in the last part. By taking expectation from both sides, we reach:
\begin{align}\label{eq.TVEdist2}
&\mathds{E}\mathrm{dist}^2(\bm{g},t\partial\|.\|_{\mathrm{TV},\bm{w}}(\bm{x}))=\nonumber\\
&|\mathcal{S}_1\cup\mathcal{S}_2|+\sum_{i\in\mathcal{S}_1\cup\mathcal{S}_2}t^2(w_isgn(d_i)-w_{i-1}sgn(d_{i-1}))^2\nonumber\\
&+\sum_{i\in\mathcal{S}_3}\mathds{E}(\zeta_1-tw_{i-1})_{+}^2+\sum_{i\in\mathcal{S}_4}\mathds{E}(\zeta_2-tw_{i})_{+}^2+\nonumber\\
&\sum_{i\in\mathcal{S}_5}\mathds{E}(|g_i|-tw_{i}-tw_{i-1})_{+}^2+(1+t^2w_1^2)1_{1\in\mathcal{S}_g}\nonumber\\
&\mathds{E}(|g_1|-tw_{1})_{+}^2+(1+t^2w_{n-1}^2)1_{n-1\in\mathcal{S}_g}+\mathds{E}(|g_n|-tw_{n-1})_{+}^2
\end{align}
where,
\begin{align}\label{eq.EzetalTV}
&\mathds{E}(\zeta_1-tw_{i-1})_{+}^2=2\int_{0}^{\infty}a\mathds{P}(\zeta_1\ge a+tw_{i-1})da\nonumber\\
&=2{\frac{1}{\sqrt{2\pi}}}\int_{0}^{\infty}\int_{tw_{i-1}+a}^{\infty}a~(e^{-\frac{(u-tw_i)^2}{2}}+e^{-\frac{(u+tw_i)^2}{2}})du ~da\nonumber\\
&=2{\frac{1}{\sqrt{2\pi}}}\int_{tw_i}^{\infty}\int_{0}^{u-tw_{i-1}}a~(e^{-\frac{(u-tw_i)^2}{2}}+e^{-\frac{(u+tw_i)^2}{2}})da~du\nonumber\\
&=\phi_1(tw_i,tw_{i-1}),
\end{align}
and in (\ref{eq.EzetalTV}), the order of integration is changed together with a change of variable. Similarly, we have: $\mathds{E}(\zeta_2-tw_{i})_{+}^2=\phi_1(tw_{i-1},tw_i)$. With the same reasoning as in (\ref{eq.Ezetal1}), we have: $\mathds{E}(|g_i|-tw_{i}-tw_{i-1})_{+}^2=\phi_2(t(w_i+w_{i-1}))$.
As a consequence, (\ref{eq.TVEdist2}) becomes:
\begin{align}\label{eq.TVEdist2asli}
&\mathds{E}\mathrm{dist}^2(\bm{g},t\partial\|.\|_{\mathrm{TV},\bm{w}}(\bm{x}))=\nonumber\\
&|\mathcal{S}_1\cup\mathcal{S}_2|+\sum_{i\in\mathcal{S}_1\cup\mathcal{S}_2}t^2(w_isgn(d_i)-w_{i-1}sgn(d_{i-1}))^2\nonumber\\
&+\sum_{i\in\mathcal{S}_3}\phi_1(tw_{i},tw_{i-1})+\sum_{i\in\mathcal{S}_4}\phi_1(tw_{i-1},tw_i)+\nonumber\\
&\sum_{i\in\mathcal{S}_5}\phi_2(t(w_i+w_{i-1}))+(1+t^2w_1^2)1_{1\in\mathcal{S}_g}+\phi_2(tw_1)1_{1\in{\mathcal{\bar{S}}}_g}\nonumber\\
&+(1+t^2w_{n-1}^2)1_{n-1\in\mathcal{S}_g}+\phi_2(tw_{n-1})1_{n-1\in{\mathcal{\bar{S}}}_g}.
\end{align}
By incorporating prior information and assigning weights to the sorted partitions $\{\mathcal{P}_i\}_{i=1}^L$ using $\bm{w}=\bm{D\omega} \in\mathbb{R}^{n-1}$ with $\bm{D}:=[\bm{1}_{\mathcal{P}_1},...,\bm{1}_{\mathcal{P}_L}]\in\mathbb{R}^{n-1\times L}$ we reach:
%|\mathcal{P}_i\cap\mathcal{S}_1'|(1+t^2(\omega_i-\omega_{i-1})^2)+|\mathcal{P}_i\cap\mathcal{S}_2'|(1+t^2(\omega_i+\omega_{i-1})^2)+|\mathcal{P}_i\cap\mathcal{S}_3'|\phi_1(t\omega_i,t\omega_{i-1})+|\mathcal{P}_i\cap\mathcal{S}_4'|\phi_1(t\omega_{i-1},t\omega_i)+|\mathcal{P}_i\cap\mathcal{S}_5'|\phi_3(t(\omega_i+\omega_{i-1}))
\begin{align}\label{eq.Edist2wTV}
&\mathds{E}\mathrm{dist}^2(\bm{g},t\partial\|.\|_{\mathrm{TV},\bm{D\omega}}(\bm{x}))=\nonumber\\
&\sum_{i=1}^{L}\Bigg(|\mathcal{P}_i\cap\mathcal{S}_1|+|\mathcal{P}_i\cap\mathcal{S}_2|(1+4t^2\omega_i^2)+\nonumber\\
&|\mathcal{P}_i\cap\mathcal{S}_3|\phi_1(t\omega_i,t\omega_i)+|\mathcal{P}_i\cap\mathcal{S}_4|\phi_1(t\omega_i,t\omega_{i})+\nonumber\\
&|\mathcal{P}_i\cap\mathcal{S}_5|\phi_2(2t\omega_i)+|\mathcal{P}_i\cap\mathcal{S}_1'|(1+t^2(\omega_i-\omega_{i-1})^2)+\nonumber\\
&|\mathcal{P}_i\cap\mathcal{S}_2'|(1+t^2(\omega_i+\omega_{i-1})^2)+|\mathcal{P}_i\cap\mathcal{S}_3'|\phi_1(t\omega_i,t\omega_{i-1})+\nonumber\\
&|\mathcal{P}_i\cap\mathcal{S}_4'|\phi_1(t\omega_{i-1},t\omega_i)+|\mathcal{P}_i\cap\mathcal{S}_5'|\phi_3(t(\omega_i+\omega_{i-1}))+\nonumber\\
&|\mathcal{P}_i\cap\{\mathcal{S}_6\cup\mathcal{S}_7\}|(1+t^2\omega_i^2)+|\mathcal{P}_i\cap\{\mathcal{\tilde{S}}_6\cup\mathcal{\tilde{S}}_7\}|\phi_2(t\omega_i)\Bigg).
\end{align}
By normalizing to $n-1$, we reach (\ref{eq.mhat wTVs}).
\end{proof}
\begin{proof}[Proof of Lemma \ref{lemma.mhat tv}]
By considering (\ref{eq.mTVs}), we find an upper bound for $\delta(\mathcal{D}(\|.\|_{\mathrm{TV}},\bm{x}))/n$ which we denote by $\Psi_{t}(\sigma_1,\sigma_2,\sigma_3,\sigma_4)$. The procedure is exactly the same as the proof of Lemma \ref{lemma.mhat tvweighted} with the weights set to one i.e. $\bm{w}=\bm{1}\in \mathbb{R}^{n-1}$. In fact, we have: $\Psi_t(\sigma_1,\sigma_2,\sigma_3,\sigma_4)=\Psi_{t,\bm{1}}(\bm{\alpha},\bm{\alpha}',\bm{\beta},\bm{\beta}',\bm{\gamma},\bm{\gamma}',\bm{\varsigma},\bm{\varsigma}',\xi,\breve{\xi})$. By replacing $\bm{w}=\bm{1}\in \mathbb{R}^{n-1}$ in (\ref{eq.TVEdist2asli}) and (\ref{eq.mTVweighted}), we reach $\Psi_t(\sigma_1,\sigma_2,\sigma_3,\sigma_4)$ and $\hat{m}_{\mathrm{TV},s}$ in Lemma \ref{lemma.mhat tv}.
\end{proof}
In the following propositions, we prove that the obtained upper bounds in Lemma's \ref{lemma.mhat tv} and \ref{lemma.mhat tvweighted}, are asymptotically tight.
\begin{prop}\label{prop.error for mhat TV}. Normalized number of Gaussian linear measurements required for $\mathsf{P}_{\mathrm{TV}}$ to succeed (i.e. $m_{\mathrm{TV},s}$) satisfies the following error bound.
\begin{align}\label{eq.TVerrorbound}
\hat{m}_{\mathrm{TV},s}-\frac{2\kappa(\bm{\Omega}_d)}{\sqrt{s(n-1)}}\le m_{\mathrm{TV},s}\le \hat{m}_{\mathrm{TV},s}.
\end{align}
\end{prop}
\begin{proof}
By using the error bound (\ref{eq.errorbound}) and (\ref{eq.subdiffTVW}) with $\bm{w}=\bm{1}\in \mathbb{R}^{n-1}$, we obtain the numerator of the error bound (\ref{eq.errorbound}) as given by:
\begin{align}
&2\sup_{\bm{s}\in\partial\|.\|_{\mathrm{TV}}(\bm{x})}\|\bm{s}\|_2=\sup_{\|\bm{z}\|_{\infty}\le1} 2\|\bm{\Omega_d}^T\bm{z}\|_2,\nonumber\\
&\le 2\|\bm{\Omega_d}\|_{2\rightarrow2}\sqrt{n-1}.
\end{align}
Also, the denominator can be lower bounded as below:
\begin{align}\label{eq.TVdenominator}
&\frac{\|\bm{\Omega_dx}\|_1}{\|\bm{x}\|_2}\ge \frac{\|\bm{\Omega_dx}\|_1}{\|\bm{\Omega_dx}\|_2\|\bm{\Omega}_d^\dagger\|_{2\rightarrow2}}\le
\sqrt{s}\|\bm{\Omega}_d^\dagger\|_{2\rightarrow2}^{-1}.
\end{align}
The error bound in (\ref{eq.errorbound}) depends only on $\mathcal{D}(\|.\|_{\mathrm{TV}},\bm{x})$ and $\partial\|.\|_{\mathrm{TV}}(\bm{x})$. Further, $\partial\|.\|_{\mathrm{TV}}(\bm{x})$ depends only on $sgn(\bm{d})$ not the magnitudes of $\bm{d}$. So a vector
\begin{align}
    \bm{z} = \left\{\begin{array}{lr}
        sgn(v)_i, &  i\in \mathcal{S}_{{g}}\\
        0, & i\in {\mathcal{\bar{S}}}_{{g}}
        \end{array}\right\}\in \mathbb{R}^{n-1},\nonumber
  \end{align}
  with $sgn(\bm{z})=sgn(\bm{d})$ can be chosen to have equality in the last inequality in (\ref{eq.TVdenominator}). Therefore, the error of obtaining $\hat{m}_{\mathrm{TV},s}$ in (\ref{eq.mhat TVs}) is at most $\frac{2\kappa(\bm{\Omega}_d)}{\sqrt{(n-1)s}}$.
\end{proof}
\begin{prop}\label{prop.error for mhat TVw}
Normalized number of Gaussian linear measurements required for $\mathsf{P}_{\mathrm{TV},\bm{w}}$ to successfully recover a non-uniform gradient sparse vector in $\mathbb{R}^{n}$ with parameters $\{\rho_i\}_{i=1}^L$ and $\{\alpha_i,\alpha_i',\beta_i,\beta_i',\gamma_i,\gamma_i',\varsigma_i,\varsigma_i',\xi_i,\breve{\xi}_i\}_{i=1}^L$ (i.e. $m_{\mathrm{TV},\bm{w}}$) satisfies the following error bound.
\begin{align}\label{eq.TVwerrorbound}
\hat{m}_{\mathrm{TV},\bm{w}}-\frac{2\kappa(\bm{\Omega}_d)}{\sqrt{(n-1)L}}\le m_{\mathrm{TV},\bm{w}}\le \hat{m}_{\mathrm{TV},\bm{w}}
\end{align}
\end{prop}
\begin{proof}
By using the error bound (\ref{eq.errorbound}) with $f(\bm{x}):=\|\bm{x}\|_{TV,\bm{w}}$, the subdifferntial  (\ref{eq.subdiffTVW}), and $\bm{w}=\sum_{i=1}^{L}\omega_i \bm{1}_{\mathcal{P}_i}$, the numerator in (\ref{eq.errorbound}) is obtained by:
\begin{align}
&2\sup_{\bm{s}\in\partial\|.\|_{TV,\bm{w}}(\bm{x})}\|\bm{s}\|_2=\sup_{|z_i|\le w_i,~i=1,..., n-1} 2\|\bm{\Omega}_d^T\bm{z}\|_2\le\nonumber\\
&2\|\bm{\Omega}_d\|_{2\rightarrow2}\sqrt{\sum_{i=1}^{p}w_i^2}=2\|\bm{\Omega}_d\|_{2\rightarrow2}\sqrt{\sum_{i=1}^{L}|\mathcal{P}_i|\omega_i^2}.
\end{align}
Also, the denominator in (\ref{eq.errorbound}) can be lower bounded as:
\begin{align}\label{eq.TVwdenominator}
&\frac{\|\bm{\Omega}_dx\|_{1,\bm{w}}}{\|\bm{x}\|_2}\ge \frac{\|\bm{\Omega}_dx\|_{1,\bm{w}}}{\|\bm{\Omega}_dx\|_2\|\bm{\Omega}_d^\dagger\|_{2\rightarrow2}}\le
\sqrt{\sum_{i\in\mathcal{S}_{{g}}}w_i^2}\|\bm{\Omega}_d^\dagger\|_{2\rightarrow2}^{-1}\nonumber\\
&=\sqrt{\sum_{i=1}^{L}|\mathcal{P}_i\cap\mathcal{S}_{{g}}|\omega_i^2} \|\bm{\Omega}_d^\dagger\|_{2\rightarrow2}^{-1}
\end{align}
The error bound (\ref{eq.errorbound}) for the function $\|.\|_{\mathrm{TV},\bm{w}}$ only depends on $\mathcal{D}(\|.\|_{\mathrm{TV},\bm{w}},\bm{x})$ and $\partial\|.\|_{\mathrm{TV},\bm{w}}(\bm{x})$. Further, $\partial\|.\|_{\mathrm{TV},\bm{w}}(\bm{x})$ only depends on $sgn(\bm{d})$ not the magnitudes of $\bm{d}$. So we may choose a vector
\begin{align}
    \bm{z} = \left\{\begin{array}{lr}
        w_isgn(v)_i, &  i\in \mathcal{S}_{{g}}\\
        0, & i\in \mathcal{\bar{S}}_g
        \end{array}\right\}\in \mathbb{R}^{n-1},\nonumber
\end{align}
  with $sgn(\bm{z})=sgn(\bm{d})$ in (\ref{eq.TVwdenominator}) to have equality in the last inequality in (\ref{eq.TVwdenominator}). Therefore, the error in obtaining $\hat{m}_{TV,s,\bm{w}}$ in (\ref{eq.mhat wTVs}) is at most:
  \begin{align}
  &\frac{2\|\bm{\Omega}_d\|_{2\rightarrow2}\sqrt{\sum_{i=1}^{L}|\mathcal{P}_i|\omega_i^2}}{n-1\|\bm{\Omega}_d^\dagger\|_{2\rightarrow2}^{-1}\sqrt{\sum_{i=1}^{L}|\mathcal{P}_i\cap\mathcal{S}_{{g}}|\omega_i^2}}
  \le\nonumber\\
  &\frac{2\kappa(\bm{\Omega}_d)}{n-1}\sqrt{\frac{1}{\min_{i\in[L]}\frac{|\mathcal{P}_i\cap\mathcal{S}_{{g}}|}{|\mathcal{P}_i|}}}\le\frac{2\kappa(\bm{\Omega}_d)}{\sqrt{L(n-1)}}
  \end{align}
where in the last step, we used the fact that $|\mathcal{P}_i\cap\mathcal{S}_{{g}}|\ge1,~|\mathcal{P}_i|\le\frac{n-1}{L}$ for at least one $i\in[L]$. As a consequence, the error in obtaining $\hat{m}_{\mathrm{TV},\bm{w}}$ is at most $\frac{2\kappa(\bm{\Omega}_d)}{\sqrt{L(n-1)}}$.
\end{proof}
\subsection{Optimal Weights}
Infimum of (\ref{eq.TVwerrorbound}) gives:
\begin{align}\label{eq.infmhatqswTVerrorbound}
&\inf_{\bm{\omega}\in\mathbb{R}_{+}^L}\hat{m}_{\mathrm{TV},s,\bm{D\omega}}-\frac{2\kappa(\bm{\Omega}_d)}{\sqrt{(n-1)L}}\le\inf_{\bm{\omega}\in\mathbb{R}_{+}^L}{m}_{TV,s,\bm{D\omega}}\le\nonumber\\
& \inf_{\bm{\omega}\in\mathbb{R}_{+}^L}\hat{m}_{\mathrm{TV},s,\bm{D\omega}},
\end{align}
in which, $\bm{D}=[\bm{1}_{\mathcal{P}_1}, ..., \bm{1}_{\mathcal{P}_L}]\in\mathbb{R}^{n-1\times L}$.\par
In $\mathsf{P}_{\mathrm{TV},\bm{D\omega}}$, we call the weight $\bm{\omega}^*=\underset{\bm{\omega}\in\mathbb{R}_{+}^L}{\arg\min}~~\hat{m}_{\mathrm{TV},s,\bm{D\omega}}\in\mathbb{R}_{+}^L$ optimal since it asymptotically minimizes the number of measurements required for $\mathsf{P}_{\mathrm{TV},\bm{D\omega}}$ to succeed.
Before proving that the optimal weights are in fact unique, we state a beneficial lemma that establish our purpose.
\begin{lem}\label{lemma.Jt(v)}
Let $\mathcal{C}:=\partial\|.\|_{\mathrm{TV}}\subseteq\mathbb{R}^{n-1}$. Suppose that $\mathcal{C}$ does not contain the origin. In particular, it is compact and there are upper and lower bounds that satisfy $1\le\|\bm{z}\|_2\le \sqrt{n-1}$ for all $\bm{z}\in \mathcal{C}$. Also denote the standard normal vector by $\bm{g}\in \mathbb{R}^n$. Consider the function
\begin{align}
J(\bm{\nu}):=\mathds{E}\mathrm{dist}^2(\bm{g},\bm{\Omega}_d^T(\bm{\upsilon}\odot \mathcal{C}))=\mathds{E}[J_{\bm{g}}(\bm{\nu})] \nonumber\\
 \text{with}~~\bm{\upsilon}=\bm{D\nu} \in \mathbb{R}^{n-1} ~~\text{for}~~ \bm{\nu}\in\mathbb{R}_{+}^L.
\end{align}
The function $J$ is strictly convex, continuous at $\bm{\nu}\in \mathbb{R}_{+}^L$ and differentiable for $\bm{\nu}\in \mathbb{R}_{++}^L$. Moreover, there exist a unique point that minimize $J$.
\end{lem}
\begin{proof}
\textit{Continuity in bounded points}.
We must show that sufficiently small changes in $\bm{\nu}$ result in arbitrary small changes in $J(\bm{\nu})$. By the definition of $J_{\bm{g}}(\bm{\nu})$, we have:
\begin{align}
&J_{\bm{g}}(\bm{\nu})-J_{\bm{g}}(\tilde{\bm{\nu}})=\|\bm{g}-\mathcal{P}_{\bm{\Omega}_d^T(\bm{\upsilon}\odot \mathcal{C})}(\bm{g})\|_2^2-\nonumber\\
&\|\bm{g}-\mathcal{P}_{\bm{\Omega}_d^T(\tilde{\bm{\upsilon}}\odot \mathcal{C})}(\bm{g})\|_2^2=2\langle \bm{g},\mathcal{P}_{\bm{\Omega}_d^T(\tilde{\bm{\upsilon}}\odot \mathcal{C})}(\bm{g})-\mathcal{P}_{\bm{\Omega}_d^T({\bm{\upsilon}}\odot \mathcal{C})}(\bm{g})\rangle+\nonumber\\
&\big(\|\mathcal{P}_{\bm{\Omega}_d^T(\bm{\upsilon}\odot \mathcal{C})}(\bm{g})\|_2-\|\mathcal{P}_{\bm{\Omega}_d^T(\tilde{\bm{\upsilon}}\odot \mathcal{C})}(\bm{g})\|_2\big)\big(\|\mathcal{P}_{\bm{\Omega}_d^T(\bm{\upsilon}\odot \mathcal{C})}(\bm{g})\|_2+\nonumber\\
&\|\mathcal{P}_{\bm{\Omega}_d^T(\tilde{\bm{\upsilon}}\odot \mathcal{C})}(\bm{g})\|_2\big)\nonumber\\
&\text{The absolute value satisfies:}\nonumber\\
&|J_{\bm{g}}(\bm{\nu})-J_{\bm{g}}(\tilde{\bm{\nu}})|\nonumber\\
&\le \bigg(2\|\bm{g}\|_2\|\bm{\Omega}_d\|_{2\rightarrow2}\sqrt{n-1}+\nonumber\\
&(n-1)\|\bm{\Omega}_d\|_{2\rightarrow2}^2(\|\bm{\nu}\|_1+\|\tilde{\bm{\nu}}\|_1)\bigg)\|\tilde{\bm{\nu}}-\bm{\nu}\|_1
\end{align}
where in the line above, we used the fact that,
\begin{align}
&\|\mathcal{P}_{\bm{\Omega}_d^T(\bm{\upsilon}\odot \mathcal{C})}(\bm{g})\|_2\le\sup_{\bm{z}\in \mathcal{C}}\|\bm{\Omega}_d^T(\bm{\upsilon}\odot \bm{z})\|_2\le\nonumber\\
&\|\bm{\Omega}_d\|_{2\rightarrow2}\|\bm{D\nu}\|_{\infty}\sqrt{n-1}\le\|\bm{\Omega}_d\|_{2\rightarrow2}\|\bm{\nu}\|_1\sqrt{n-1}\|\bm{D}\|_{1\rightarrow\infty}\nonumber\\
&=\|\bm{\Omega}_d\|_{2\rightarrow2}\|\bm{\nu}\|_1\sqrt{n-1},
\end{align}
and,
\begin{align}
&\|\mathcal{P}_{\bm{\Omega}_d^T(\bm{\upsilon}\odot \mathcal{C})}(\bm{g})\|_2-\|\mathcal{P}_{\bm{\Omega}_d^T(\tilde{\bm{\upsilon}}\odot \mathcal{C})}(\bm{g})\|_2\le \nonumber\\
&\sup_{\bm{z}\in \mathcal{C}}\bigg(\|\bm{\Omega}_d^T(\bm{\upsilon}\odot \bm{z})\|_2-\|\bm{\Omega}_d^T(\tilde{\bm{\upsilon}}\odot\bm{z})\|_2\bigg)\nonumber\\
&\le
\|\bm{\Omega}_d\|_{2\rightarrow2}\sup_{\bm{z}\in \mathcal{C}}\big(\|(\bm{\upsilon}-\tilde{\bm{\upsilon}})\odot \bm{z}\|_2\big)\nonumber\\
&\le
\|\bm{\nu}-\tilde{\bm{\nu}}\|_1\sqrt{n-1}\|\bm{D}\|_{1\rightarrow\infty}\|\bm{\Omega}_d\|_{2\rightarrow2}=\nonumber\\
&\|\bm{\Omega}_d\|_{2\rightarrow2}\|\bm{\nu}-\tilde{\bm{\nu}}\|_1\sqrt{n-1}.
\end{align}
As a consequence, we obtain:
\begin{align}
&|J(\bm{\nu})-J(\tilde{\bm{\nu}})|\nonumber\\
&\le \bigg(2\|\bm{\Omega}_d\|_{2\rightarrow2}\sqrt{(n^2-n)L}+\nonumber\\
&(n-1)\|\bm{\Omega}_d\|_{2\rightarrow2}^2\sqrt{L}(\|\bm{\nu}\|_1+\|\tilde{\bm{\nu}}\|_1)\bigg)\|\tilde{\bm{\nu}}-\bm{\nu}\|_2\rightarrow 0 \nonumber\\
 &~\text{as}~ \bm{\nu}\rightarrow~ \tilde{\bm{\nu}}
\end{align}
Since $\|\bm{\nu}\|_1$ is bounded, continuity holds.\par
\textit{Convexity}. Let $\bm{\nu}\in \mathbb{R}_{+}^L$ , $\tilde{\bm{\nu}}\in\mathbb{R}_{+}^L$ and $\theta\in[0,1]$ with $\bm{\upsilon}=\bm{D\nu}$ and $\tilde{\bm{\upsilon}}=\bm{D}\tilde{\bm{\nu}}$. Then we have:
\begin{align}\label{helpconvexityTV}
&\forall \epsilon , \tilde{\epsilon}>0~\exists \bm{z} ,\tilde{\bm{z}}\in \mathcal{C} ~~\text{such that}\nonumber\\
&\|\bm{g}-\bm{\Omega}_d^T(\bm{\upsilon}\odot \bm{z})\|_2\le \mathrm{dist}(\bm{g},\bm{\Omega}_d^T(\bm{\upsilon}\odot \mathcal{C}))+\epsilon\nonumber\\
&\|\bm{g}-\bm{\Omega}_d^T(\tilde{\bm{\upsilon}}\odot \tilde{\bm{z}})\|_2\le \mathrm{dist}(\bm{g},\bm{\Omega}_d^T(\tilde{\bm{\upsilon}}\odot \mathcal{C}))+\tilde{\epsilon}
\end{align}
Since otherwise we have:
\begin{align}
&\forall \bm{z},\tilde{\bm{z}}\in \mathcal{C}:\nonumber\\
&\|\bm{g}-\bm{\Omega}_d^T(\bm{\upsilon}\odot \bm{z})\|_2>\mathrm{dist}(\bm{g},\bm{\Omega}_d^T(\bm{\upsilon}\odot \mathcal{C}))+\epsilon\nonumber\\
&\|\bm{g}-\bm{\Omega}_d^T(\tilde{\bm{\upsilon}}\odot \tilde{\bm{z}})\|_2> \mathrm{dist}(\bm{g},\bm{\Omega}_d^T(\tilde{\bm{\upsilon}}\odot \mathcal{C}))+\tilde{\epsilon}
\end{align}
By taking the infimum over $\bm{z},\tilde{\bm{z}}\in \mathcal{C}$, we reach a contradiction. We proceed to prove convexity of $\mathrm{dist}(\bm{g},\bm{\Omega}_d^T((\bm{D\nu})\odot \mathcal{C}))$.
\begin{align}\label{eq.convexitydistTV}
&\mathrm{dist}(\bm{g},\bm{\Omega}_d^T((\theta\bm{\upsilon}+(1-\theta)\tilde{\bm{\upsilon}})\odot \mathcal{C}))=\nonumber\\
&\inf_{\bm{z} \in \mathcal{C}}\|\bm{g}-\bm{\Omega}_d^T((\theta\bm{\upsilon}+(1-\theta)\tilde{\bm{\upsilon}})\odot \bm{z})\|_2\nonumber\\
&\le\inf_{\bm{z}_1\in \mathcal{C},\bm{z}_2\in \mathcal{C}}\|\bm{g}-\theta\bm{\Omega}_d^T(\bm{\upsilon}\odot \bm{z}_1)-(1-\theta)\bm{\Omega}_d^T(\bm{\upsilon}\odot \bm{z}_2)\|_2\le\nonumber\\
&\theta\|\bm{g}-\bm{\Omega}_d^T(\bm{\upsilon}\odot \bm{z}_1)\|_2+(1-\theta)\|\bm{g}-\bm{\Omega}_d^T(\tilde{\bm{\upsilon}}\odot \bm{z}_2)\|_2\le\nonumber\\
&\theta \mathrm{dist}(\bm{g},\bm{\Omega}_d^T(\bm{\upsilon}\odot \mathcal{C}))+(1-\theta)\mathrm{dist}(\bm{g},\bm{\Omega}_d^T(\tilde{\bm{\upsilon}}\odot \mathcal{C}))+\epsilon+\tilde{\epsilon}.
\end{align}
Since this holds for any $\epsilon$ and $\tilde{\epsilon}$, $\mathrm{dist}(\bm{g},\bm{\Omega}_d^T((\bm{D\nu})\odot \mathcal{C}))$ is a convex function. As the square of a non-negative convex function is convex, $J_{\bm{g}}(\bm{\nu})$ is a convex function. At last, the function $J(\bm{\nu})$ is the average of convex functions, hence is convex.
In (\ref{eq.convexitydistTV}), the first inequality comes from the following lemma.
In the second inequality, we used the triangle inequality of norms. The third inequality uses the relation (\ref{helpconvexityTV}).
\begin{lem}{\label{lemma.TVbenefit}}
For all $\bm{z}_1,\bm{z}_2 \in \mathcal{C}$, the following relation holds.
\begin{align}\label{eq.TVbenefit}
&\forall \bm{z}_1,\bm{z}_2 \in \mathcal{C} ~\text{and}~\theta\in(0,1):\nonumber\\
&\theta\bm{\Omega}_d^T(\bm{\upsilon}\odot \bm{z}_1)+(1-\theta)\bm{\Omega}_d^T(\tilde{\bm{\upsilon}}\odot \bm{z}_2)\nonumber\\
&\in~~\bm{\Omega}_d^T((\theta \bm{\upsilon}+(1-\theta)\tilde{\bm{\upsilon}})\odot \mathcal{C})\nonumber\\
\end{align}
\end{lem}
\begin{proof}[Proof of Lemma \ref{lemma.TVbenefit}]
For $i\in \mathcal{S}_1\cup\mathcal{S}_2\subseteq\mathbb{R}^{n}$, we have:
\begin{align}
&\theta(\Omega_d^T(\bm{\upsilon}\odot \bm{z}_1))(i)+(1-\theta)(\Omega_d^T(\tilde{{\upsilon}}\odot {z}_2))(i)=\nonumber\\
&\theta\upsilon(i)sgn(d_i)-\theta\upsilon(i-1)sgn(d_{i-1})+\nonumber\\
&(1-\theta)\tilde{\upsilon}(i)sgn(d_i)-(1-\theta)\tilde{\upsilon}(i-1)sgn(d_{i-1})=\nonumber\\
&(\Omega_d^T((\theta \upsilon+(1-\theta)\tilde{\upsilon})\odot z))(i).
\end{align}
For $i\in\mathcal{S}_3$, we can always find a $\bm{z}\in \mathcal{C}$ supported on $\mathcal{S}_3$ such that:
\begin{align}
&\theta\Omega_d^T(\upsilon\odot z_1)(i)+(1-\theta)\Omega_d^T(\tilde{\upsilon}\odot z_2)(i)=\nonumber\\
&\theta\upsilon(i)sgn(d_i)-\theta\upsilon(i-1)z_1(i-1)+\nonumber\\
&(1-\theta)\tilde{\upsilon}(i)sgn(d_i)-(1-\theta)\upsilon(i-1)z_2(i-1)\nonumber\\
&=\theta\upsilon(i)sgn(d_i)+(1-\theta)\tilde{\upsilon}(i)sgn(d_i)\nonumber\\
&-(\theta\upsilon(i-1)+(1-\theta)\upsilon(i-1))z(i-1)
\end{align}.
Since otherwise with $z_1(i-1)=z_2(i-1)=1$ and $z(i-1)=1$ we reach a contradiction.\\
For $i\in\mathcal{S}_4$, we can always find a $\bm{z}\in \mathcal{C}$ supported on $\mathcal{S}_4$ such that:
\begin{align}
&\theta\Omega_d^T(\upsilon\odot z_1)(i)+(1-\theta)\Omega_d^T(\tilde{\upsilon}\odot z_2)(i)=\nonumber\\
&-\theta\upsilon(i-1)sgn(d_{i-1})+\theta\upsilon(i)z_1(i)\nonumber\\
&-(1-\theta)\tilde{\upsilon}(i-1)sgn(d_{i-1})+(1-\theta)\tilde{\upsilon}(i)z_2(i)=\nonumber\\
&-(\theta\upsilon(i-1)+(1-\theta)\tilde{\upsilon}(i-1))sgn(d_{i-1})+\nonumber\\
&(\theta\upsilon(i)+(1-\theta)\tilde{\upsilon}(i))z(i).
\end{align}
Since otherwise with $z_1(i)=z_2(i)=1$ and $z(i)=1$ we reach a contradiction.\\
For $i\in\mathcal{S}_5$, we can always find a $\bm{z}\in \mathcal{C}$ supported on $\mathcal{S}_5$ such that:
\begin{align}
&\theta\Omega_d^T(\upsilon\odot z_1)(i)+(1-\theta)\Omega_d^T(\tilde{\upsilon}\odot z_2)(i)=\nonumber\\
&\theta\upsilon(i)z_1(i)-\theta\upsilon(i-1)z_1(i-1)+\nonumber\\
&(1-\theta)\tilde{\upsilon}(i)z_2(i)-(1-\theta)\tilde{\upsilon}(i-1)z_2(i-1)=\nonumber\\
&(\theta\upsilon(i)+(1-\theta)\tilde{\upsilon}(i))z(i)-\nonumber\\
&(\theta\upsilon(i-1)+(1-\theta)\tilde{\upsilon}(i-1))z(i-1).
\end{align}
Since otherwise with $z_1(i)=z_2(i)=1$ and $z(i)=z(i-1)=1$ we reach a contradiction.\\
For $i\in\mathcal{S}_g,i-1\notin\mathcal{S}_g,{\mathcal{\bar{S}}}_g$, we have:
\begin{align}
&\theta\Omega_d^T(\upsilon\odot z_1)(i)+(1-\theta)\Omega_d^T(\tilde{\upsilon}\odot z_2)(i)=\nonumber\\
&\theta\upsilon(i)sgn(d_i)+(1-\theta)\tilde{\upsilon}(i)sgn(d_i)=\nonumber\\
&(\Omega_d^T((\theta \upsilon+(1-\theta)\tilde{\upsilon})\odot z))(i).
\end{align}
For $i=n$ with $n-1\in{\mathcal{\bar{S}}}_g$, we can always find a $z\in\mathcal{C}$ supported on the index $n-1$ such that:
\begin{align}
&\theta\Omega_d^T(\upsilon\odot z_1)(n)+(1-\theta)\Omega_d^T(\tilde{\upsilon}\odot z_2)(n)=\nonumber\\
&-\theta \upsilon(n-1)z_1(n-1)-(1-\theta)\tilde{\upsilon}(n-1)z_2(n-1)=\nonumber\\
&(-\theta\upsilon(n-1)-(1-\theta)\tilde{\upsilon}(n-1))z(n-1).
\end{align}
Since otherwise with $z_1(n-1)=z_2(n-1)=1$ and $z(n-1)=1$ we reach a contradiction.
\end{proof}
\textit{Strict convexity}. We show strict convexity by contradiction. If $J(\bm{\nu})$ were not strictly convex, there would be vectors $\bm{\nu},\tilde{\bm{\nu}}\in\mathbb{R}_{+}^L$ with $\bm{\upsilon}=\bm{D\nu}, \tilde{\bm{\upsilon}}=\bm{D}\tilde{\bm{\nu}}$ and $\theta\in (0,1)$ such that,
\begin{align}\label{eq.TVstrictconvex}
\mathds{E}[J_{\bm{g}}(\theta\bm{\nu}+(1-\theta)\tilde{\bm{\nu}})]=\mathds{E}[\theta J_{\bm{g}}(\bm{\nu})+(1-\theta)J_{\bm{g}}(\tilde{\bm{\nu}})].
\end{align}
For each $\bm{g}$ in (\ref{eq.TVstrictconvex}) the left-hand side is smaller than or equal to the right-hand side. Therefore, in (\ref{eq.TVstrictconvex}), $J_{\bm{g}}(\theta\bm{\nu}+(1-\theta)\tilde{\bm{\nu}})$ and $\theta J_{\bm{g}}(\bm{\nu})+(1-\theta)J_{\bm{g}}(\tilde{\bm{\nu}})$ are almost surely equal(except at a measure zero set) with respect to Gaussian measure. Moreover, we have:
\begin{align}\label{eq.J0TV}
&J_{\bm{0}}(\theta\bm{\nu}+(1-\theta)\tilde{\bm{\nu}})=\mathrm{dist}^2(\bm{0},\bm{\Omega}_d^T\big(\theta\bm{\upsilon}+(1-\theta)\tilde{\bm{\upsilon}}\big)\odot \mathcal{C})\nonumber\\
&\le \inf_{\bm{z}_1,\bm{z}_2\in \mathcal{C}}\|\theta\bm{\Omega}_d^T(\bm{\upsilon}\odot \bm{z}_1)+(1-\theta)\bm{\Omega}_d^T(\tilde{\bm{\upsilon}}\odot \bm{z}_2)\|_2^2\nonumber\\
&<\theta\inf_{\bm{z}_1\in \mathcal{C}}\|\bm{\Omega}_d^T(\bm{\upsilon}\odot \bm{z}_1)\|_2^2+(1-\theta)\inf_{\bm{z}_2\in \mathcal{C}}\|\bm{\Omega}_d^T(\tilde{\bm{\upsilon}}\odot \bm{z}_2)\|_2^2\nonumber\\
&=\theta J_{\bm{0}}(\bm{\nu})+(1-\theta)J_{\bm{0}}(\tilde{\bm{\nu}}),
\end{align}
where, the first inequality comes from (\ref{eq.TVbenefit}) and the second inequality stems from the strict convexity of $\|.\|_2^2$. From (\ref{eq.TVbenefit}), it is easy to verify that the set $\bm{\Omega}_d^T(\bm{\nu}\odot \mathcal{C})$ is a convex set. The distance to a convex set e.g. $\mathcal{E}$ i.e. $\mathrm{dist}(\bm{g},\mathcal{E})$ is a $1$-lipschitz function (i.e. $|\mathrm{dist}(\bm{g},\mathcal{E})-\mathrm{dist}(\tilde{\bm{g}},\mathcal{E})|\le\|\bm{g}-\tilde{\bm{g}}\|_2~:~\forall~\bm{g},\tilde{\bm{g}}\in\mathbb{R}^n$) and hence continuous with respect to $\bm{g}$. Therefore, $J_{\bm{g}}(\bm{\nu})$ is continuous with respect to $\bm{g}$. So there exists an open ball around $\bm{g}=\bm{0}\in\mathbb{R}^n$ that similar to (\ref{eq.J0TV}), we may write the following relation for some $\epsilon>0$
\begin{align}
&\exists \bm{u}\in\mathds{B}_{\epsilon}^n: \nonumber\\
&J_{\bm{u}}(\theta\bm{\nu}+(1-\theta)\tilde{\bm{\nu}})<\theta J_{\bm{u}}(\bm{\nu})+(1-\theta)J_{\bm{u}}(\tilde{\bm{\nu}}).
\end{align}
The above statement contradicts with (\ref{eq.TVstrictconvex}) and hence we have strict convexity. Continuity along with convexity of $J$ implies that $J$ is convex on the whole domain $\bm{\nu}\in\mathbb{R}_{+}^L$.\par
\textit{Differentiability} The function $J_{\bm{g}}(\bm{\nu})$ is continuously differentiable and the gradient for $\bm{\nu}\in\mathbb{R}_{++}^L$ is:
\begin{align}\label{eq.diffJg}
&\nabla_{\bm{\nu}}J_{\bm{g}}(\bm{\nu})=\nonumber\\
&\frac{\partial J_{\bm{g}}(\bm{\nu})}{\partial\bm{\nu}}=-2\bm{D}^T\bm{z}^*\odot\bm{\Omega}_d(\bm{g}-\mathcal{P}_{\bm{\Omega}_d^T(\bm{D\nu}\odot\mathcal{C})}(\bm{g})),
\end{align}
where $\Omega_d^T(\bm{z}^*\odot\bm{D\nu})=\mathcal{P}_{\bm{\Omega}_d^T(\bm{D\nu}\odot\mathcal{C})}$. Continuity of $\frac{\partial J_{\bm{g}}(\bm{\nu})}{\partial\bm{\nu}}$ at $\bm{\nu}\in\mathbb{R}_{+}^L$ stems from the fact that the projection onto a convex set is continuous. For each compact set $\mathcal{I}\subseteq\mathbb{R}_{+}^L$ we have:
\begin{align}
&\mathds{E}\sup_{\bm{\nu}\in\mathcal{I}}\|\nabla_{\bm{\nu}} J_{\bm{g}}(\bm{\nu})\|_2\le2\sqrt{n^2-n}\|\bm{\Omega}_d\|_{2\rightarrow2}\|\bm{D}\|_{2\rightarrow2}+\nonumber\\
&2(n-1)\|\Omega_d\|_{2\rightarrow2}^2\|\bm{D}\|_{2\rightarrow2}\sup_{\bm{\nu}\in\mathcal{I}}\nu_{\max}<\infty,
\end{align}
where $\nu_{\max}:=\underset{i\in[L]}{\max}~\nu(i)$. Therefore, we have:
\begin{align}
\nabla_{\bm{\nu}} J(\bm{\nu})=(\frac{\partial}{\partial\bm{\nu}})\mathds{E} J_{\bm{g}}(\bm{\nu})=\mathds{E}[\nabla_{\bm{\nu}} J_{\bm{g}}(\bm{\nu})] ~:~\forall \bm{\nu}\in\mathbb{R}_{+}^L,
\end{align}
where in the last equality, the Lebesgue's dominated convergence theorem is used. \par
\textit{Attainment of the minimum}. Suppose that $\bm{\nu}\ge{\|\bm{g}\|_2}\bm{1}_{L\times 1}$ where $\bm{g}\in\mathbb{R}^n$ is a fixed vector. With this assumption we may write:
\begin{align}\label{eq.TVattain}
&\mathrm{dist}(\bm{g},\bm{\Omega}_d^T((\bm{D\nu})\odot \mathcal{C}))=\inf_{\bm{z}\in \mathcal{C}}\|\bm{g}-\bm{\Omega}_d^T(\bm{\upsilon}\odot \bm{z})\|_2\ge\nonumber\\
&\inf_{\bm{z}\in \mathcal{C}}(\|\bm{\Omega}_d^T(\bm{\upsilon}\odot \bm{z})\|_2-\|\bm{g}\|_2)\ge\nu_{\min}\|\bm{\Omega}_d^\dagger\|_{2\rightarrow2}^{-1}-\|\bm{g}\|_2\ge0,
\end{align}
where, $\nu_{\min}:=\underset{i\in[L]}{\min}~\nu(i)$. By squaring (\ref{eq.TVattain}), we reach:
\begin{align}\label{eq.TVattainJ_g}
J_{\bm{g}}(\bm{\nu})\ge (\nu_{\min}\|\bm{\Omega}_d^\dagger\|_{2\rightarrow2}^{-1}-\|\bm{g}\|_2)^2~~:~\forall \nu\ge\|\bm{\Omega}_d^\dagger\|_{2\rightarrow2}{\|\bm{g}\|_2}\bm{1}_{L\times 1}.
\end{align}
Using the relation $\mathds{E}\|\bm{g}\|_2\ge\frac{n}{\sqrt{n+1}}$ and Marcov's inequality we obtain:
\begin{align}
\mathds{P}(\|\bm{g}\|_2\le\sqrt{n})\ge1-\sqrt{\frac{n}{n+1}}\nonumber.
\end{align}
Then we reach:
\begin{align}\label{eq.TVattainJ}
&J(\bm{\nu})\ge\mathds{E}[J_{\bm{g}}(\bm{\nu})|\|\bm{g}\|_2\le\sqrt{n}]\mathds{P}(\|\bm{g}\|_2\le\sqrt{n})\nonumber\\
&\ge(1-\sqrt{\frac{n}{n+1}})\mathds{E}\big[(\nu_{\min}\|\bm{\Omega}_d^\dagger\|_{2\rightarrow2}^{-1}-\|\bm{g}\|_2)^2|\|\bm{g}\|_2\le\sqrt{n}\big]\nonumber\\
&\ge(1-\sqrt{\frac{n}{n+1}})(\nu_{\min}\|\bm{\Omega}_d^\dagger\|_{2\rightarrow2}^{-1}-\sqrt{n})^2,
\end{align}
where in (\ref{eq.TVattainJ}), the first inequality stems from total probability theorem, the second inequality comes from (\ref{eq.attainJ_g}). From (\ref{eq.TVattainJ}), we find out that $J(\bm{\nu})>J(\bm{0})$ when $\nu>(2^{\frac{1}{4}}+1)\|\bm{\Omega}_d^\dagger\|_{2\rightarrow2}{\sqrt{p}}\bm{1}_{L\times1}$. Therefore, the unique minimizer of $J$ must occur in the interval $[\bm{0},(2^{\frac{1}{4}}+1)\|\bm{\Omega}_d^\dagger\|_{2\rightarrow2}{\sqrt{p}}\bm{1}_{L\times1}]$
\end{proof}
\begin{prop}\label{prop.uniqnessoptimalTVweights}
Let $\bm{x}$ be a non-uniform $s$-gradient sparse vector in $\mathbb{R}^n$ with parameters $\{\rho_i\}_{i=1}^L$ and $\{\alpha_i,\alpha_i',\beta_i,\beta_i',\gamma_i,\gamma_i',\varsigma_i,\varsigma_i',\xi_i,\breve{\xi}_i\}_{i=1}^L$. Then there exist unique optimal weights $\bm{\omega}^* \in \mathbb{R}_{+}^L$ (up to a positive scaling) that minimize $\hat{m}_{\mathrm{TV},s,\bm{D\omega}}$. Moreover, the optimal weights $\bm{\omega}^*\in\mathbb{R}$ are obtained by simultaneously solving the following integral equations.
\begin{align}\label{eq.TVoptimalweights}
&2\alpha'_i(\omega_i-\omega_{i-1})+8\beta_i\omega_i+2\beta'_i(\omega_i+\omega_{i-1})\nonumber\\
&+(\gamma_i+\varsigma_i)\frac{\partial\phi_1(\omega_i,\omega_i)}{\partial \omega_i}+\gamma'_i\frac{\partial\phi_1(\omega_{i-1},\omega_i)}{\partial\omega_i}+\varsigma'_i \frac{\partial\phi_1(\omega_{i-1},\omega_i)}{\partial\omega_i}\nonumber\\
&+(1-\alpha_i-\beta_i-\gamma_i-\varsigma_i-\frac{1}{|\mathcal{P}_i|})\frac{\partial\phi_2(2\omega_i)}{\partial\omega_i}\nonumber\\
&+(1-\alpha'_i-\beta'_i-\gamma'_i-\varsigma'_i-\xi_i)\frac{\partial\phi_2(\omega_i+\omega_{i-1})}{\partial\omega_i}+2(\xi_i+\breve{\xi}_i)\omega_i\nonumber\\
&+(\bar{\xi}_i+\bar{\breve{\xi}}_i)\frac{\partial\phi_2(\omega_i)}{\partial\omega_i}=0~:~\forall i=1,..., L,
\end{align}
where in (\ref{eq.TVoptimalweights}),
\begin{align}
&\frac{\partial\phi_1(a,a)}{\partial a}=\frac{1}{\sqrt{2\pi}}\int_{a}^{\infty}\bigg[-2(u-a)(e^{-\frac{(u-a)^2}{2}}+e^{-\frac{(u+a)^2}{2}})+\nonumber\\
&(u-a)^2((u-a)e^{-\frac{(u-a)^2}{2}}-(u+a)e^{-\frac{(u+a)^2}{2}}))\bigg]du\nonumber\\
&\frac{\partial\phi_1(a,b)}{\partial b}=\frac{1}{\sqrt{2\pi}}\int_{a}^{\infty}\bigg[-2(u-b)(e^{-\frac{(u-a)^2}{2}}+e^{-\frac{(u+a)^2}{2}})\bigg]du\nonumber\\
&\frac{\partial\phi_1(a,b)}{\partial a}=\frac{1}{\sqrt{2\pi}}\int_{b}^{\infty}\bigg[(u-a)^2((u-a)e^{-\frac{(u-a)^2}{2}}-\nonumber\\
&(u+a)e^{-\frac{(u+a)^2}{2}})\bigg]du\nonumber\\
&\frac{\partial\phi_2(z)}{\partial z}=-2\sqrt{\frac{2}{\pi}}\int_{z}^{\infty}(u-z)e^{-\frac{u^2}{2}}du.
\end{align}
\end{prop}
\begin{proof}
Define $\mathcal{C}:=\partial\|.\|_{\mathrm{TV}}(\bm{x})$ and use Lemma \ref{lemma.mhat tvweighted} to obtain:
\begin{align}
&\inf_{\bm{\omega}\in\mathbb{R}_{+}^L}\hat{m}_{\mathrm{TV},s,\bm{D\omega}}=\nonumber\\
&\inf_{\bm{\omega}\in\mathbb{R}_{+}^L}\underset{{t\in\mathbb{R}_{+}}}{\inf}\Psi_{t,\bm{D\omega}}(\bm{\alpha},\bm{\alpha}',\bm{\beta},\bm{\beta}',\bm{\gamma},\bm{\gamma}',\bm{\varsigma},\bm{\varsigma}',\bm{\xi},\breve{\bm{\xi}})\nonumber\\
&=\inf_{\bm{\nu}\in\mathbb{R}_{+}^L}J_{\mathrm{TV}}(\bm{\nu}),
\end{align}
where, $\Psi_{t,\bm{D\omega}}(\bm{\alpha},\bm{\alpha}',\bm{\beta},\bm{\beta}',\bm{\gamma},\bm{\gamma}',\bm{\varsigma},\bm{\varsigma}',\bm{\xi},\breve{\bm{\xi}})$ is defined in (\ref{eq.mhat wTVs}). Also, we used a change of variable $\bm{\nu}=t\bm{\omega}$ to convert multivariate optimization problem to a single variable optimization problem. Thus, the function $J_{\mathrm{TV}}(\bm{\nu})$ is obtained via the following equation:
\begin{align}
&J_{\mathrm{TV}}(\bm{\nu})=\sum_{i=1}^{L}\rho_i\bigg(\alpha_i+\beta_i[1+4\nu_i^2]+\nonumber\\
&(\gamma_i+\varsigma_i)\phi_1(\nu_i,\nu_i)+(1-\alpha_i-\beta_i-\gamma_i-\varsigma_i-\frac{1}{|\mathcal{P}_i|})\phi_2(2\nu_i)+\nonumber\\
&+\alpha'_i[1+(\nu_i-\nu_{i-1})^2]+\beta'_i[1+(\nu_i+\nu_{i-1})^2]+\nonumber\\
&\gamma'_i\phi_1(\nu_{i-1},\nu_i)+\varsigma'_i\phi_1(\nu_{i-1},\nu_i)+\nonumber\\
&(1-\alpha'_i-\beta'_i-\gamma'_i-\varsigma'_i-\xi_i)\phi_2(\nu_i+\nu_{i-1})+\nonumber\\
&(\xi_i+\breve{\xi}_i)(1+\nu_i^2)+ (\bar{\xi}_i+\bar{\breve{\xi}}_i)\phi_2(\nu_i)\bigg).
\end{align}
By considering Lemma \ref{lemma.Jt(v)} and $\bm{D}:=[\bm{1}_{\mathcal{P}_1},..., \bm{1}_{\mathcal{P}_L}]_{n-1\times L}$, the function $J_{\mathrm{TV}}(\bm{\nu})$ is continuous and strictly convex and thus the unique minimizer can be obtained using $\nabla J_{\mathrm{TV}}(\bm{\nu})=\bm{0}\in\mathbb{R}^{L}$ which leads to (\ref{eq.TVoptimalweights}).
\end{proof}
\begin{thm}\label{thm.TVweighted}
Let $\bm{x}\in\mathbb{R}^n$ be a non-uniform $s$-gradient sparse vector with parameters $\{\rho_i\}_{i=1}^L$ and $\{\alpha_i,\alpha_i',\beta_i,\beta_i',\gamma_i,\gamma_i',\varsigma_i,\varsigma_i'\}_{i=1}^L$. Then, the number of measurements required for $\mathsf{P}_{\mathrm{TV},\bm{D}\bm{\omega}^*}$ is upper bounded by the sum of number of measurements required for $\mathsf{P}_{\mathrm{TV}}$ to recover each $\{\bm{x}_{\mathcal{P}_i}\in\mathbb{R}^{n}\}_{i=1}^L$ by considering the variations between partitions and lower bounded by the whole number of measurements one needs to recover each $\{\bm{x}_{\mathcal{P}_i}\in\mathbb{R}^{n}\}_{i=1}^L$ separately ignoring the variations between partitions. Written in a mathematical form, we have:
\begin{align}\label{eq.TVtheoremerrorbound}
\sum_{i=1}^{L}m_{\mathrm{TV},\breve{s}_i}-\frac{2\kappa(\bm{\Omega}_d)}{\sqrt{nL}}\le m_{\mathrm{TV},s,\bm{w}^*}\le\sum_{i=1}^{L}\big(m_{\mathrm{TV},\hat{s}_i}+\frac{2\kappa(\bm{\Omega}_d)}{\sqrt{n\hat{s}_i}}\big),
\end{align}
where,
\begin{align}
&\hat{s}_i=\rho_i(n-1)[\alpha_i+\alpha'_i+\beta_i+\beta'_i+\gamma_i+\gamma'_i+\varsigma_i+\varsigma'_i+\xi_i]\nonumber\\
&\breve{s}_i=\rho_i(n-1)[\alpha_i+\beta_i+\gamma_i+\xi_i]\nonumber\\
&s=|\mathcal{S}_1\cup\mathcal{S}_2\cup\mathcal{S}_3\cup\mathcal{S}_6|.\nonumber
\end{align}
\end{thm}
\begin{proof}
As previously defined in (\ref{eq.mhat wTVs}), with optimal weights, the upper bound for normalized number of measurements required for $\mathsf{P}_{\mathrm{TV},\bm{D}\bm{\omega}^*}$ to succeed is given by:
\begin{align}\label{eq.TVupperpart1}
&\hat{m}_{\mathrm{TV},s,\bm{w}^*}=\inf_{\bm{\omega}\in\mathbb{R}_{+}^L}\hat{m}_{\mathrm{TV},s,\bm{D\omega}}=\nonumber\\
&\sum_{i=1}^{L}\inf_{\nu_i,\nu_{i-1}\in\mathbb{R_{+}}}\Bigg[ \rho_i\Bigg(\alpha_i+\beta_i[1+4\nu_i^2]+(\gamma_i+\varsigma_i)\phi_1(\nu_i,\nu_i)+\nonumber\\
&(1-\alpha_i-\beta_i-\gamma_i-\varsigma_i-\frac{1}{|\mathcal{P}_i|})\phi_2(2\nu_i)+\nonumber\\
&\alpha'_i[1+(\nu_i-\nu_{i-1})^2]+\beta'_i[1+(\nu_i+\nu_{i-1})^2]+\nonumber\\
&\gamma'_i\phi_1(\nu_{i-1},\nu_i)+\varsigma'_i\phi_1(\nu_{i-1},\nu_i)+\nonumber\\
&(1-\alpha'_i-\beta'_i-\gamma'_i-\varsigma'_i-\xi_i)\phi_2(\nu_i+\nu_{i-1})+(\xi_i+\breve{\xi}_i)(1+\nu_i^2)\nonumber\\
&+(\bar{\xi}_i+\bar{\breve{\xi}}_i)\phi_2(\nu_i)\Bigg)\Bigg]\le \sum_{i=1}^{L}\inf_{\nu_i\in\mathbb{R_{+}}}\Bigg(\frac{|\mathcal{P}_i\cap \{\mathcal{S}_1\cup\mathcal{S}_1'\}|}{n-1}+\nonumber\\
&\frac{|\mathcal{P}_i\cap \{\mathcal{S}_2\cup \mathcal{S}_2'\}|}{n-1}[1+4\nu_i^2]+\nonumber\\
&\frac{|\mathcal{P}_i\cap\{ \mathcal{S}_3\cup\mathcal{S}_4\cup\mathcal{S}_3'\cup\mathcal{S}_4'\}|}{n-1}\phi_1(\nu_i,\nu_i)\nonumber\\
&+\frac{|\mathcal{P}_i\cap\{\mathcal{S}_5\cup\mathcal{S}_5'\}|}{n-1}\phi_2(2\nu_i)+\frac{|\mathcal{P}_i\cap\{\mathcal{S}_6\cup\mathcal{S}_7\}|}{n-1}(1+\nu_i^2)\nonumber\\
&+\frac{|\mathcal{P}_i\cap\{\tilde{\mathcal{S}}_6\cup\tilde{\mathcal{S}}_7\}|}{n-1}\phi_2(\nu_i)\Bigg)=\nonumber\\
&\sum_{i=1}^{L}\inf_{\nu_i\in\mathbb{R_{+}}}\Psi_{t}\big(\frac{|P_i\cap\{\mathcal{S}_1\cup\mathcal{S}_1'\}|}{n-1},\frac{|\mathcal{P}_i\cap \{\mathcal{S}_2\cup \mathcal{S}_2'\}|}{n-1}\nonumber\\
&,\frac{|\mathcal{P}_i\cap\{ \mathcal{S}_3\cup\mathcal{S}_4\cup\mathcal{S}_3'\cup\mathcal{S}_4'\}|}{n-1},\frac{|\mathcal{P}_i\cap\{\mathcal{S}_6\cup\mathcal{S}_7\}|}{n-1}\big)\nonumber\\
&=\sum_{i=1}^{L}\hat{m}_{\mathrm{TV},\hat{s}_i}.
\end{align}
On the other hand, by using the relation $\underset{\mathcal{A}}{\inf}(f+g)\ge\underset{\mathcal{A}}{\inf}(f)+\underset{\mathcal{A}}{\inf}(g)~:~\forall f,g:\mathcal{A}\rightarrow\mathbb{R}$, we have:
\begin{align}
&\inf_{\nu_i,\nu_{i-1}\in\mathbb{R_{+}}}\Bigg[ \rho_i\Bigg(\alpha_i+\beta_i[1+4\nu_i^2]+(\gamma_i+\varsigma_i)\phi_1(\nu_i,\nu_i)+\nonumber\\
&(1-\alpha_i-\beta_i-\gamma_i-\varsigma_i-\frac{1}{|\mathcal{P}_i|})\phi_2(2\nu_i)+\alpha'_i[1+(\nu_i-\nu_{i-1})^2]\nonumber\\
&+\beta'_i[1+(\nu_i+\nu_{i-1})^2]+\gamma'_i\phi_1(\nu_{i-1},\nu_i)+\varsigma'_i\phi_1(\nu_{i-1},\nu_i)+\nonumber\\
&(1-\alpha'_i-\beta'_i-\gamma'_i-\varsigma'_i-\xi_i)\phi_2(\nu_i+\nu_{i-1})+(\xi_i+\breve{\xi}_i)(1+\nu_i^2)\nonumber\\
&+(\bar{\xi}_i+\bar{\breve{\xi}}_i)\phi_2(\nu_i)\Bigg)\Bigg]\nonumber\\
&\ge \inf_{\nu_i\in\mathbb{R_{+}}}\rho_i\bigg(\alpha_i+\beta_i[1+4\nu_i^2]+(\gamma_i+\varsigma_i)\phi_1(\nu_i,\nu_i)+\nonumber\\
&(1-\alpha_i-\beta_i-\gamma_i-\varsigma_i-\frac{1}{|\mathcal{P}_i|})\phi_2(2\nu_i)+(\xi_i+\breve{\xi}_i)(1+\nu_i^2)\nonumber\\
&+(\bar{\xi}_i+\bar{\breve{\xi}}_i)\phi_2(\nu_i)\bigg):=m_{\mathrm{TV},\breve{s}_i}.
\end{align}
As a consequence, regarding the error bounds obtained in Propositions \ref{prop.error for mhat TV} and \ref{prop.error for mhat TVw}, we reach (\ref{eq.TVtheoremerrorbound}).
\end{proof}
\section{Simulations}\label{section.simulation}
In this section, we numerically verify our theoretical bounds on the optimal weights and number of measurements required for successful recovery in the three inherent structures. Remark that, MATLAB package cvx \cite{cvx} is used to implement the optimization problems.
\subsection{Entrywise sparsity in a redundant dictionary}
In this subsection, we intend to find how the required  number of Gaussian linear measurements for successful recovery of $\mathsf{P}_{1,\Omega}$, scales with the signal sparsity $s$ in the redundant discrete cosine transform (DCT) dictionary $\Omega\in\mathbb{R}^{p\times n}$. For this purpose, we consider a grid of $(m,s)$ values and count number of successful recovery of $\mathsf{P}_{1,\Omega}$ in $100$ trials to find the empirical probability. The heatmap in Figure \ref{fig.heatmapl1analysis} shows the empirical probability that is consistent with the theory obtained by ($\ref{eq.mhatps}$). In a separate experiment, we generate a random vector $\bm{x}\in\mathbb{R}^n$ that is $10$-sparse in the DFT domain $\bm{\Omega}\in \mathbb{R}^{100\times 90}$ and consider two sets $\mathcal{P}_1$ and $\mathcal{P}_2$ in the DFT domain $\bm{\Omega}\in \mathbb{R}^{100\times 90}$ that partition $[100]$ with $\alpha_1=\frac{7}{10}$ and $\alpha_2=\frac{3}{90}$. In each trial, we recover the signal from Gaussian i.i.d linear measurements by solving $\mathsf{P}_{1,\bm{\Omega}}$ and $\mathsf{P}_{1,\bm{\Omega},\bm{D}\bm{\omega}^*}$. The optimal weights $\bm{\omega}^*\in\mathbb{R}^2$ are obtained by solving the equation (\ref{eq.optimalweightsl1analysis}) using MATLAB function \textsf{fzero}. Figure \ref{fig.l1analysis} shows that $\mathsf{P}_{1,\bm{\Omega},\bm{D}\bm{\omega}^*}$ with optimal weights needs less measurements than $\mathsf{P}_{1,\bm{\Omega}}$.
\begin{figure}[t]
\centering
\includegraphics[scale=.54]{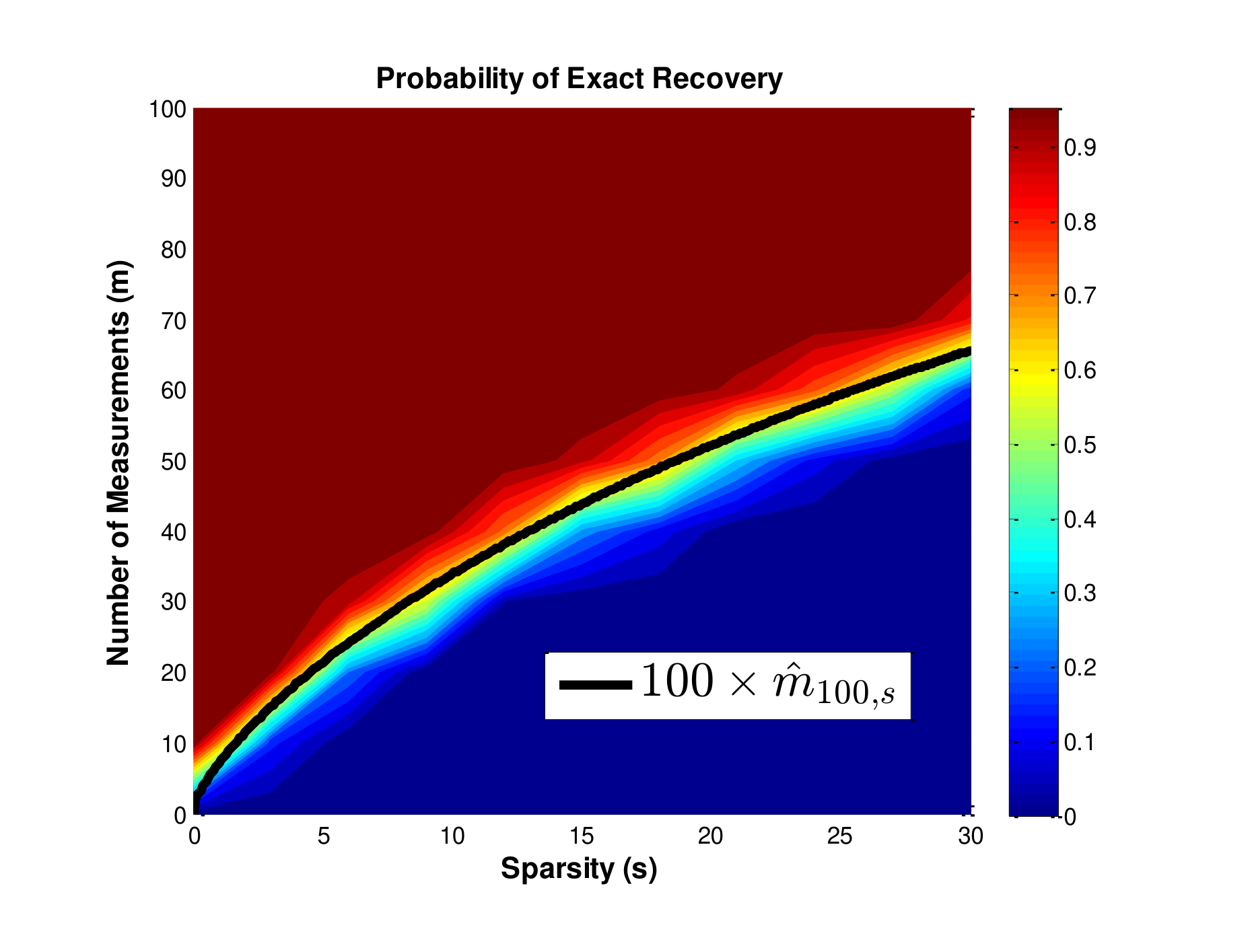}
\caption{This plot shows the empirical probability that $\mathsf{P}_{1,\Omega}$ recovers $\bm{x}\in\mathbb{R}^{90}$ with $s$ nonzero entries in redundant DCT basis from $m$ Gaussian linear measurements. The black line shows the number of measurements obtained by Lemma \ref{lemma.l1anaupper}.}
\label{fig.heatmapl1analysis}
\end{figure}
\begin{figure}[t]
\centering
\includegraphics[scale=.5]{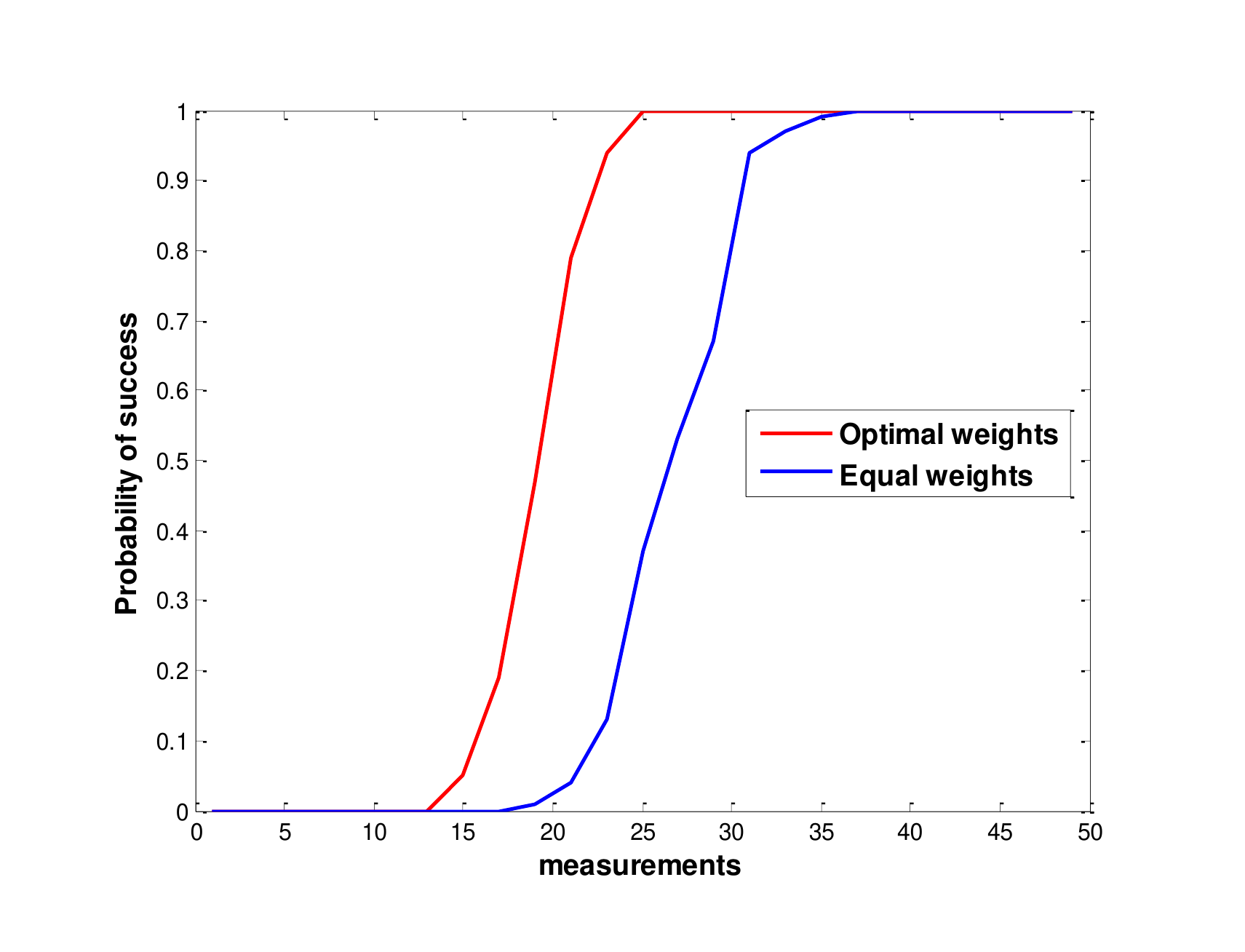}
\caption{This plot shows the probability that $\mathsf{P}_{1,\bm{\Omega}}$ and $\mathsf{P}_{1,\bm{\Omega},\bm{D}\bm{\omega}^*}$ succeed to recover $\bm{x}\in\mathbb{R}^{90}$ from Gaussian linear measurements. The parameters we used in this figure, are: $n=100$, $\sigma=\rho_1=.1,~\alpha_1=\frac{7}{10}$, $\rho_2=.9,~\alpha_2=\frac{3}{90}$. The optimal weights obtained via (\ref{eq.optimalweightsl1analysis}) with the aforementioned parameters are $\bm{\omega}=[.1599;1]$ }
\label{fig.l1analysis}
\end{figure}
\subsection{Block sparsity}
In this subsection, we simulate how the required number of measurements for successful recovery of $\mathsf{P}_{1,2}$ scales with block sparsity. The heatmap in Figure \ref{fig.heatmapl12} shows the empirical probability of success which is consistent with the theory obtained by (\ref{lemma.mhat qs}). In the second experiment, we generate an $s=50$-block sparse random vector $\bm{x}\in\mathbb{R}^{1280}$ with $128$ blocks of equal size $10$. Then, we consider three sets $\mathcal{P}_1$, $\mathcal{P}_2$ and $\mathcal{P}_3$ with $\alpha_1=\frac{27}{50},~\alpha_2=\frac{18}{20},~\alpha_3=\frac{5}{58}$ that partition the set of blocks $[128]$. We implement $100$ trials. In each trials, we solve $\mathsf{P}_{1,2,\bm{\omega}^*}$ with optimal weights $\bm{\omega}^*$ and recover $\bm{x}\in\mathbb{R}^{1280}$ from $m$ Gaussian linear measurements. The optimal weights are obtained via the equation (\ref{eq.12optimalweights}) by MATLAB function \textsf{fzero}. Figure \ref{fig.l12} shows that $\mathsf{P}_{1,2,\bm{D}\bm{\omega}^*}$ with optimal weights needs less measurements than $\mathsf{P}_{1,2}$ for recovering $\bm{x}\in\mathbb{R}^{1280}$.
 \begin{figure}[t]
 \centering
\includegraphics[scale=.58]{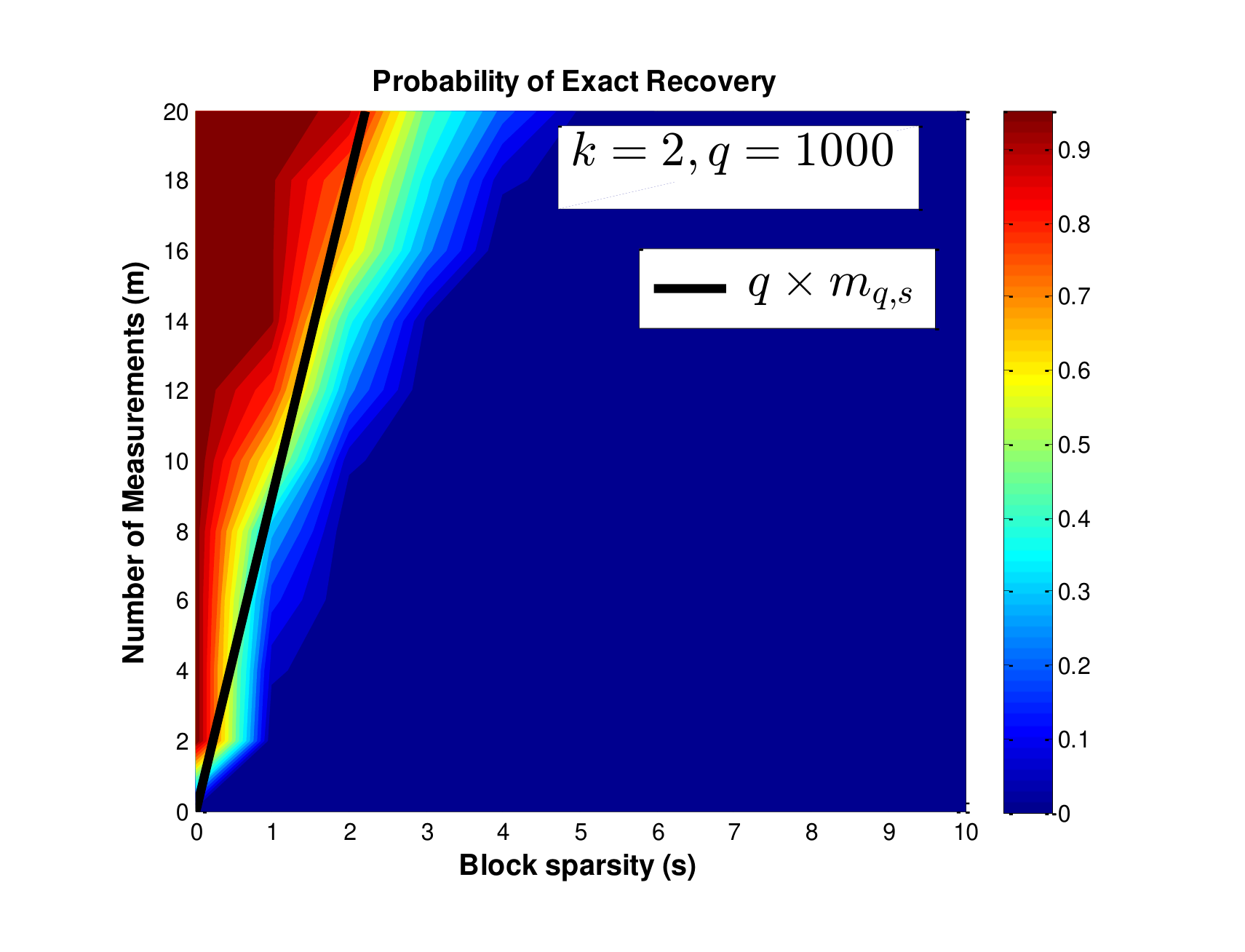}
\caption{This plot shows the empirical probability that $\mathsf{P}_{1,2}$ recovers $\bm{x}\in\mathbb{R}^{2000}$ with $s$ blocks with nonzero $\ell_2$ norm from $m$ Gaussian linear measurements. The black line shows the number of measurements obtained by Lemma \ref{lemma.mhat qs}.}
\label{fig.heatmapl12}
\end{figure}
\begin{figure}[t]
\centering
\includegraphics[scale=.5]{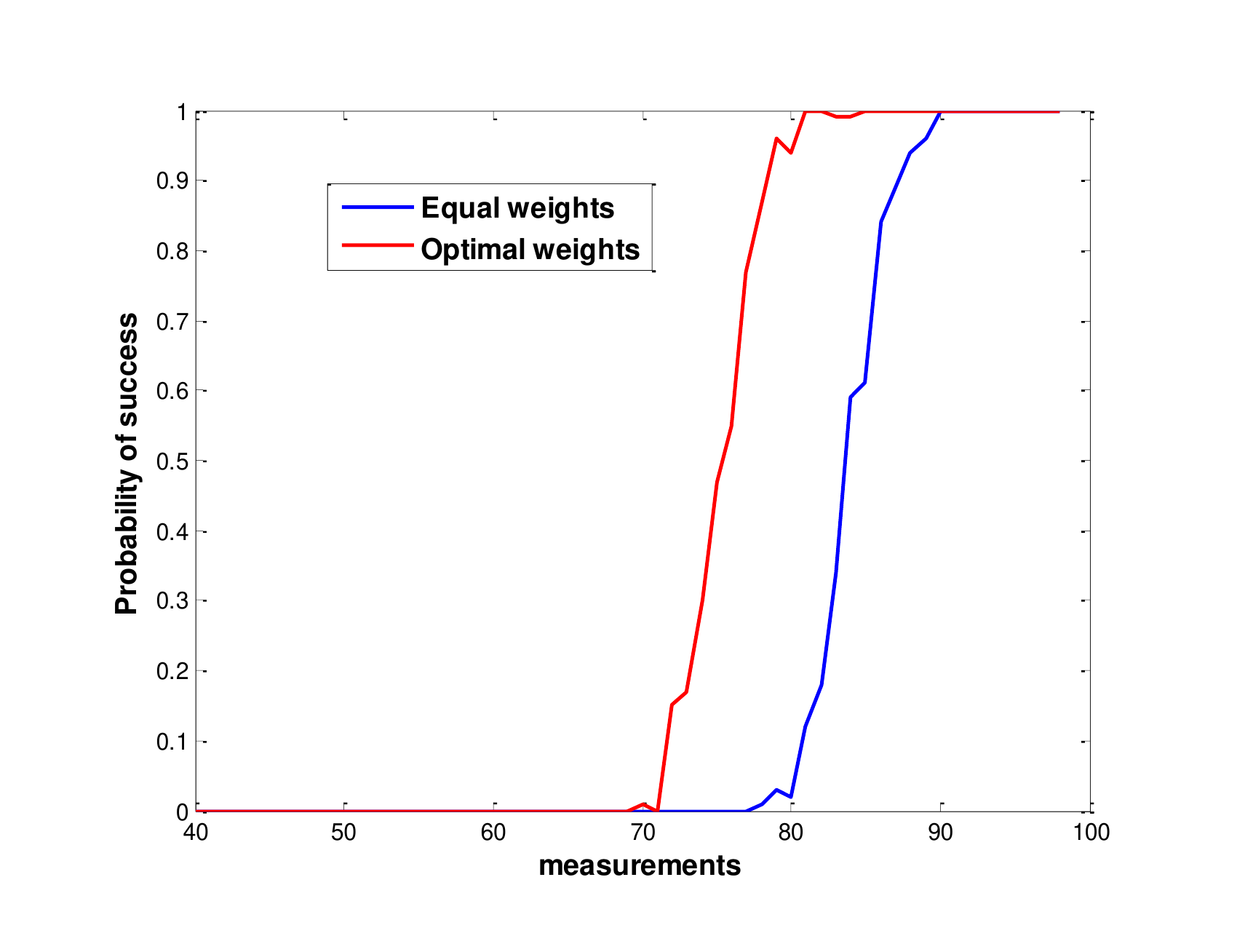}
\caption{This plot shows the probability that $\mathsf{P}_{1,2}$ and $\mathsf{P}_{1,2,\bm{D}\bm{\omega}^*}$ succeed to recover $\bm{x}\in\mathbb{R}^{1280}$ from Gaussian linear measurements. The parameters we used in this figure, are: $q=128$, $\sigma=\rho_1=\frac{50}{128}~\alpha_1=\frac{27}{50}$, $\rho_2=\frac{20}{128},~\alpha_2=\frac{9}{10},~\rho_3=\frac{29}{64},~\alpha_3=\frac{5}{58}$. The optimal weights obtained via (\ref{eq.12optimalweights}) with the aforementioned parameters are $\bm{\omega}=[.46317;.100671;1]$ }
\label{fig.l12}
\end{figure}
\subsection{Gradient sparsity}
In this subsection, we investigate the probability that $\mathsf{P}_{\mathrm{TV}}$ and $\mathsf{P}_{\mathrm{TV},\bm{D}\bm{\omega}^*}$ successfully recover an $s$-gradient sparse vector from i.i.d Gaussian linear measurements. We consider two partitions $\mathcal{P}_1$ and $\mathcal{P}_2$ with known probability of intersection with consecutive and individual supports. In the first experiment, we consider a strict nonuniform case with random partition. In this case, we fixed the accuracy of $\mathcal{P}_1$ and $\mathcal{P}_2$ at  $\frac{|\mathcal{P}_1\cap\mathcal{S}_g|}{|\mathcal{P}_1|}=\frac{5}{6}$ and $\frac{|\mathcal{P}_2\cap\mathcal{S}_g|}{|\mathcal{P}_2|}=\frac{1}{13}$, respectively. The set $\mathcal{P}_1$ is most probable to include the elements of the gradient support $\mathcal{S}_g$ than the set $\mathcal{P}_2$. We calculate the optimal weights $\bm{\omega}^*$ by solving the relation (\ref{eq.TVoptimalweights}) using MATLAB function \textsf{fsolve} for each measurement and trial. We repeat the experiment $20$ times. Figure \ref{fig.TV} indicates the empirical probability computed over $20$ trials. As it is clear from Figure \ref{fig.TV}, $\mathsf{P}_{\mathrm{TV},\bm{D}\bm{\omega}^*}$ needs less measurement than $\mathsf{P}_{\mathrm{TV}}$ to reach the same success rate. In the second experiment, we consider a two partitioned $4$-gradient sparse vector $\bm{x}\in\mathbb{R}^{21}$ with random accuracy and partition size \footnote{The smooth signal may be uniform or non-uniform.}. We repeat the experiment $50$ times and calculate the optimal weights $\bm{\omega}^*$ via (\ref{eq.TVoptimalweights}). Figure \ref{fig.TVrandom} indicates the empirical probability computed over $50$ trials.
\begin{figure}[t]
\centering
\includegraphics[scale=.5]{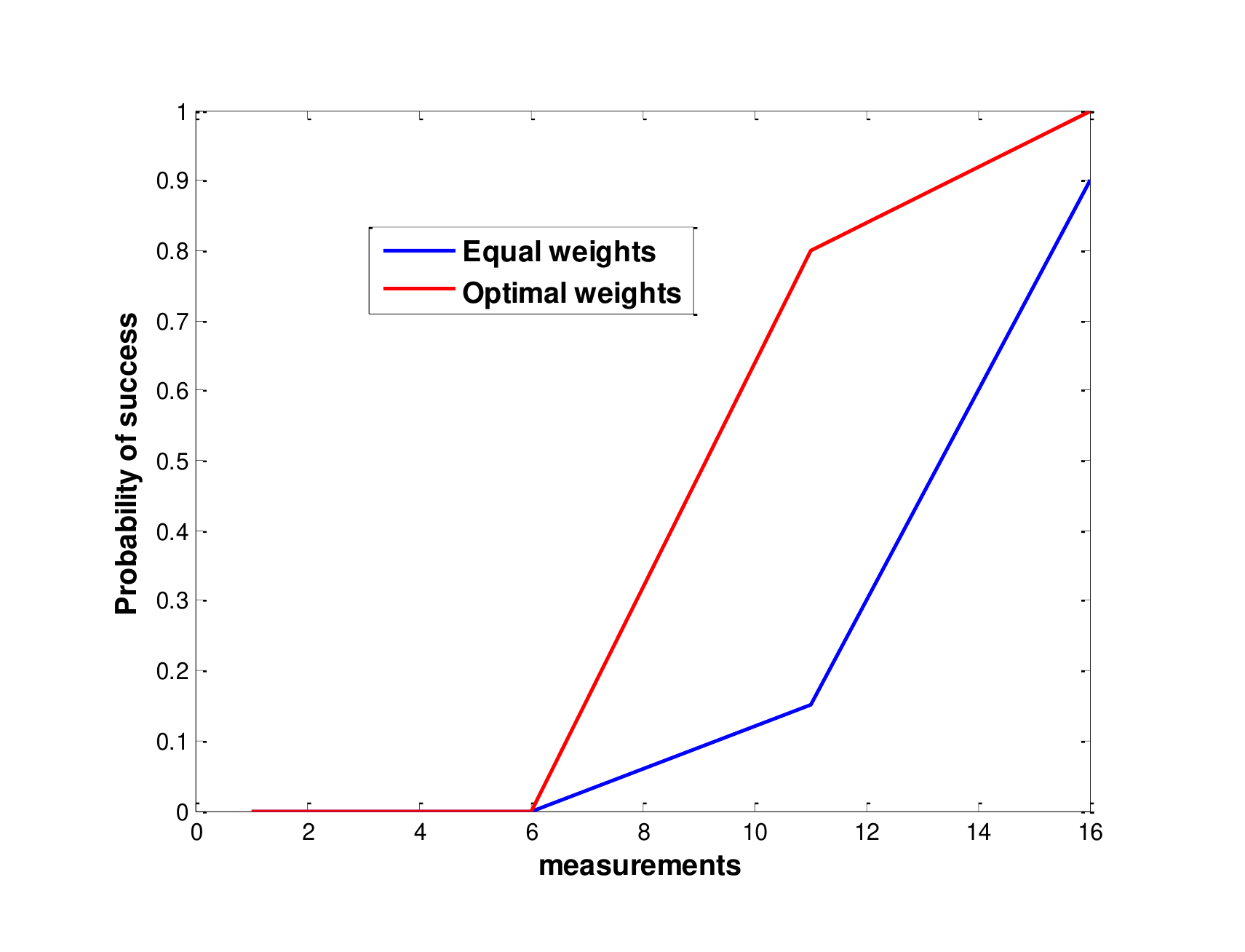}
\caption{This plot shows the probability that $\mathsf{P}_{\mathrm{TV}}$ and $\mathsf{P}_{\mathrm{TV},\bm{D}\bm{\omega}^*}$ succeed to recover $\bm{x}\in\mathbb{R}^{20}$ from i.i.d Gaussian linear measurements. The parameters we used in this figure, are: $n=20$, $\sigma=\rho_1=\frac{6}{19}~\frac{\mathcal{P}_1\cap\mathcal{S}_g}{|\mathcal{P}_1|}=\frac{5}{6}$ and $\rho_2=\frac{13}{19},\frac{\mathcal{P}_2\cap\mathcal{S}_g}{|\mathcal{P}_2|}~=\frac{1}{13}$. The optimal weights are obtained via (\ref{eq.TVoptimalweights}) with the aforementioned parameters in each trial using \textsf{fsolve} MATLAB function.}
\label{fig.TV}
\end{figure}
\begin{figure}[t]
\centering
\includegraphics[scale=.5]{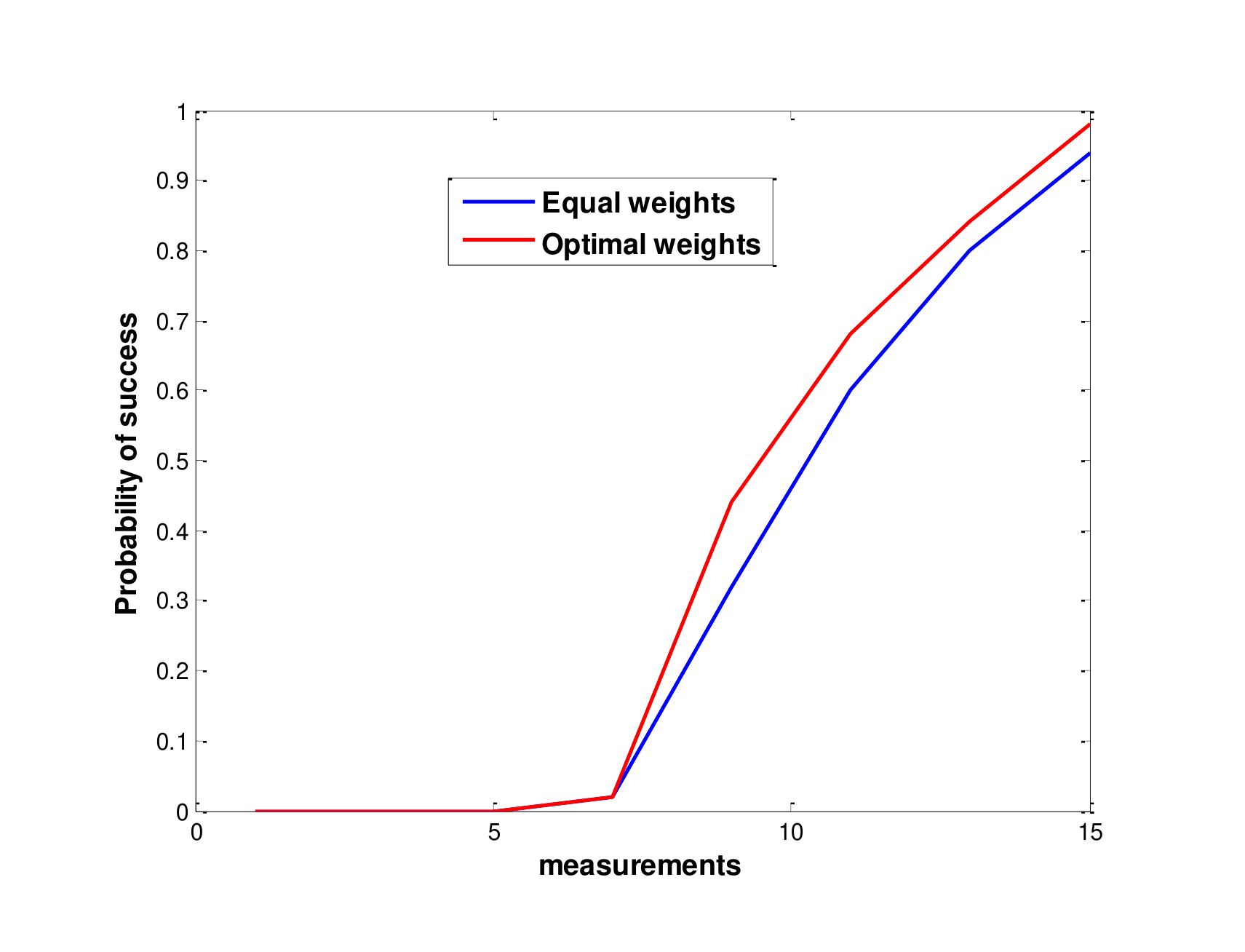}
\caption{This plot shows the probability that $\mathsf{P}_{\mathrm{TV}}$ and $\mathsf{P}_{\mathrm{TV},\bm{D}\bm{\omega}^*}$ succeed to recover an $4$-gradient sparse $\bm{x}\in\mathbb{R}^{21}$ from i.i.d Gaussian linear measurements. The probability is computed over $20$ trials for each measurement. In each trial, we change the accuracy and size of the sets $\mathcal{P}_1$ and $\mathcal{P}_2$. The optimal weights are obtained via (\ref{eq.TVoptimalweights}) in each trial using \textsf{fsolve} MATLAB function.}
\label{fig.TVrandom}
\end{figure}
\section{Conclusion}\label{section.conclusion}
In this paper, we investigated structured signal recovery with extra prior information for three popular class of low-complexity models: entrywise, block and gradient sparse vectors. We found explicit formulas for number of Gaussian linear measurements needed to recover an structured vector. An important question central to the weighted convex minimization problems is the optimal choices of weights. In this paper, we provide a general framework based on techniques developed in conic integral geometry to find optimal weights. Also, we proved that the optimal weights are unique up to a positive scaling in the aforementioned low-complexity models. Additionally, we have included numerical experiments that verify out theoretical results.

\appendices
\section{Proof of Proposition \ref{prop.upperforstatis}}\label{appendix.1}
\begin{proof}In this appendix, we prove Proposition \ref{prop.upperforstatis}. In the following, we relate $\mathcal{D}(\|\bm{\Omega}.\|_1,\bm{x})$ to $\mathcal{D}(\|.\|_1,\bm{\Omega x})$.
\begin{align}\label{eq.descent cone relation}
&\mathcal{D}(\|\bm{\Omega}.\|_1,\bm{x})^{\circ\circ}=closure(\mathcal{D}(\|\bm{\Omega}.\|_1,\bm{x}))\nonumber\\
&\mathcal{D}(\|\Omega.\|_1,\bm{x})=\mathrm{cone}^{\circ}(\bm{\Omega}^T\partial\|.\|_1(\bm{\Omega x}))\nonumber\\
&\{\bm{w}\in\mathbb{R}^n:\langle \bm{w},\bm{\Omega}^T\bm{v}\rangle\le0~:~\forall \bm{v}\in \mathrm{cone}(\partial\|.\|_1(\bm{\Omega x}))\}=\nonumber\\
&\{\bm{w}\in\mathbb{R}^n:~\bm{\Omega w}\in \mathrm{cone}^{\circ}(\partial\|.\|_1(\bm{\Omega x}))\}=\nonumber\\
&\{\bm{w}\in\mathbb{R}^n:~\bm{\Omega w}\in\mathcal{D}(\|.\|_1,\bm{\Omega x})\}\nonumber\\
&\subset \bm{\Omega}^\dagger\mathcal{D}(\|.\|_1,\bm{\Omega x}),
\end{align}
where in the above equation, we used the facts that $\mathcal{D}(\|\bm{\Omega}.\|_1,\bm{x})$ is a closed convex set and $\bm{\Omega}^\dagger\bm{\Omega}=\bm{I}$ for redundant dictionary $\bm{\Omega}\in\mathbb{R}^{p\times n}$ with $p\ge n$. In the following, we state Sudakov-Fernique inequality which helps to control the superimum of a random variable by that of a simpler random variable and is used to find an upper bound for $\omega(\mathcal{D}(\|\bm{\Omega} .\|_1,\bm{x})\cap \mathds{B}_1^n)$.
\begin{thm}\label{thm.sodakov}(Sudakov-Fernique inequality).
 Let $T$ be a set and $\mathbf{X}=(X_t)_{t\in T}$ and $\mathbf{Y}=(Y_t)_{t\in T}$ be Gaussian processes satisfying $\mathds{E}[X_t]=\mathds{E}[Y_t]~:~\forall t\in T$ and $\mathds{E}|X_t-X_s|^2\le\mathds{E}|Y_t-Y_s|^2~:~\forall s,t\in T$, then
 \begin{align}
 \mathds{E}\sup_{t\in T}X_t\le\mathds{E}\sup_{t\in T}Y_t
 \end{align}
\end{thm}
\begin{align}\label{eq.Gausianwidupper}
&\omega(\mathcal{D}(\|\bm{\Omega} .\|_1,\bm{x})\cap \mathds{B}_1^n)=\mathds{E}\sup_{\bm{v}\in \mathcal{D}(\|\bm{\Omega} .\|_1,\bm{x})\cap \mathds{B}_1^n }\langle \bm{g},\bm{v} \rangle\le\nonumber\\
&\mathds{E}\sup_{\substack{\bm{z}\in \mathcal{D}(\|.\|_1,\bm{\Omega x})\\\|\bm{\Omega}^\dagger \bm{z}\|_2\le1}}\langle \bm{g},\bm{\Omega}^\dagger \bm{z} \rangle\le\|\bm{\Omega}^\dagger\|_{2\rightarrow2}\mathds{E}\sup_{\substack{\bm{z}\in \mathcal{D}(\|.\|_1,\bm{\Omega x})\\\|\bm{\Omega}^\dagger \bm{z}\|_2\le1}}\langle \bm{h},\bm{z} \rangle\nonumber\\
&\le \|\bm{\Omega}^\dagger\|_{2\rightarrow2}\mathds{E}\sup_{\substack{\bm{z}\in \mathcal{D}(\|.\|_1,\bm{\Omega x})\\\|\bm{z}\|_2\le\|\Omega\|_{2\rightarrow2}}}\langle \bm{h},\bm{z} \rangle\nonumber\\
&\le \|\bm{\Omega}^\dagger\|_{2\rightarrow2}\|\bm{\Omega}\|_{2\rightarrow2}\mathds{E}\sup_{\bm{w}\in\mathcal{D}(\|.\|_1,\bm{\Omega x})\cap \mathds{B}_1^p}\langle \bm{h},\bm{w} \rangle=\nonumber\\
&\kappa(\bm{\Omega})\omega(\mathcal{D}(\|.\|_1,\bm{x})\cap \mathds{B}_1^p),
\end{align}
where in (\ref{eq.Gausianwidupper}), $\bm{h}\in\mathbb{R}^p$ is a standard normal vector with i.i.d components. The first inequality in (\ref{eq.Gausianwidupper}) comes from the relation (\ref{eq.descent cone relation}), the second inequality comes from Theorem \ref{thm.sodakov} with $X_{\bm{t}}=\langle \bm{g},\bm{\Omega}^\dagger \bm{t} \rangle$ and $Y_{\bm{t}}=\|\bm{\Omega}^\dagger\|_{2\rightarrow2}\langle \bm{h},\bm{t} \rangle$ and the fact that:
\begin{align}
&\mathds{E}|X_{\bm{t}}-X_{\bm{s}}|^2\le \|\bm{\Omega}^\dagger(\bm{t}-\bm{s})\|_2^2\le \|\Omega^\dagger\|_{2\rightarrow2}^2\|\bm{t}-\bm{s}\|_2^2\nonumber\\
&=\mathds{E}|Y_{\bm{t}}-Y_{\bm{s}}|^2~:~\forall \bm{t},\bm{s}\in\mathbb{R}^p.
\end{align}
In the last inequality of (\ref{eq.Gausianwidupper}), we used the change of variable $\bm{w}=\|\bm{\Omega}\|_{2\rightarrow2}^{-1}\bm{z}$. The following Proposition states the connection between statistical dimension and Gaussian width. In the same setting, the work of Amelunxen et al. \cite[Proposition 10.2]{amelunxen2013living} relates the notion of statistical dimension to Gaussian width with a different proof.
\begin{prop}\label{prop.staticgaussian} Let $\mathcal{C}\subseteq\mathbb{R}^n$ be a convex cone. Then the following relation between statistical dimension and Gaussian width holds.
\begin{align}
\omega^2(\mathcal{C}\cap\mathds{B}_1^n)\le\delta(\mathcal{C})\le\omega^2(\mathcal{C}\cap\mathds{B}_1^n)+1.
\end{align}
\end{prop}
\begin{proof}
\begin{align}
&\delta(\mathcal{C})=\mathds{E}\inf_{\bm{z}\in \mathcal{C}^{\circ}}\|\bm{g}-\bm{z}\|_2^2=\mathds{E}\|\bm{g}-\mathcal{P}_{\mathcal{C}^{\circ}}(\bm{g})\|_2^2\le\nonumber\\
&(\mathds{E}\|\bm{g}-\mathcal{P}_{\mathcal{C}^{\circ}}(\bm{g})\|_2)^2+1=(\mathds{E}\inf_{\bm{z}\in\mathcal{C}^{\circ}}\|\bm{g}-\bm{z}\|_2)^2+1\nonumber\\
&=(\mathds{E}\sup_{\bm{z}\in\mathcal{C}\cap\mathds{B}_1^n}\langle \bm{g},\bm{z} \rangle)^2+1=\omega^2(\mathcal{C}\cap\mathds{B}_1^n)+1,
\end{align}
where the first inequality comes from Gaussian poincar\'{e} inequality \cite[Theorem 1.6.4]{gaussianmeasures} for $1$-lipschitz function $\|\mathcal{P}_{\mathcal{C}}(\bm{g})\|_2$, the forth equality stems from convexity of the problem $\underset{{\bm{z}\in\mathcal{C}^{\circ}}}{\inf}\|\bm{g}-\bm{z}\|_2$ and strong duality. On the other hand, we have:
\begin{align}
&\delta(\mathcal{C})=\mathds{E}\inf_{\bm{z}\in \mathcal{C}^{\circ}}\|\bm{g}-\bm{z}\|_2^2\ge (\mathds{E}\inf_{\bm{z}\in \mathcal{C}^{\circ}}\|\bm{g}-\bm{z}\|_2)^2=\nonumber\\
&(\mathds{E}\sup_{\bm{z}\in\mathcal{C}\cap\mathds{B}_1^n}\langle \bm{g},\bm{z} \rangle)^2=\omega^2(\mathcal{C}\cap\mathds{B}_1^n),
\end{align}
where the first inequality uses the Jensen's inequality for concave functions.
\end{proof}
By using Proposition \ref{prop.staticgaussian} and (\ref{eq.Gausianwidupper}), we have:
\begin{align}
&\delta(\mathcal{D}(\|\bm{\Omega}.\|_1,\bm{x}))\le \omega^2(\mathcal{D}(\|\bm{\Omega}.\|_1,\bm{x})\cap\mathds{B}_1^n)+1\le\nonumber\\
&\kappa^2(\bm{\Omega})\omega^2(\mathcal{D}(\|.\|_1,\bm{\Omega x})\cap\mathds{B}_1^p)+1\le \nonumber\\
&\kappa^2(\bm{\Omega})\delta(\mathcal{D}(\|.\|_1,\bm{\Omega x}))+1.
\end{align}
\end{proof}
\section{Proof of Theorem \ref{thm.Omega}}\label{appendix.2}
\begin{proof}
The proof is based on \cite[Theorem 6.1]{amelunxen2013living} with a bit different terminology. By taking $\lambda=\sqrt{8\log\frac{4}{\eta}}\sqrt{n}$ when $m\ge \delta(\mathcal{D}(\|\bm{\Omega} .\|_1,\bm{x}))+\mu+\lambda$ we have:
\begin{align}
&\mathds{P}(\mathcal{D}(\|\bm{\Omega}.\|_1,\bm{x})\cap \mathrm{null}(\bm{A})\neq0)=2h_{m+1}(\mathcal{D}(\|\bm{\Omega}.\|_1,\bm{x}))\le\nonumber\\
&t_{m+1}(\mathcal{D}(\|\bm{\Omega}.\|_1,\bm{x}))\le 4\exp\bigg(\frac{-(\lambda+\mu)^2/8}{\delta(\mathcal{D}(\|\bm{\Omega} .\|_1,\bm{x}))+\lambda+\mu}\bigg)\le\nonumber\\
&4\exp\bigg(\frac{-(\lambda+\mu)^2}{8n}\bigg):=\eta',
\end{align}
where in the above equation, $h_i(\mathcal{D}(\|\bm{\Omega}.\|_1,\bm{x}))$ and $t_i(\mathcal{D}(\|\bm{\Omega}.\|_1,\bm{x}))$ are the $i$th half-tail functional and $i$th tail-functional, respectively which are defined in \cite[Definition 5.8]{amelunxen2013living}.\par
The relation of the tolerance $\eta$ in Theorem \ref{thm.Pfmeasurement} and $\eta'$ is given below:
\begin{align}
&\eta'=4\exp\bigg(-\frac{\mu^2}{8n}-\frac{\mu}{4n}\sqrt{8\log\frac{4}{\eta}}\sqrt{n}+\log\frac{\eta}{4}\bigg)\nonumber\\
&=\eta \exp\bigg(-\frac{\mu^2}{8n}-\frac{\mu}{4\sqrt{n}}\sqrt{8\log\frac{4}{\eta}}\sqrt{n}\bigg)\nonumber\\
&<\eta e^{-\frac{\mu^2}{8n}}.
\end{align}
\end{proof}


\ifCLASSOPTIONcaptionsoff
  \newpage
\fi

\bibliographystyle{ieeetr}
\bibliography{mypaperbibe}
\end{document}